\documentclass[numbers]{sigplanconf}

\usepackage[utf8]{inputenc}
\usepackage{amsfonts}
\usepackage{amsmath}
\usepackage{amsthm}
\usepackage{complexity}
\usepackage{etoolbox}  
\usepackage{mathpartir}
\usepackage{mathtools}
\usepackage{MnSymbol}
\usepackage{refcount}
\usepackage{stmaryrd}
\usepackage[svgnames]{xcolor}
\usepackage{enumitem}
\usepackage{listings}
\usepackage{caption}
\usepackage{flushend}
\usepackage{subcaption}

\usepackage{array}
\newcolumntype{L}[1]{>{\raggedright\let\newline\\\arraybackslash\hspace{0pt}}m{#1}}
\newcolumntype{C}[1]{>{\centering\let\newline\\\arraybackslash\hspace{0pt}}m{#1}}
\newcolumntype{R}[1]{>{\raggedleft\let\newline\\\arraybackslash\hspace{0pt}}m{#1}}

\lstdefinelanguage{scala}{
  morekeywords={abstract,case,catch,class,def,    do,else,extends,false,final,finally,    for,if,implicit,import,match,mixin,    new,null,object,override,package,    private,protected,requires,return,sealed,    super,this,throw,trait,true,try,    type,val,var,while,with,yield},
  otherkeywords={=>,<-,<\%,<:,>:,\#,@},
  sensitive=true,
  morecomment=[l]{//},
  morecomment=[n]{/*}{*/},
  morestring=[b]", 
  morestring=[b]',
  morestring=[b]"""
}

\usepackage{thmtools}
\usepackage{thm-restate}
\usepackage{authblk}

\usepackage{pgf}
\usepackage{tikz}
\usetikzlibrary{calc,arrows,automata}
\usetikzlibrary{positioning,shapes,shadows}
\usetikzlibrary{decorations.pathmorphing,snakes}

\usepackage{titlesec}

\titlespacing*{\section}{0pt}{1.5ex}{1ex}
\titlespacing*{\subsection}{0pt}{1.5ex}{1ex}

\newtheorem{lemma}{Lemma}

\newtheorem{definition}{Definition}

\newtheorem{theorem}{Theorem}  
\newtheorem{example}{Example}
\newtheorem{corollary}{Corollary}
\newtheorem{remark}{Remark}

\newtoggle{long}
\toggletrue{long}

\title{On Verifying Causal Consistency}

\authorinfo{Ahmed Bouajjani}
           {Univ. Paris Diderot \& IUF, France}
           {abou@irif.fr}
          
\authorinfo{Constantin Enea}
           {Univ. Paris Diderot, France}
           {cenea@irif.fr}

\authorinfo{Rachid Guerraoui}
           {EPFL, Switzerland}
           {rachid.guerraoui@epfl.ch}

\authorinfo{Jad Hamza}
           {EPFL-INRIA, Switzerland}
           {jad.hamza@epfl.ch}

\DeclareMathAlphabet{\mathpzc}{OT1}{pzc}{m}{it}

\iftoggle{long}{}{
\pagenumbering{gobble}
}

\begin{document}

\toappear{} 

\newcommand{\set}[1]{\{{#1}\}}
\newcommand{\mset}[1]{\{\{{#1}\}\}}
\newcommand{\todo}[1]{\textcolor{red}{TODO: {#1}}}
\newcommand{\lem}[1]{Lemma~\ref{#1}}
\newcommand{\thm}[1]{Theorem~\ref{#1}}
\newcommand{\coro}[1]{Corollary~\ref{#1}}
\newcommand{\fig}[1]{Figure~\ref{#1}}
\newcommand{\sect}[1]{Section~\ref{#1}}
\newcommand{\apd}[1]{Appendix~\ref{#1}}
\newcommand{\tab}[1]{Table~\ref{#1}}
\newcommand{\ie}{i.e.~}
\newcommand{\resp}{resp.,~}

\colorlet{color1}{Red!50!DarkRed}
\colorlet{color2}{DarkGreen}
\colorlet{color3}{Blue}
\colorlet{color4}{Orange!70!Black}
\colorlet{color5}{DarkViolet}
\colorlet{color6}{Olive}
\colorlet{color7}{ForestGreen}
\colorlet{color8}{FireBrick}
\newcommand{\first}[1]{\textcolor{color1}{#1}}
\newcommand{\second}[1]{\textcolor{color2}{#1}}
\newcommand{\third}[1]{\textcolor{color3}{#1}}

\newcommand{\makeline}[3]{
  \node at (#2,#1) (a) {};
  \node at (#3,#1) (b) {};   
  \draw[-] (a) -- coordinate (d) (b);
}
\newcommand{\placeat}[3]{
    \node[circle,fill=Black,inner sep=0, minimum width=4pt] at (#1,#2) (x) {};
    \node[above=1pt of x] { #3 };
}
\newcommand{\placebt}[3]{
    \node[circle,fill=Black,inner sep=0, minimum width=4pt] at (#1,#2) (x) {};
    \node[below=1pt of x] { #3 };
}
\tikzset{rfstyle/.append style={
    ->,
  thick,
  rounded corners=10 pt,
  inner sep=7pt
}}
\tikzset{costyle/.append style={
    ->,
  thick,
  rounded corners=10 pt,
  inner sep=7pt
}}

\newcommand{\history}{history}
\newcommand{\histories}{histories}
\newcommand{\programorder}{program order}
\newcommand{\operation}{operation}
\newcommand{\specification}{specification}
\newcommand{\dataindependent}{data independent}
\newcommand{\dataindependence}{data independence}
\newcommand{\differentiated}{differentiated}
\newcommand{\readfrom}{read-from}
\newcommand{\linearizable}{linearizable}
\newcommand{\sequence}{sequential poset}
\newcommand{\sequences}{sequential posets}
\newcommand{\execution}{execution}
\newcommand{\event}{action}
\newcommand{\library}{\implementation}
\newcommand{\libraries}{\implementation{s}}
\newcommand{\recurrentreach}{recurrent reachability}
\newcommand{\finitestate}{finite-state}
\newcommand{\nod}{node}
\newcommand{\action}{action}
\newcommand{\site}{site}
\newcommand{\sites}{sites}
\newcommand{\Site}{Site}
\newcommand{\Sites}{Sites}
\newcommand{\siteidentifier}{\site{} identifier}
\newcommand{\observer}{observer}
\newcommand{\registerautomaton}{register automaton}
\newcommand{\registerautomata}{register automata}
\newcommand{\roi}{implementation}
\newcommand{\implementation}{implementation}
\newcommand{\validanswer}{valid answer}
\newcommand{\replicated}{replicated object}

\newcommand{\crita}{{\tt CC}}
\newcommand{\critb}{{\tt CM}}
\newcommand{\critc}{{\tt CCv}}

\newcommand{\wcc}{\crita}
\newcommand{\Wcc}{\crita}
\newcommand{\wcct}{\crita}
\newcommand{\weakcc}{causal consistency}
\newcommand{\weakcct}{causally consistent}

\newcommand{\Scc}{\critb}
\newcommand{\scc}{\critb}
\newcommand{\scct}{\critb}

\newcommand{\ccvt}{\critc}
\newcommand{\Ccv}{\critc}
\newcommand{\ccv}{\critc}
\newcommand{\conv}{causal convergence}
\newcommand{\Conv}{Causal convergence}
\newcommand{\convt}{causally convergent}

\newcommand{\domseq}{sequence of dominoes}
\newcommand{\domseqs}{sequences of dominoes}
\newcommand{\happensbefore}{happens-before}
\newcommand{\writesinto}{writes-into}
\newcommand{\arbitration}{arbitration}
\newcommand{\happenedbefore}{happened-before}

\newcommand{\hist}{h}
\newcommand{\hists}{H}
\newcommand{\op}{o}
\newcommand{\getvar}[1]{{\sf var}({#1})}
\newcommand{\getval}[1]{{\sf value}({#1})}
\newcommand{\wop}{w}
\newcommand{\rop}{r}
\newcommand{\ops}{O}
\newcommand{\getlabel}{\ell}

\newcommand{\po}{{\sf PO}}
\renewcommand{\co}{\mathit{co}}
\newcommand{\hb}[1]{{\sf HB}_{#1}}
\newcommand{\cf}{{\sf CF}}
\newcommand{\trueco}{\mathit{co}}
\newcommand{\propco}{{\sf CO}}
\newcommand{\ltpropco}{\ltrel{\propco}}
\newcommand{\rf}{{\sf RF}}
\newcommand{\locof}[1]{\rho_{#1}}
\newcommand{\loc}{\locof{\op}}
\newcommand{\smallerF}{\xrightarrow{\renamings}}
\newcommand{\smallerFData}{\xrightarrow{\datarens}}
\newcommand{\wi}{\mathit{wi}}
\newcommand{\theirarb}{\mathit{arb}}

\newcommand{\rel}{\mathcal{R}}
\newcommand{\ltrel}[1]{<_{{#1}}}
\newcommand{\leqrel}[1]{\leq_{{#1}}}
\newcommand{\gtrel}[1]{>_{{#1}}}
\newcommand{\ltr}{\ltrel{\rel}}
\newcommand{\ggtr}{\gtrel{\rel}}
\newcommand{\ltpo}{\ltrel{\po}}
\newcommand{\leqpo}{\leqrel{\po}}
\newcommand{\ltco}{\ltrel{\co}}
\newcommand{\leqco}{\leqrel{\co}}
\newcommand{\ltrf}{\ltrel{\rf}}
\newcommand{\ltcf}{\ltrel{\cf}}
\newcommand{\comp}{\circ}
\newcommand{\project}[2]{{#1}_{|{#2}}}
\newcommand{\proj}[2]{\project{#1}{#2}}
\newcommand{\projrel}[2]{\project{#1}{#2}}
\newcommand{\chan}{{\mathit{ch}}}

\newcommand{\newco}{\co_2} 
\newcommand{\newloc}{\loc^2}

\newcommand{\Meth}{\mathbb{M}}
\newcommand{\Domain}{\mathbb{D}}
\newcommand{\DDomain}{\mathbb{D}_\top}
\newcommand{\Nats}{\mathbb{N}}
\newcommand{\Tid}{{\sf PId}}
\newcommand{\Var}{\mathbb{X}}
\newcommand{\theclass}{\mathcal{C}}
\newcommand{\localmem}{\mathbb{L}}

\newcommand{\spec}{S}

\newcommand{\pos}{\rho}
\newcommand{\meth}{m}
\newcommand{\dat}{d}
\newcommand{\dats}{D}
\newcommand{\act}{a}
\newcommand{\acts}{B}
\newcommand{\bijection}{f}
\newcommand{\exec}{e}
\newcommand{\myexec}{E}
\newcommand{\tid}{p}
\newcommand{\argv}{{\mathit{arg}}}
\newcommand{\rv}{{\mathit{rv}}}
\newcommand{\var}{x}
\newcommand{\val}{v}
\newcommand{\obs}{\mathcal{M}}
\newcommand{\clink}{{\sf CausalLink}}
\newcommand{\obscc}{\mathcal{M}_{\crita}}

\newcommand{\ncycle}{n}
\newcommand{\nthreads}{k}
\newcommand{\execlength}{l}
\newcommand{\nres}{{\mu}}
\newcommand{\mres}{{\nu}}
\newcommand{\numa}{A}
\newcommand{\numb}{B}
\newcommand{\numc}{C}

\renewcommand{\cc}{\cdot}

\newcommand{\unit}{\bot}

\newcommand{\register}{{\sf register}}
\newcommand{\mvs}{{\sf multi-value store}}
\newcommand{\kvs}{{\sf read/write\ memory}}
\newcommand{\kvsimp}{{key-value store}}
\newcommand{\specreg}{\spec_{\sf reg}}
\newcommand{\speckvs}{\spec_{\sf RW}}
\newcommand{\specmvs}{\spec_{\sf MVS}}
\newcommand{\writemeth}{{\tt wr}}
\newcommand{\readmeth}{{\tt rd}}
\newcommand{\mkwrite}[2]{\writemeth({#1},{#2})}
\newcommand{\mkread}[2]{\mkaction{\readmeth({#1})}{#2}}

\newcommand{\vass}{{\sf VASS}}

\newcommand{\linable}{\sqsubseteq} 

\newcommand{\lib}{\mathcal{I}}
\newcommand{\alet}{a}

\newcommand{\default}{\top}
\newcommand{\responsec}[2]{\mkaction{#1}{#2}}
\newcommand{\mCheck}{{\sf Check}}
\newcommand{\qend}{{\sf end}}
\newcommand{\qsent}{{\sf sent}}
\newcommand{\qrecA}{{\sf recA}}
\newcommand{\qrec}{{\sf recBoth}}
\newcommand{\qchecked}{{\sf checked}}

\newcommand{\checkmate}{\state_{\sf check}}
\newcommand{\loststate}{\state_{\sf lost}}

\newcommand{\msgs}{{\sf Msg}}
\newcommand{\Msg}{\msgs}
\newcommand{\msg}{{\sf msg}}

\newcommand{\boolt}{{\sf T}}
\newcommand{\boolf}{{\sf F}}

\newcommand{\autL}{L}
\newcommand{\wu}{u}
\newcommand{\wv}{v}
\newcommand{\alp}{\Sigma}
\newcommand{\shuffle}[2]{{#1}\|{#2}}
\newcommand{\pcp}{{\sf PCP}}
\newcommand{\morph}{h}
\newcommand{\morphu}{\morph_u}
\newcommand{\morphv}{\morph_v}

\newcommand{\alppcp}{\Sigma_\pcp}
\newcommand{\apcp}{a}
\newcommand{\bpcp}{b}
\newcommand{\npcp}{n}
\newcommand{\kpcp}{k}
\newcommand{\instpcp}{P}

\newcommand{\emptyseq}{\varepsilon}
\newcommand{\alpu}{\Gamma_u}
\newcommand{\alpinu}{\Sigma_u}
\newcommand{\alpv}{\Gamma_v}
\newcommand{\alpinv}{\Sigma_v}
\newcommand{\au}{a_u}
\newcommand{\av}{a_v}
\newcommand{\bu}{b_u}
\newcommand{\bv}{b_v}
\newcommand{\xu}{\mathpzc{s}_u}
\newcommand{\xv}{\mathpzc{s}_v}
\newcommand{\sep}{\#}
\newcommand{\sigu}{u}
\newcommand{\sigv}{v}
\newcommand{\subu}[1]{x_{#1}}
\newcommand{\subv}[1]{y_{#1}}
\newcommand{\alpwhole}{\Sigma}
\newcommand{\alpinpcp}{\Sigma}
\newcommand{\bbbu}{S_u}
\newcommand{\bbbv}{S_v}

\newcommand{\ww}{w}
\newcommand{\remlett}[2]{{#1}\setminus{#2}}
\newcommand{\remlang}[2]{\remlett{#1}{#2}}

\newcommand{\uify}[1]{\morphu({#1})}
\newcommand{\vify}[1]{\morphu({#1})}
\newcommand{\remsubu}[1]{\morph({#1})}
\newcommand{\remsubv}[1]{\morph({#1})}
\newcommand{\mkaction}[2]{{#1} \ret {#2}}
\newcommand{\mkpartialaction}[2]{{#1}({#2})}
\newcommand{\mktrueaction}[3]{{#1}({#2})\ifstrequal{#3}{\unit}{}{ \ret {#3}}}

\newcommand{\ret}{{\,\triangleright\,}}
\newcommand{\states}{Q}
\newcommand{\vistate}{\state_{\sf init}}
\newcommand{\state}{q}

\newcommand{\siteIds}{\Tid}
\newcommand{\siteid}{\tid}
\newcommand{\parts}[1]{{\mathcal P}({#1})}

\newcommand{\transitions}{\Delta}
\newcommand{\istate}{q_0}
\newcommand{\blet}{b}
\newcommand{\fstate}{q_f}
\newcommand{\ustate}{q_u}
\newcommand{\vstate}{q_v}

\newcommand{\qmap}{\mu}
\newcommand{\chanmap}{{\sf chan}}
\newcommand{\cfgc}{{\sf cfg}}
\newcommand{\trans}{\rightarrow}
\newcommand{\ltrans}[1]{\xrightarrow{#1}}
\newcommand{\sendcToLab}[2]{\sendc({#1},{#2})}
\newcommand{\receivecFromLab}[2]{\receivec({#1},{#2})}
\newcommand{\msgword}{{\sf w}}
\newcommand{\sendc}{!}
\newcommand{\receivec}{?}
\newcommand{\Bag}[1]{{\sf Bag}({#1})}
\newcommand{\bag}{C}

\newcommand{\wordu}{u}
\newcommand{\wordv}{v}

\newcommand{\oper}[1]{O_{#1}}
\newcommand{\orset}{{{\sf OR}\mbox{-}{\sf S}}}
\newcommand{\mklookup}[1]{{\tt rem}_{#1}}

\newcommand{\getexecs}[1]{\llbracket {#1} \rrbracket}

\newcommand{\labeledsystems}{labeled transition system}
\newcommand{\LabeledSystems}{Labeled Transition Systems}
\newcommand{\istates}{I}
\newcommand{\fstates}{F}
\newcommand{\runlts}{\mathit{run}}

\newcommand{\applyrenaming}[2]{{#1}[{#2}]}
\newcommand{\getdiff}[1]{{#1}_{\neq}}
\newcommand{\getcomplete}[1]{{#1}'}
\newcommand{\complete}{complete}
\newcommand{\Complete}{Complete}
\newcommand{\renaming}{f}
\newcommand{\renamings}{F}
\newcommand{\datarenaming}{data-renaming}
\newcommand{\datarens}{\renamings_{\sf Data}}

\newcommand{\wccfor}[1]{{\crita}({#1})}
\newcommand{\sccfor}[1]{{\critb}({#1})}
\newcommand{\ccvfor}[1]{{\critc}({#1})}

\newcommand{\visible}[1]{\vis^{-1}({#1})}
\newcommand{\vposet}[1]{(\visible{#1},\locof{#1})}
\newcommand{\lposet}[1]{(\visible{#1},\locof{#1},\getmethod)}
\newcommand{\iposet}{(\ops,\arb,\getmethod)}
\newcommand{\posetm}[1]{li_\Meth[{#1}]}
\newcommand{\iso}{{f}}
\newcommand{\prefix}{\preceq}

\newcommand{\assign}[2]{{#1} \coloneqq {#2}}
\newcommand{\equality}[2]{{#1} == {#2}}
\newcommand{\reg}{{\sf reg}}
\newcommand{\regvar}{\reg_{\var}}
\newcommand{\regsite}{\reg_{\siteid}}

\newcommand{\firstsite}{\textcolor{color1}{\siteid_1}}
\newcommand{\secondsite}{\textcolor{color2}{\siteid_2}}
\newcommand{\thirdsite}{\textcolor{color3}{\siteid_3}}
\newcommand{\firstdat}{\textcolor{color4}{1}}
\newcommand{\seconddat}{\textcolor{color5}{2}}
\newcommand{\thirddat}{\textcolor{color6}{3}}
\newcommand{\fourthdat}{\textcolor{color7}{4}}
\newcommand{\fifthdat}{\textcolor{color8}{5}}
\newcommand{\errorstate}{q_{err}}
\newcommand{\thedat}{\dat_0}
\newcommand{\thealphabet}[1]{{\sf Act}_{#1}}

\newcommand{\buchi}{\mathcal{A}}
\newcommand{\finitestates}{B}
\newcommand{\proc}{p}

\newcommand{\diffblock}[1]{#1}
\newcommand{\seqposet}{\rho}
\newcommand{\causalpast}[1]{{\sf CausalHist}({#1})}
\newcommand{\poback}[1]{{\sf POPast}({#1})}
\newcommand{\projectrv}[2]{{#1}\{{#2}\}}
\newcommand{\weaker}{\preceq}
\newcommand{\autLL}{\autL'}
\newcommand{\causalarb}[1]{{\sf CausalArb}({#1})}

\newcommand{\finproc}{\proc_f}
\newcommand{\finishedval}{{\sf NVA}}
\newcommand{\statusvar}{{\tt status}}
\newcommand{\statevar}{{\tt state}}
\newcommand{\procu}{\proc_u}
\newcommand{\procv}{\proc_v}
\newcommand{\overval}{{\sf over}}
\newcommand{\puover}{\overval}
\newcommand{\pvover}{\overval}
\newcommand{\mkget}[2]{{\tt uniq\_rd}({#1}) \ret {#2}}
\newcommand{\fu}{\tid_{gt}}
\newcommand{\fv}{\tid_{lt}}
\newcommand{\fneq}{\tid_{\neq}}

\newcommand{\corr}{\leftrightarrow}
\newcommand{\addmeth}{{\tt add}}
\newcommand{\remmeth}{{\tt rem}}
\newcommand{\containsmeth}{{\tt contains}}
\newcommand{\mkadd}[1]{\addmeth({#1})}
\newcommand{\mkrem}[1]{\remmeth({#1})}
\newcommand{\inset}{\top}
\newcommand{\notinset}{\bot}
\newcommand{\mkcontains}[1]{\mkaction{\containsmeth({#1})}{\inset}}
\newcommand{\mknotcontains}[1]{\mkaction{\containsmeth({#1})}{\notinset}}

\newcommand{\letu}{L}
\newcommand{\letv}{L_v}
\newcommand{\tickuf}{T_u}
\newcommand{\tickus}{T_u'}
\newcommand{\tickvf}{T_v}
\newcommand{\tickvs}{T_v'}
\newcommand{\fauxval}{0}
\newcommand{\vraival}{1}
\newcommand{\pufin}{F_u}
\newcommand{\pvfin}{F_v}
\newcommand{\placeholder}{{p_{BU}}}
\newcommand{\ticking}[1]{{p_{T_{#1}}}}
\newcommand{\makechecker}[1]{ch_{#1}}

\newcommand{\cmd}{{\tt cmd}}
\newcommand{\boolexpr}{{\tt bool\_expr}}
\newcommand{\intexpr}{{\tt value\_expr}}
\newcommand{\sendcmd}[2]{{\tt send}({#1},{#2})}
\newcommand{\returncmd}[1]{{\tt return}({#1})}
\newcommand{\internal}{T}
\newcommand{\internaltransition}{internal transition}
\newcommand{\piece}{{\tt Code}}
\newcommand{\getcode}{{\sf code}}
\newcommand{\Args}{{\sf Dom}}
\newcommand{\Rvs}{{\sf Img}}
\newcommand{\itr}{t}
\newcommand{\other}[1]{{#1}'}
\newcommand{\extra}{\tid_E}
\newcommand{\choicevar}{\mathit{Ch}}
\newcommand{\whilec}{{\tt while}}
\newcommand{\ifc}{{\tt if}}
\newcommand{\elsec}{{\tt else}}
\newcommand{\eval}{{\sf eval}}
\newcommand{\update}{{\sf update}}
\newcommand{\causaldep}[1]{{\sf CausalPast}({#1})}

\newcommand{\bpcocycle}{there is a cycle in $\po \cup \rf$ (in $\propco$)}
\newcommand{\bprfi}{
  there is a $\mkread{\var}{0}$ operation $\rop$, and an operation $\wop$
  such that $\wop \ltpropco \rop$ and $\getvar{\wop} = \getvar{\rop}$}
\newcommand{\bprf}{
  there is a $\mkread{\var}{\val}$ operation $\rop$ such that $\val \neq 0$,
  and there is no $\wop$ operation with $\wop \ltrf \rop$}
\newcommand{\bpdisagree}{there is a cycle in $\cf \cup \propco$}
\newcommand{\bpchangerfi}{
  there is a $\mkread{\var}{0}$ operation $\rop$, and an operation $\wop$
  such that $\wop \ltrel{\hb{\op}} \rop$ and $\getvar{\wop} = \getvar{\rop}$,
  for some $\op$, with $\rop \leqpo \op$
}
\newcommand{\bpchange}{there is a cycle in $\hb{\op}$ for 
some $\op \in \ops$}
\newcommand{\bpco}{
  there exist write operations $\wop_1,\wop_2$ and a read
  operation $\rop_1$ in $\ops$ such that
  $\wop_1 \ltpropco \wop_2 \ltpropco \rop_1$,  
  $\wop_1 \ltrf \rop_1$, and 
  $\getvar{\wop_1} = \getvar{\wop_2}$}
  
\newcommand{\axpoco}{{\sf AxCausal}}
\newcommand{\axcoarb}{{\sf AxArb}}
\newcommand{\axwcc}{{\sf AxCausalValue}}
\newcommand{\axscc}{{\sf AxCausalSeq}}
\newcommand{\axccv}{{\sf AxCausalArb}}
\newcommand{\crit}{\mathit{Crit}}

\newcommand{\bpbpcyclicco}{{\sf CyclicCO}}
\newcommand{\bpbpcycliccf}{{\sf CyclicCF}}
\newcommand{\bpbpcyclichb}{{\sf CyclicHB}}
\newcommand{\bpbprfi}{{\sf WriteCOInitRead}}
\newcommand{\bpbpchangerfi}{{\sf WriteHBInitRead}}
\newcommand{\bpbpco}{{\sf WriteCORead}}
\newcommand{\bpbprf}{{\sf ThinAirRead}}

\newcommand{\npcomplete}{\NP-complete}
\newcommand{\conflict}{conflict}
\newcommand{\sat}{{\sf SAT}}
\newcommand{\sx}{x} 
\newcommand{\scl}{C}
\newcommand{\form}{\phi}
\newcommand{\inits}{f_i}

\newcommand{\getpos}[1]{Pos({#1})}
\newcommand{\getneg}[1]{Neg({#1})}
\newcommand{\gettrue}[1]{\tid_{false}^i}
\newcommand{\getfalse}[1]{\tid_{true}^i}
\newcommand{\npvar}{y}
\newcommand{\peval}{\tid_{eval}}
\newcommand{\encoding}[2]{H_{({#1},{#2})}}
\newcommand{\finaling}{\#}
\newcommand{\savior}[1]{p_{S_{#1}}}
\newcommand{\saver}{\mathit{saver}}
\newcommand{\pag}{p_f}

\newcommand{\causallyconsistent}{causally consistent}
\newcommand{\cct}{causally consistent}
\newcommand{\inc}[2]{inc({#1},{#2})}
\newcommand{\dec}[2]{dec({#1},{#2})}
\newcommand{\axcausal}{{\sc Causal}}
\newcommand{\axcausalvis}{{\sc CausalVis}}
\newcommand{\axrval}{{\sc RVal}}
\newcommand{\eventual}{\textsc{Eventual}}
\newcommand{\interwf}{\textsc{GI\textsc{pf}}}
\newcommand{\wf}{prefix-founded}
\newcommand{\wccpt}{causally+ consistent}
\newcommand{\wccp}{causal+ consistency}
\newcommand{\ccvpt}{causally+ convergent}
\newcommand{\ccvp}{causal+ convergence}
\newcommand{\arb}{{go}}
\newcommand{\otherr}[1]{{#1}''}

\newcommand{\cm}{causal memory}
\newcommand{\sct}{strongly causally consistent}
\newcommand{\obsreg}{\mathcal{M}_{Reg}}
\newcommand{\obsregp}{\mathcal{M}_{Reg+}}
\newcommand{\obsccp}{\mathcal{M}_{\sf RW+}}
\newcommand{\getmethod}{\overline{\ell}}
\newcommand{\cah}{causal arbitrated history}
\newcommand{\falsifier}{checker}

\newcommand{\mEndA}{{\sf End_A}}
\newcommand{\mEndB}{{\sf End_B}}
\newcommand{\msgFinished}{{\sf Finished}}
\newcommand{\msgFinishedA}{{\sf Finished_A}}
\newcommand{\msgFinishedB}{{\sf Finished_B}}
\newcommand{\mChoice}{{\sf Choice}}

\maketitle

{\bf Abstract}. 
Causal consistency is one of the most adopted consistency criteria for distributed implementations of data structures. 
It ensures that operations are executed at all sites according to their causal precedence.
We address the issue of verifying automatically whether the executions of an implementation of a data structure are causally consistent.
We consider two problems: (1) checking whether {\em one} single execution is causally consistent, which is relevant for developing testing and bug finding algorithms, and (2) verifying whether {\em all} the executions of an implementation are causally consistent.  

We show that the first problem is NP-complete. This holds even for the read-write memory abstraction, which is a building block of many modern distributed systems. Indeed, such systems often store data in key-value stores, which are instances of the read-write memory abstraction.
Moreover, we prove that, surprisingly, the second problem is {\em undecidable}, and again this holds even for the read-write memory abstraction. 
However, we show that for the read-write memory abstraction, these negative results can be circumvented if the implementations are {\em data independent}, i.e., their behaviors do not depend on the data values that are written or read at each moment, which is a realistic assumption.

We prove that for data independent implementations, the problem of checking the correctness of a single execution w.r.t. the read-write memory abstraction is polynomial time.
Furthermore, we show that for such implementations the set of non-causally consistent executions can be represented by means of a finite number of {\em register automata}. Using these  machines as observers (in parallel with the implementation) allows to reduce polynomially the problem of checking causal consistency to a state reachability problem. This reduction holds regardless of the class of programs used for the implementation, of the number of read-write variables, and of the used data domain. It allows leveraging existing techniques for assertion/reachability checking to causal consistency verification. Moreover, for a significant class of implementations, we derive from this reduction the decidability of verifying causal consistency w.r.t. the read-write memory abstraction.



\setlist{nosep}

\section{Introduction}

Causal consistency~\citep{Lamport:1978:TCO:359545.359563} (CC for short)
is one of the oldest and most widely spread correctness criterion
for distributed systems. For a distributed system composed of several sites connected through a network where
each site executes some set of operations, if an operation $\op_1$ affects another
operation $\op_2$ ($\op_2$ causally depends on $\op_1$), 
causal consistency ensures that all sites must 
execute these operations in that order.
There exist many efficient implementations satisfying this criterion, e.g.,
~\citep{DBLP:conf/cloud/DuI0Z14,DBLP:conf/cloud/DuE0Z13,Bailis:2013:BCC:2463676.2465279,Jimenez:2008:PAI:1316083.1316267,Lloyd:2011:DSE:2043556.2043593},
contrary to strong consistency (linearizability) which cannot be ensured in the presence of network partitions and while the system
remains available ~\citep{DBLP:journals/sigact/GilbertL02,Fischer:1985:IDC:3149.214121} (the \sites{}  
answer to clients' requests without delay).

However, developing distributed implementations satisfying causal consistency poses
many challenges: Implementations may involve a large number 
of sites communicating through unbounded\footnote{Throughout the paper, {\em unbounded} means finite but arbitrarily large.} communications channels. Roughly speaking, causal consistency can be ensured if each operation (issued by some site) is broadcast to the other sites together with its whole ``causal past'' (the other operations that affect the one being broadcast). But this is not feasible in practice, and various optimizations have been proposed that involve for instance the use of vector clocks~\citep{fidge1987timestamps,Mattern88virtualtime}. Defining and implementing such optimizations is generally very delicate and error prone. 
Therefore, it is appealing to consider formal methods to help developers write correct implementations. At different stages of the development, both testing and verification techniques are needed either for detecting bugs or for establishing correctness w.r.t abstract specifications. We study in this paper two fundamental problems in this context: (1) checking whether one given execution of an implementation is causally consistent, a problem that is relevant for the design of testing algorithms, and (2) the problem of verifying whether all the executions of an implementation are causally consistent.

First, we prove that checking causal consistency for a single execution is NP-hard in general. We prove in fact that this problem is NP-complete for the read-write memory abstraction (RWM for short), which is at the basis of many distributed data structures used in practice. 

Moreover, we prove that the problem of verifying causal consistency of an implementation is undecidable in general. We prove this fact in two different ways. First, we prove that for regular specifications  (i.e., definable using finite-state automata), this problem is undecidable even for finite-state implementations with two sites communicating through bounded-size channels. Furthermore, we prove that even for the particular case of the RWM specification, the problem is undecidable in general. (The proof in this case is technically more complex and requires the use of implementations with more than two sites.) 

This undecidability result might be surprising, since it is known
that linearizability (stronger than CC)~\citep{journals/toplas/HerlihyW90} and eventual consistency (weaker than CC)~\citep{DBLP:conf/sosp/TerryTPDSH95} are decidable
to verify in that same setting~\citep{journals/iandc/AlurMP00,DBLP:conf/popl/BouajjaniEH14,netys-lin}.
This result reveals an interesting aspect in the definition of causal consistency. Intuitively, two key properties of causal consistency are that (1) 
it requires that the order between operations issued by the same \site{} to be preserved globally at all the sites, and that 
(2) it allows an operation $\op_1$ which happened arbitrarily sooner than an operation $\op_2$ to be executed after $\op_2$ (if $\op_1$ and $\op_2$ are not causally related).
Those are the essential ingredients that are used in the undecidability proofs (that are based on encodings of the Post Correspondence Problem). In comparison, linearizability does not satisfy (2) because for a fixed number of sites/threads, the reordering between operations is bounded (since only operations which overlap in time can be reordered), while eventual consistency does not satisfy (1).

Our NP-hardness and undecidability results show that reasoning about causal consistency is intrinsically hard in general. However, by focusing on the case of the RWM abstraction, and by considering commonly used objects that are instances of this abstraction, e.g., key-value stores, one can observe that their implementations are 
typically {\em data independent}~\citep{conf/popl/Wolper86,conf/tacas/AbdullaHHJR13}. This means that the way these implementations handle data with read and write instructions is insensitive to the actual data values that are read or written. We prove that reasoning about causal consistency w.r.t. the RWM abstraction becomes tractable under the natural assumption of data independence. More precisely, we prove that checking causal consistency for a single computation is polynomial in this case, and that verifying causal consistency of an implementation is polynomially reducible to a state reachability problem, the latter being decidable for a significant class of implementations. Let us explain how we achieve that.

In fact, data independence implies that it is sufficient to consider executions where each value is written at most once; let us call such executions {\em differentiated}~(see, e.g., \citep{conf/tacas/AbdullaHHJR13}). 
The key step toward the results mentioned above is a characterization of the set of all differentiated executions that violate causal consistency w.r.t. the RWM. 
This characterization is based on the notion of a \emph{bad pattern} that can be seen as a set of operations occurring (within an execution) in some particular order corresponding to a causal consistency violation.  We express our bad patterns using appropriately defined conflict/dependency relations between operations along executions. We show that there is a {\em finite number} of bad patterns such that an execution is consistent w.r.t. the RWM abstraction {\em if and only if} the execution does not contain any of these patterns. 

In this characterization, the fact that we consider only differentiated executions is crucial. The reason is that all relations used to express bad patterns include the read-from relation that associates with each read operation the write operation that provides its value. This relation is uniquely defined for differentiated executions, while for arbitrary  executions where writes are not unique, reads can take their values from an arbitrarily large number of writes. This is actually the source of complexity and undecidability in the non-data independent case. 

Then, we exploit this characterization in two ways. First, we show that for a given execution, checking that it contains a bad pattern can be done in polynomial time, which constitutes an important gain in complexity w.r.t. to the general algorithm that does not exploit data independence (precisely because the latter needs to consider all possible read-from relations in the given execution.)

Furthermore, we show that for each bad pattern, it is possible to construct effectively an {\em observer} (which is a state-machine of some kind) that is able, when running in parallel with an implementation, to detect all the executions containing the bad pattern. A crucial point is to show that these observers are in a class of state-machines that has ``good'' decision properties. (Basically, it is important that checking whether they detect a violation is decidable for a significant class of implementations.) We show that the observers corresponding to the bad patterns we identified can  be defined as {\em register automata} \citep{DBLP:conf/concur/BouyerPT01}, i.e., finite-state state machines supplied with a finite number of registers that store data over a potentially infinite domain (such as integers, strings, etc.) but on which the only allowed operation is checking equality. An important feature of these automata is that their state reachability problem can be reduced to the one for (plain) finite-state machines. The construction of the observers is actually independent from the type of programs used for the implementation, leading to a semantically sound and complete reduction to a state reachability problem (regardless of the decidability issue) even when the implementation is deployed 
over an unbounded number of sites, has an unbounded number of variables (keys) storing data over an unbounded domain. 

Our reduction enables the use of any reachability analysis or assertion checking tool 
for the verification  of causal consistency. 
Moreover, for an important class of implementations, this reduction leads to decidability and provides a verification algorithm for causal consistency w.r.t. the RWM abstraction.
We consider implementations consisting of a finite number of  state 
machines communicating through a network (by message passing). Each machine has a finite number of finite-domain 
(control) variables with unrestricted use, in addition to a finite number of 
{\em data variables} that are used only to store and move data, and on which 
no conditional tests can be applied. 
Moreover, we do not make any assumption on the 
network: the machines communicate through unbounded unordered channels, which is the usual setting 
in large-scale distributed networks. 
(Implementations can apply ordering protocols on top of this most permissive model.) 

Implementations in the class we consider have an infinite number of 
configurations (global states) due to (1) the unboundedness of the data domain, and (2) the unboundedness of the communication 
channels. First, we show that due to 
data independence and the special form of the observers detecting bad patterns, 
proving causal consistency for any given implementation in this class (with 
any data domain) reduces to proving its causal consistency for a {\em 
bounded} data domain (with precisely 5 elements). This crucial fact allows to 
get rid of the first source of unboundedness in the configuration space. The 
second source of unboundedness is handled using counters: we prove that 
checking causal consistency in this case can be reduced to the state 
reachability problem in Vector Addition Systems with States (equivalent to 
unbounded Petri Nets), and conversely. This implies that verifying causal consistency w.r.t. 
the RWM (for this class of implementations) is EXPSPACE-complete.

It is important to notice that causal consistency has different meanings depending on the context and the targeted applications.
Several efforts have been made recently for formalizing various notions of causal consistency (e.g., \citep{Gotsman05,principles-of-eventual-consistency, Perrin:2016:CCB:2851141.2851170,jadphd,raynal1995causal}). In this paper we consider three important variants. The variant called simply causal consistency (abbreviated as \wcc) allows non-causally dependent operations to be executed in different orders by different sites, and decisions about these orders to be revised by each site. This models mechanisms for solving the conflict between non-causally dependent operations where each site speculates on an order between such operations and possibly roll-backs some of them if needed later in the execution, e.g.,~\cite{DBLP:conf/sosp/TerryTPDSH95,C-Praxis,IceCube,Telex}.
We also consider two stronger notions, namely {\em causal memory} (\scc) \citep{Ahamad94causalmemory,Perrin:2016:CCB:2851141.2851170}, and {\em causal convergence} (\ccv) \citep{Gotsman05,principles-of-eventual-consistency, Perrin:2016:CCB:2851141.2851170}. The latter assumes that there is a total order between non-causally dependent operations and each site can execute operations only in that order (when it sees them). Therefore, a site is not allowed to revise its ordering of non-causally dependent operations, and all sites execute in the same order the operations that are visible to them. This notion is used in a variety of systems~\cite{Bailis:2013:BCC:2463676.2465279,Zawirski:2015:WFR:2814576.2814733,Terry:2013:CSL:2517349.2522731,Terry:1994:SGW:645792.668302,Ladin:1992:PHA:138873.138877,DBLP:conf/cloud/DuI0Z14} because it also implies a strong variant of convergence, i.e., that every two sites that receive the same set of updates execute them in the same order. As for \scc, a site is allowed to diverge from another site on the ordering of non-causally dependent operations, but is not allowed to revise its ordering later on. \scc\, and \ccv\, are actually incomparable \citep{Perrin:2016:CCB:2851141.2851170}. 

All the contributions we have described above in this section hold for the \wcc\, criterion. In addition, concerning \scc\, and \ccv, we prove that (1) the NP-hardness and undecidability results hold, (2) a characterization by means of a finite number of bad patterns is possible, and (3) checking consistency for a single execution is polynomial time.
   
To summarize, this paper establishes the first complexity and (un)decidability results concerning the verification of causal consistency:
\begin{itemize}
\item NP-hardness of the problems of checking \wcc, \scc, and \ccv\, for a single execution 
(\sect{sec:npnpcomplete}).
\item Undecidability of the problems of verifying \wcc, \scc, and \ccv\, for regular specifications, and actually even for the RWM specification 
(\sect{sec:kvsundec}).
\item
A polynomial-time procedure for verifying that a single execution of a data independent implementation is \wcc, \scc, and \ccv\, w.r.t. RWM
(\sect{sec:poly}).
\item 
Decidability and complexity for the verification of \wcc\, w.r.t. the RWM for a significant class of data independent implementations (\sect{sec:decidability}).
\end{itemize}

The complexity and decidability results obtained for the RWM (under the assumption of data independence) are based on two key contributions that provide a deep insight on the problem of verifying causal consistency, and open the door to efficient automated testing/verification techniques:
\begin{itemize}
\item 
A characterization as a finite set of ``bad patterns'' of the set of violations to \wcc, \scc, and \ccv\, w.r.t. the RWM, under the assumption of data independence
(\sect{sec:dataind}).
\item 
A polynomial reduction of the problem of verifying that a data independent implementation is \wcc\, w.r.t. to the RWM to a state reachability (or dually to an invariant checking) problem
(\sect{sec:safety}).
\end{itemize}  

\vspace{1ex}
\iftoggle{long}{}{
Some proofs are deferred to the long version~\citep{arxixpopl17}.
}
 \section{Notations}

\subsection{Sets, Multisets, Relations}

Given a set $\ops$ and a relation $\rel \subseteq \ops \times \ops$,  
we denote by $\op_1 \ltr \op_2$ the fact that 
$(\op_1,\op_2) \in \rel$. We denote by 
$\op_1 \leqrel{\rel} \op_2$ the fact that $\op_1 \ltr \op_2$
or $\op_1 = \op_2$.
We denote by $\rel^+$ the \emph{transitive closure} of $\rel$, which is
the composition of one or more copies of $\rel$.

Let $\ops'$ be a subset of $\ops$. Then $\project{\rel}{\ops'}$ is
the relation $\rel$ projected on the set $\ops'$, that is
$\set{(\op_1,\op_2) \in \rel\ |\ \op_1, \op_2 \in \ops'}$.
The set $\ops' \subseteq \ops$ is said to be 
\emph{downward-closed} (with respect to relation $\rel$) if 
$\forall \op_1,\op_2$, if $\op_2 \in \ops'$ and $\op_1 \ltr \op_2$,
then $\op_1 \in \ops'$ as well. 
We define \emph{upward-closed} similarly.

\subsection{Labeled Posets}

A relation $< \,\, \subseteq \ops \times \ops$ is a 
\emph{strict partial order} if it is
transitive and irreflexive.
A \emph{poset} is a pair $(\ops,<)$ where 
$<$ is a strict partial order over $\ops$.
Note here that we use the strict version of posets, and 
not the ones where the underlying partial order is \emph{weak}, \ie
reflexive, antisymmetric, and transitive.

Given a set $\Sigma$, a \emph{$\Sigma$ labeled poset} $\pos$
is a tuple $(\ops,<,\getlabel)$ where $(\ops,<)$ is a poset and 
$\getlabel: \ops \rightarrow \Sigma$ is the labeling function.

We say that $\pos'$ is a \emph{prefix} of $\pos$ if there exists a downward 
closed set $A \subseteq O$ (with respect to relation $<$) such that
$\pos' = (A,<,\getlabel)$.
A (\resp labeled) \emph{\sequence} (sequence for short) 
is a (\resp labeled) poset where the relation $<$ is 
a strict total order.
We denote by $\exec \cc \exec'$ the concatenation of \sequence{s}.

\section{Replicated Objects}

We define an abstract model for the class of distributed objects called \emph{\replicated s}~\citep{birman1985replication}, where the object 
state is replicated at different sites in a network, called also \emph{processes}, and updates or queries to the object
can be submitted to any of these sites. This model reflects the view that a client
has on an execution of this object, i.e., a set of operations with their inputs and outputs where
every two operations submitted to the same site are ordered. It abstracts away the implementation internals 
like the messages exchanged by the sites in order to coordinate about the object state. 
Such a partially ordered set of operations is called a \emph{history}. 
The correctness (consistency) of a replicated object is defined with respect to a \emph{specification}
that captures the behaviors of that object in the context of sequential programs.

\subsection{Histories}

A \replicated{} implements a programming interface (API) defined by a set of \emph{methods} $\Meth$ with input or output values from  
a domain $\Domain$. 

For instance, in the case of the \kvs{}, the set of methods $\Meth$ is $\set{\writemeth,\readmeth}$ for writing or reading a variable.
Also, given a set of \emph{variables} $\Var$, the domain 
$\Domain$ is defined as
$(\Var \times \Nats) \uplus \Var \uplus \Nats \uplus \set{\unit}$. Write operations take as input a variable in $\Var$ and a value in $\Nats$
and return $\unit$ while read operations take as input a variable in $\Var$ and return a value in $\Nats$. The return value $\unit$ is often omitted for
better readability.

A \emph{\history} $\hist = (\ops,\po,\getlabel)$ is a 
poset labeled by $\Meth \times \Domain \times \Domain$,
where:
\begin{itemize}
\item $\ops$ is a set of operation identifiers, or simply \emph{operations}, \item $\po$ is a union of total orders between operations called \emph{program order}:  
  for $\op_1,\op_2 \in \ops$, $\op_1 \ltpo \op_2$ means that 
  $\op_1$ and $\op_2$ were submitted to the same \site{}, and
  $\op_1$ occurred before $\op_2$,
\item 
  for $\meth \in \Meth$ and $\argv,\rv \in \Domain$, and $\op \in \ops$, 
  $\getlabel(\op) = (\meth,\argv,\rv)$ means that
  operation $\op$ is an invocation of $\meth$ with input 
  $\argv$ and returning $\rv$. The label 
  $\getlabel(\op)$ is sometimes denoted 
  $\mktrueaction{\meth}{\argv}{\rv}$.
\end{itemize}

Given an operation $\op$ from a \kvs{} history, whose label is either
$\mkwrite{\var}{\val}$ or $\mkread{\var}{\val}$, for
some $\var \in \Var$, $\val \in \Domain$,
we define
$\getvar{\op} = \var$ and
$\getval{\op} = \val$.

\subsection{Specification}

The consistency of a \replicated{}
is defined with respect to a particular 
specification, describing the correct behaviors of that object in a sequential setting.
A \emph{\specification} $\spec$ is thus defined\footnote{In general, specifications can be defined as 
sets of posets instead of sequences. This is to model conflict-resolution policies which are more 
general than choosing a total order between operations. In this paper, we focus on the \kvs{} 
whose specification is a set of sequences.
} 
as a set of
sequences
labeled by $\Meth \times \Domain \times \Domain$.

In this paper, we focus on the \kvs{} whose specification $\speckvs$ is defined inductively as
the smallest set of sequences closed
under the following rules ($\var \in \Var$ and $\val \in \Nats$): 
\begin{enumerate}
\item
  $\emptyseq \in \speckvs$,
\item
  if $\seqposet \in \speckvs$, then
  $\seqposet \cc \mkwrite{\var}{\val} \in \speckvs$,

\item
  if $\seqposet \in \speckvs$ contains no
  write on $\var$, then
  $\seqposet \cc \mkread{\var}{0} \in \speckvs$,
\item
   if $\seqposet \in \speckvs$ and the last write in $\rho$ on variable $\var$ is
  $\mkwrite{\var}{\val}$, then
  $\seqposet \cc \mkread{\var}{\val} \in \speckvs$.
\end{enumerate}

\begin{figure}
\centering
\begin{tikzpicture}[xscale=1.8,yscale=0.5]

\node (CC) at (1,0) {\Wcc};
\node (CCv) at (0,1) {\Ccv};
\node (SCC) at (0,-1) {\Scc};

\draw[->,very thick] (CCv) -- (CC);
\draw[->,very thick] (SCC) -- (CC);

\end{tikzpicture}

\caption{
Implication graph of causal consistency definitions.\\
}
\label{fig:relationships} 

\end{figure}
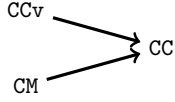

\begin{figure}
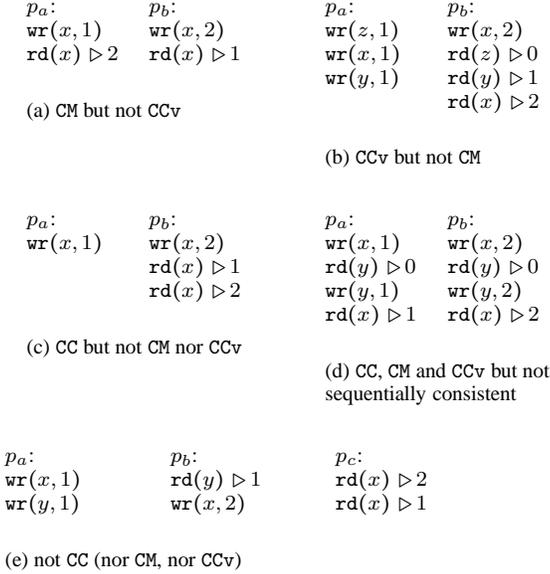

\footnotesize
\centering

\begin{subfigure}[t]{0.22\textwidth}

\begin{minipage}[t]{0.4\textwidth}
$p_a$:\\
$\mkwrite{\var}{1}$\\ 
$\mkread{\var}{2}$\\
\end{minipage}
\begin{minipage}[t]{0.2\textwidth}
$p_b$:\\
$\mkwrite{\var}{2}$\\ 
$\mkread{\var}{1}$\\
\end{minipage}

\caption{\scct{} but not \ccvt{}}
\label{fig:sccnotccv}

\end{subfigure}
\begin{subfigure}[t]{0.22\textwidth}
\begin{minipage}[t]{0.4\textwidth}
$p_a$:\\
$\mkwrite{z}{1}$\\
$\mkwrite{\var}{1}$\\ 
$\mkwrite{y}{1}$\\
\end{minipage}
\begin{minipage}[t]{0.20\textwidth}
$p_b$:\\
$\mkwrite{\var}{2}$\\ 
$\mkread{z}{0}$\\
$\mkread{y}{1}$\\
$\mkread{\var}{2}$\\
\end{minipage}

\caption{\ccvt{} but not \scct{}}
\label{fig:ccvnotscc}

\end{subfigure}

\vspace{2em}

\begin{subfigure}[t]{0.22\textwidth}

\begin{minipage}[t]{0.4\textwidth}
$p_a$:\\
$\mkwrite{\var}{1}$\\
\end{minipage}
\begin{minipage}[t]{0.20\textwidth}
$p_b$:\\
$\mkwrite{\var}{2}$\\
$\mkread{\var}{1}$\\
$\mkread{\var}{2}$\\
\end{minipage}

\caption{\wcct{} but not \scct{} nor \ccvt{}}
\label{fig:wccnotccvnotscc}

\end{subfigure}
\begin{subfigure}[t]{0.22\textwidth}

\begin{minipage}[t]{0.4\textwidth}
$p_a$:\\
$\mkwrite{\var}{1}$\\ 
$\mkread{y}{0}$\\
$\mkwrite{y}{1}$\\
$\mkread{\var}{1}$\\
\end{minipage}
\begin{minipage}[t]{0.2\textwidth}
$p_b$:\\
$\mkwrite{\var}{2}$\\ 
$\mkread{y}{0}$\\
$\mkwrite{y}{2}$\\
$\mkread{\var}{2}$\\
\end{minipage}

\caption{
\wcct{}, \scct{} and \ccvt{} but not \\
sequentially consistent
}
\label{fig:ccnotsc}

\end{subfigure}

\vspace{2em}

\begin{subfigure}[t]{0.6\textwidth}

\begin{minipage}[t]{0.2\textwidth}
$p_a$:\\
$\mkwrite{\var}{1}$\\ 
$\mkwrite{y}{1}$\\
\end{minipage}
\begin{minipage}[t]{0.2\textwidth}
$p_b$:\\
$\mkread{y}{1}$\\
$\mkwrite{\var}{2}$\\
\end{minipage}
\begin{minipage}[t]{0.2\textwidth}
$p_c$:\\
$\mkread{\var}{2}$\\
$\mkread{\var}{1}$\\
\end{minipage}

\caption{
not \wcct{} (nor \scct{}, nor \ccvt{})
}
\label{fig:notcc}

\end{subfigure}

\caption{
Histories showing the differences between the 
consistency criteria $\wcc$, $\scc$, and $\ccv$.
}
\label{fig:allfigs}
\end{figure}

\section{Causal Consistency}
\label{sec:causal}

Causal consistency is one of the most widely used consistency 
criterion for replicated objects. 
Informally speaking, it ensures that, if an operation $\op_1$ is 
\emph{causally related} to an operation $\op_2$ (e.g., some \site{} knew about $\op_1$ when executing $\op_2$), 
then all \sites{} must execute operation $\op_1$ before operation $\op_2$.
Operations which are not causally related may be executed in different 
orders by different \sites{}. 

From a formal point of view, there are several variations of causal consistency that apply to 
slightly different classes of implementations. In this paper, we consider three such variations
that we call causal consistency ($\wcc$), causal memory ($\scc$), and causal convergence ($\ccv$).
We start by presenting \emph{\wcc} followed by $\scc$ and $\ccv$,
which are both strictly stronger than $\wcc$. $\scc$ and $\ccv$
are not comparable (see \fig{fig:relationships}).

\subsection{Causal Consistency: Informal Description}

Causal consistency~\cite{jadphd,Perrin:2016:CCB:2851141.2851170} 
($\wcc$ for short) 
corresponds to the weakest notion of causal consistency that
exists in the literature.
We describe the intuition behind this notion of consistency using several examples, and
then give the formal definition.

Recall that a history $\hist$ models the point of view of a client using a \replicated{}, and 
it contains no information regarding the internals of the implementation, in particular,
the messages exchanged between \sites. This means that a history contains no notion of 
\emph{causality order}.
Thus, from the point of view of the client, a history is $\wcc$ 
as long as there \emph{exists} a causality order which explains the 
return value of each operation.
This is why, in the formal definition of $\wcc$ given in the 
next section,
the causality order $\co$ is existentially quantified.

\begin{example}
History~(\ref{fig:notcc}) is not $\wcc$. The reason is that there 
does not exist a causality order which explains the return values of 
all operations in the history. Intuitively, in any causality order,
$\mkwrite{y}{1}$ must be causally related to $\mkread{y}{1}$
(so that the read can return value $1$). By transitivity of the causality
order and because any causality order must contain the program order, 
$\mkwrite{x}{1}$ must be causally related to $\mkwrite{x}{2}$.
However, \site{} $p_c$ first reads $\mkread{x}{2}$, and then $\mkread{x}{1}$.
This contradicts the informal constraint that every \site{} must see
operations which are causally related in the same order.
\end{example}

\begin{example}
History~(\ref{fig:wccnotccvnotscc}) is $\wcc$. The reason is that 
we can define a causality order where the writes 
$\mkwrite{x}{1}$ and $\mkwrite{x}{2}$ are not causally related between them, but each write
is causally related to both reads.
Since the writes are not causally related, \site{} $p_b$ can read 
them in any order.
\end{example}

There is a subtlety here. In History~(\ref{fig:wccnotccvnotscc}),
\site{} $p_b$ first does $\mkread{x}{1}$, which implicitly means that it 
executed $\mkwrite{x}{1}$ after $\mkwrite{x}{2}$. Then $p_b$ 
does $\mkread{x}{2}$
which means that $p_b$ ``changed its mind'', and decided to order 
$\mkwrite{x}{2}$ after $\mkwrite{x}{1}$. This is allowed by $\wcc$, but as we
will see later, not allowed by the stronger criteria 
$\scc$ and $\ccv$.

This feature of $\wcc$ is useful for systems which do speculative 
executions and rollbacks~\citep{DBLP:conf/sosp/TerryTPDSH95,Telex}. 
It allows systems to execute 
operations by speculating on an order, and then possibly rollback, and 
change the order of previously executed operations.
This happens in particular in systems where convergence is important, 
where a consensus protocol is running in the background to make all \sites{} 
eventually agree on a total order of operations.
The stronger definitions, $\scc$ and $\ccv$,
are not suited to represent such speculative implementations.

\subsection{Causal Consistency: Definition}

\begin{table}

\begin{tabular}{l|l}
\axpoco & $\po \subseteq \co$ \\
\axcoarb & $\co \subseteq \theirarb$ \\
\axwcc & $\projectrv{\causalpast{\op}}{\op} \weaker \loc$ \\
\axscc & $\projectrv{\causalpast{\op}}{\poback{\op}} \weaker \loc$ \\
\axccv & $\projectrv{\causalarb{\op}}{\op} \weaker \loc$ \\
\end{tabular}

\vspace{2em}
where:\\
$\causalpast{\op} = (\causaldep{\op},\co,\getlabel)$ \\
$\causalarb{\op} = (\causaldep{\op},\theirarb,\getlabel)$ \\
$\causaldep{\op} = \set{\op' \in \ops\ |\ \op' \leqco \op}$ \\
$\poback{\op} = \set{\op' \in \ops\ |\ \op' \leqpo \op}$ \\

\caption{Axioms used in the definitions of causal consistency.}
\label{fig:axioms}

\end{table}

We now give the formal definition of $\wcc$, which corresponds
to the description of the previous section.
A \history{} $\hist = (\ops,\po,\getlabel)$ is $\wcct$ with
respect to a \specification{} $\spec$ when there
exists a strict partial order $\co \subseteq \ops \times \ops$, called  \emph{causal order}, such
that, for all operations $\op \in \ops$, there exists 
a specification sequence $\loc \in \spec$ such that axioms $\axpoco$ and $\axwcc$ hold 
(see \tab{fig:axioms}).

Axiom~\axpoco{} states that the causal order must at least contain
the program order. 
Axiom~\axwcc{} states that, for each operation $\op \in \ops$, 
the causal history of $\op$ (roughly, all the operations which are before $\op$ in the causal order)
can be sequentialized in order to obtain a valid sequence of
the specification $\spec$. This sequentialization must also preserve the constraints given by the causal order.
Formally, we define the \emph{causal past} of $\op$, $\causaldep{\op}$, as the \emph{set} of operations 
before $\op$ in the causal order and the \emph{causal history} of $\op$, $\causalpast{\op}$ as
the restriction of the causal order to the operations in its causal past.
Since a \site{} is not required to be consistent with the return values it has provided in the past or the return values provided by the other sites,
the axiom~\axwcc{} uses the causal history where only the return value of operation $\op$
has been kept. This is denoted by $\projectrv{\causalpast{\op}}{\op}$.
The fact that the latter can be sequentialized to a sequence $\loc$ in the specification is denoted by 
$\projectrv{\causalpast{\op}}{\op} \weaker \loc$. We defer the formal definition of these two last notations
to the next section.

\subsection{Operations on Labeled Posets}

First, we introduce an operator which projects 
away the return values of a subset of operations.
Let $\rho = (\ops,<,\getlabel)$ be a $\Meth \times \Domain \times \Domain$
labeled poset and
$\ops' \subseteq \ops$.
We denote by $\projectrv{\rho}{\ops'}$
the labeled poset where only the return values of the operations in
$\ops'$ have been kept.
Formally, $\projectrv{\rho}{\ops'}$ is
the $(\Meth \times \Domain) \cup (\Meth \times \Domain \times \Domain)$
labeled poset $(\ops,<,\getlabel')$ where
for all $\op \in \ops'$, $\getlabel'(\op) = \getlabel(\op)$,
and for all $\op \in \ops \setminus \ops'$, if
$\getlabel(\op) = (\meth,\argv,\rv)$, then
$\getlabel'(\op) = (\meth,\argv)$.
If $\ops' = \set{\op}$, we denote
$\projectrv{\rho}{\ops'}$ by $\projectrv{\rho}{\op}$.

Second, we introduce a relation on labeled posets, denoted $\weaker$.
Let 
$\rho = (\ops,<,\getlabel)$ and 
$\rho' = (\ops,<',\getlabel')$ 
be two posets labeled by 
$(\Meth \times \Domain) \cup
(\Meth \times \Domain \times \Domain)$ (the return values of some operations 
in $\ops$ might not be specified).
We denote by $\rho' \weaker \rho$
the fact that $\rho'$ has less order and label constraints on the set $\ops$.
Formally, $\rho' \weaker \rho$ if $<'\, \subseteq\, <$ and for all operation 
$\op \in \ops$, 
and for all $\meth \in \Meth$, $\argv,\rv \in \Domain$,
\begin{itemize}
  \item $\getlabel(\op) = \getlabel'(\op)$, or 
  \item
    $\getlabel(\op) = (\meth,\argv,\rv)$ implies
    $\getlabel'(\op) = (\meth,\argv)$.
\end{itemize}

\begin{example}
For any set of operations $\ops' \subseteq \ops$,
$\projectrv{\rho}{\ops'} \weaker \rho$. The reason is that
$\projectrv{\rho}{\ops'}$ has the same order constraints on
$\ops$ than $\rho$, but some return values are hidden in
$\projectrv{\rho}{\ops'}$.
\end{example}

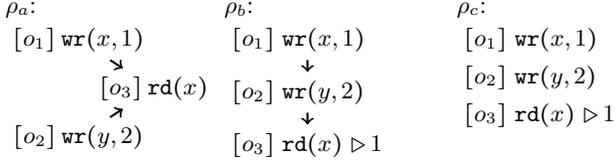
\begin{figure}
 
 \begin{minipage}[t]{0.16\textwidth}
  $\rho_a$:
  
  \begin{tikzpicture}[yscale=-1.3]
   \node (O1) at (0,0) { $[\op_1]$ $\mkwrite{\var}{1}$ };
   \node (O2) at (0,1) { $[\op_2]$ $\mkwrite{y}{2}$ };
   \node (O3) at (1,0.5) { $[\op_3]$ $\readmeth(\var)$ };
   
   \draw[->,thick] (O1) to (O3);
   \draw[->,thick] (O2) to (O3);
   \end{tikzpicture}
 \end{minipage}
 \begin{minipage}[t]{0.17\textwidth}
  $\rho_b$:
  
  \begin{tikzpicture}[yscale=-0.7]
   \node[text width=2.0cm] (O1) at (0,0) { $[\op_1]$ $\mkwrite{\var}{1}$ };
   \node[text width=2.0cm] (O2) at (0,1) { $[\op_2]$ $\mkwrite{y}{2}$ };
   \node[text width=2.0cm] (O3) at (0,2) { $[\op_3]$ $\mkread{\var}{1}$ };
   
   \draw[->,thick] (O1) to (O2);
   \draw[->,thick] (O2) to (O3);
   \end{tikzpicture}
 \end{minipage}
 \begin{minipage}[t]{0.10\textwidth}
  $\rho_c$:
  
  \begin{tikzpicture}[yscale=-0.5,xscale=3]
   \node[text width=3cm] (O1) at (0,0) { $[\op_1]$ $\mkwrite{\var}{1}$ };
   \node[text width=3cm] (O2) at (0,1) { $[\op_2]$ $\mkwrite{y}{2}$ };
   \node[text width=3cm] (O3) at (0,2) { $[\op_3]$ $\mkread{\var}{1}$ };
   
   \end{tikzpicture}
 \end{minipage}

 \caption{Illustration of the $\weaker$ relation.
 We have $\rho_a \weaker \rho_b$, but not
 $\rho_a \weaker \rho_c$.
 The label of an operation $\op$ is written next to $\op$.
 The arrows represent the transitive reduction of the strict 
 partial orders underlying the labeled posets.
 (For instance, none of the operations in $\rho_c$ are ordered.)
 }
 \label{fig:weakerrel}

\end{figure}

\begin{example}
In \fig{fig:weakerrel}, we have $\rho_a \weaker \rho_b$, 
as the only differences between $\rho_a$ and $\rho_b$ is the label of 
$\op_3$, and the fact that $\op_1 < \op_2$ holds in $\rho_b$ but not 
in $\rho_a$.

We have $\rho_a \not \weaker \rho_c$, as $\op_1 < \op_3$
holds in $\rho_a     $, but not in $\rho_c$.
\end{example}

\subsection{Causal Memory (\scc)}

Compared to causal consistency, causal memory~\citep{Ahamad94causalmemory,Perrin:2016:CCB:2851141.2851170} 
(denoted $\scc$) does not allow a \site{} 
to ``change its mind'' about the order of operations.
The original definition of causal memory of~\citet{Ahamad94causalmemory} applies only to the \kvs{} and it was extended by~\citet{Perrin:2016:CCB:2851141.2851170} to
arbitrary specifications. We use the more general definition, since it was also shown that it coincides with the original one for 
histories where for each variable $\var \in \Var$, 
the values written to $\var$ are unique.

For instance, History~(\ref{fig:wccnotccvnotscc}) is $\wcct$ but not 
$\scct$. Intuitively, the reason is that \site{} $p_b$ first decides
to order $\mkwrite{x}{1}$ after $\mkwrite{x}{2}$ (for $\mkread{x}{1}$)
and then decides to order $\mkwrite{x}{2}$ after $\mkwrite{x}{1}$
(for $\mkread{x}{2}$).

On the other hand, History~(\ref{fig:sccnotccv}) is $\scct$.
\Sites{} $p_a$ and $p_b$ disagree on the order of the two write operations,
but this is allowed by $\scc$, as we can define a causality order where 
the two writes are not causally related.

Formally, a \history{} $\hist = (\ops,\po,\getlabel)$ is $\scct$ with
respect to a \specification{} $\spec$ if there
exists a strict partial order $\co \subseteq \ops \times \ops$ such
that, for each operation $\op \in \ops$, there exists a specification sequence  
$\loc \in \spec$ such that axioms $\axpoco$ and $\axscc$ 
hold. With respect to $\wcct$, causal memory requires that each \site{} is consistent
with respect to the return values it has provided in the past. A 
\site{} is still not required
to be consistent with the return values provided by other \sites{}.
Therefore, $\axscc$ states:
\[
    \projectrv{\causalpast{\op}}{\poback{\op}} \weaker \loc
\]
where $\projectrv{\causalpast{\op}}{\poback{\op}}$ is the causal history where 
only the return values of 
the operations which are before $\op$ in the program order 
(in $\poback{\op}$) are kept. For finite histories, if we set $\op$ to be the 
last operation of a \site{} $\tid$, this means that we must explain 
all return values of operations in $\tid$ by a single sequence 
$\loc \in \spec$.
In particular, this is not possible for 
for site $p_b$ in History~(\ref{fig:wccnotccvnotscc}).

The following lemma gives the relationship between $\scc$ and $\wcc$.

\begin{lemma}[\citep{Perrin:2016:CCB:2851141.2851170}]
\label{lem:cmtocc}
If a history $\hist$ is \scct{} with respect to a specification 
$\spec$, then $\hist$ is \wcct{} with respect to $\spec$.
\end{lemma}

\begin{proof}
We know by definition that there exists a strict partial order $\co$ such
that, for all operation $\op \in \ops$, there exists $\loc \in \spec$
such that
axioms $\axpoco$ and $\axscc$.
In particular, for any $\op \in \ops$, we have 
$\projectrv{\causalpast{\op}}{\poback{\op}} \weaker \loc$.

Since $\projectrv{\causalpast{\op}}{\op} \weaker
\projectrv{\causalpast{\op}}{\poback{\op}}$, and the relation 
$\weaker$ is transitive, 
we have 
$\projectrv{\causalpast{\op}}{\op} \weaker \loc$, and axiom $\axwcc$ holds.
\end{proof}

\subsection{Causal Convergence (\ccv)}

Our formalization of \conv{} (denoted $\ccv$) corresponds to the definition of 
\emph{causal consistency} given in 
\citet{Gotsman05} and \citet{principles-of-eventual-consistency} restricted to sequential specifications.
$\ccv$  was introduced in the context of 
\emph{eventual consistency}, another consistency criterion guaranteeing that roughly, 
all \sites{} eventually converge towards the same state, when no new updates are submitted.

Causal convergence uses a total order between all the operations in a history, called the \emph{arbitration order},
as an abstraction of the conflict resolution policy applied by \sites{} to agree on how to 
order operations which are \emph{not} causally related.
As it was the case for the causal order, the arbitration order, denoted by $\theirarb$, is not encoded explicitly in the notion of history and
it is existentially quantified in the definition of $\ccv$.

\begin{example}
History~(\ref{fig:sccnotccv}) is not $\ccvt$. The reason is that, for the 
first \site{} $p_a$ to read $\mkread{\var}{2}$, 
the write $\mkwrite{\var}{2}$ must be after $\mkwrite{\var}{1}$ in the 
arbitration order. Symmetrically, 
because of the $\mkread{\var}{1}$,
$\mkwrite{\var}{2}$ must be before $\mkwrite{\var}{1}$ in the 
arbitration order, which is not possible.
\end{example}

\begin{example}
History~(\ref{fig:ccvnotscc}) gives a \history{} which is $\ccvt$ but 
not $\scct$. To prove that it is $\ccvt$, 
a possible arbitration order is to have 
the writes of $p_a$ all before the $\mkwrite{\var}{2}$ operation,
and the causality order then relates $\mkwrite{y}{1}$ to
$\mkread{y}{1}$.

On the other hand, History~(\ref{fig:ccvnotscc}) is not $\scct$.
If History~(\ref{fig:ccvnotscc})
were \scct{}, 
for \site{} $p_b$,
$\mkwrite{y}{1}$ should go before $\mkread{y}{1}$. By transitivity, 
this implies that $\mkwrite{x}{1}$ should go before $\mkread{x}{2}$.
But for the read  $\mkread{x}{2}$ to return value $2$,
$\mkwrite{x}{1}$  should then also go before $\mkwrite{x}{2}$.
This implies that $\mkwrite{z}{1}$ goes before $\mkread{z}{0}$
preventing  $\mkread{z}{0}$ from reading the initial value $0$.
\end{example}

\begin{example}
History~(\ref{fig:ccnotsc}) shows that all causal consistency
definitions (\wcc, \scc{}, and \ccv)
are strictly weaker than sequential consistency.
Sequential consistency~\citep{Lamport:1979:MMC:1311099.1311750} 
imposes a total order on all (read and write) operations. 
In particular, no such total order can exist for 
History~(\ref{fig:ccnotsc}). Because of the initial writes 
$\mkwrite{\var}{1}$ and $\mkwrite{\var}{2}$, and the final
reads $\mkread{\var}{1}$ and $\mkread{\var}{2}$, 
all the operations of $p_a$ must be completely ordered before the
operations of $p_b$, or vice versa. This would make one
of the $\mkread{y}{0}$ to be ordered after either $\mkwrite{y}{1}$
or $\mkwrite{y}{2}$, which is not allowed by the \kvs{} specification.
On the other hand, History~(\ref{fig:ccnotsc}) satisfies all criteria 
$\wcc$, $\scc$, $\ccv$. The reason is that we can set the causality order 
to not relate any operation from $p_a$ to $p_b$ nor from 
$p_b$ to $p_a$. 
\end{example}

Formally, 
a history $\hist$ is \ccvt{} with respect to $\spec$ if there exist
a strict partial order $\co \subseteq \ops \times \ops$ and
a strict total order $\theirarb \subseteq \ops \times \ops$ such that,
for each operation $\op \in \ops$, there exists a specification sequence  
$\loc \in \spec$ 
such that
the axioms $\axpoco$, $\axcoarb$, and $\axccv$ hold.
Axiom $\axcoarb$ states that the arbitration order $\theirarb$ must at least 
respect the causal order $\co$.
Axiom $\axccv$ states that, to explain the return value of an operation $\op$,
we must sequentialize the operations which are in the causal past of 
$\op$, while respecting the arbitration order $\theirarb$.

Axioms~\axccv{} and \axcoarb{} imply axiom~\axwcc{}, as
the arbitration order $\theirarb$ contains the causality order 
$\co$. We therefore have the following lemma.

\begin{lemma}[\citep{Perrin:2016:CCB:2851141.2851170}]
\label{lem:ccvtocc}
If a history $\hist$ is \ccvt{} with respect to a specification $\spec$, 
then $\hist$ is \wcct{} with respect to $\spec$.
\end{lemma}

\begin{proof}
Similar to the proof of \lem{lem:cmtocc}, but using the fact 
that $\projectrv{\causalpast{\op}}{\op} \weaker 
\projectrv{\causalarb{\op}}{\op}$ (since by axiom $\axcoarb$, 
$\co \subseteq \theirarb$).
\end{proof}

\section{Single History Consistency is NP-complete}
\label{sec:npnpcomplete}

We first focus on the problem of checking whether a given history 
is consistent, which is relevant for instance in the context of testing a given replicated object. 
We prove that this problem is \npcomplete{} for all the three
variations of causal consistency (\wcc{}, \scc{}, \ccv{}) and the \kvs{} specification.

\begin{restatable}{lemma}{npnp}
\label{lem:npcomplete}
Checking whether a history $\hist$ is \wcct{} (\resp \scct{}, \resp \ccvt{}) with respect to $\speckvs$
is \npcomplete{}.
\end{restatable}

\begin{proof}

Membership in $\NP$ holds for all the variations of causal consistency,
and any specification $\spec$ for which there is a polynomial-time 
algorithm that can check whether a given sequence is in $\spec$.
This includes the \kvs{}, and common objects such as sets,
multisets, stacks, or queues.
It 
follows from the fact that one can guess a causality order $\co$
(and an arbitration order $\theirarb$ for \ccv), and a sequence 
$\loc$ for each operation $\op$, and then check in polynomial 
time whether the axioms of \tab{fig:axioms} hold, 
and whether $\loc \in \spec$.

For NP-hardness, we reduce boolean satisfiability to checking consistency of a single history reusing the encoding from \citet{roland}.
Let  $\form$ be a boolean formula in CNF with variables 
$\sx_1,\dots,\sx_n$, and clauses 
$\scl_1,\dots,\scl_k$.
The goal is to define a history $\hist$ which is $\wcct$
if and only if $\form$ is satisfiable. All operations on $\hist$
are on a single variable $\npvar$.
For the encoding, we assume that each clause corresponds to a unique 
integer strictly larger than $n$.

For $i \in \set{1,\dots,n}$, we define 
$\getpos{\sx_i}$ as the set of clauses where
$\sx_i$ appears positively, and 
$\getneg{\sx_i}$ as the set of clauses where $\sx_i$ appears negatively.

For each $i \in \set{1,\dots,n}$, $\hist$ contains two \sites{}, 
$\gettrue{i}$ and $\getfalse{i}$.
\Site{} $\gettrue{i}$ first writes each $\scl \in \getpos{\sx_i}$ 
(in the order they appear in $\scl_1,\dots,\scl_k$) to variable $\npvar$,
and then, it writes $i$.
Similarly,
\Site{} $\getfalse{i}$ writes each $\scl \in \getneg{\sx_i}$ 
(in the order they appear in $\scl_1,\dots,\scl_k$) to variable $\npvar$,
and then, it writes $i$.

Finally, a \site{} $\peval$ does 
$\mkread{\npvar}{1} \cdots \mkread{\npvar}{n}$ followed by
$\mkread{\npvar}{\scl_1} \cdots \mkread{\npvar}{\scl_k}$.

We then prove the following equivalence:
(the equivalence for $\scc$ and $\ccv$ can be proven similarly):
$\hist$ is $\wcct$ iff
$\form$ is satisfiable.

$(\Leftarrow)$ This direction follows from the proof of \citep{roland}.
They show that if $\form$ is satisfiable, the history $\hist$ is 
sequentially consistent.

$(\Rightarrow)$ Assume $\hist$ is $\wcct$. Then, there exists $\co$, 
such that, for all $\op \in \ops$, there exists $\loc \in \speckvs$, 
such that $\axpoco$ and $\axwcc$ hold.
In particular, each $\mkread{\npvar}{i}$ of $\peval$ must have 
a $\mkwrite{\npvar}{i}$ in its causal past. 

The $\mkwrite{\npvar}{i}$ operation can either be from $\gettrue{i}$ 
(corresponding to setting variable $\sx_i$ to false in $\form$), 
or from $\getfalse{i}$
(corresponding to setting variable $\sx_i$ to true in $\form$). 
For instance, if it is from \site{} $\gettrue{i}$, then none of
of the $\mkwrite{\npvar}{\scl}$ for $\scl \in \getpos{\sx_i}$
can be used for the reads $\mkread{\npvar}{\scl_i}$ of $\peval$.

Consequently, for any variable $\sx_i$, only the writes of 
$\mkwrite{\npvar}{\scl}$ for $\scl \in \getpos{\sx_i}$, or 
the ones with $\scl \in \getneg{\sx_i}$ can be used 
for the reads $\mkread{\npvar}{\scl_i}$ of $\peval$.

Moreover, each read $\mkread{\npvar}{\scl_i}$ has a corresponding 
$\mkwrite{\npvar}{\scl_i}$, meaning that $\form$ is satisfiable.
\end{proof}

The reduction from boolean satisfiability used to prove NP-hardness uses
histories where the same value is written multiple times on the same variable.
We show in \sect{sec:poly} that this is in fact necessary to 
obtain the $\NP$-hardness: when every value is written only once per variable,
the problem becomes polynomial time.

\section{Undecidability of Verifying Causal Consistency}
\label{sec:kvsundec}

We now consider the problem of checking
whether all histories of an implementation are causally consistent.
We consider this problem for all variants of causal consistency
(\wcc{}, \scc{}, \ccv{}).

We prove that this problem is undecidable.
In order to formally prove the undecidability, 
we describe an abstract model for representing implementations.

\subsection{Executions and Implementations}

Concretely, an implementation is represented by a set of 
\emph{executions}.
Formally, an \emph{execution} is a sequence of operations.
Each operation is labeled by
an element $(\tid,\meth,\argv,\rv)$ 
of $\Tid \times \Meth \times \Domain \times \Domain$, meaning
that $\meth$ was called with argument value $\argv$ on \site{} $\tid$, 
and returned value $\rv$.
An \emph{implementation} $\lib$ is a set of executions which is prefix-closed
(if $\lib$ contains an execution $\exec \cc \exec'$, $\lib$ also contains 
$\exec$).

All definitions given for histories (and sets of histories) transfer
to executions (and sets of executions) as for each execution $\exec$, we 
can define a corresponding history $\hist$.
The history $\hist = (\ops,\po,\getlabel)$ 
contains the same operations as $\exec$, and orders
$\op_1 \ltpo \op_2$ if $\op_1$ and $\op_2$ are labeled by the same \site{}, 
and $\op_1$ occurs before $\op_2$ in $\exec$.

For instance, an \implementation{} is \emph{\dataindependent} if 
the corresponding set of histories is \dataindependent.

\subsection{Undecidability Proofs}

We prove undecidability even when 
$\lib$ and $\spec$ are regular languages (given by regular expressions
or by finite automaton). We refer to this as the first undecidability proof.
Even stronger, we give a second undecidability proof, which
shows that this problem is undecidable 
when the specification is set to $\speckvs$,
with a fixed number of variables, 
and with a fixed domain size
(which is a particular regular language).

These results imply that the undecidability does not come from the 
expressiveness of the model used to describe implementations,
nor from the complexity of the specification,
but specifically from the fact that we are checking causal consistency.

For both undecidability proofs, 
our approach is to reduce the Post Correspondence Problem ($\pcp$, 
an undecidable problem in formal languages), to the problem 
of checking whether $\lib$ is \emph{not} causally consistent
(\resp \wcc,\scc,\ccv).

\begin{definition}
 Let $\alppcp$ be a finite alphabet.
\pcp{} asks, given 
$\npcp$ pairs
$(\sigu_1,\sigv_1),\dots,(\sigu_\npcp,\sigv_\npcp) 
\in (\alppcp^* \times \alppcp^*)$,
whether there exist 
$i_1,\dots,i_\kpcp \in \set{1,\dots,\npcp}$ such that
$\sigu_{i_1} \cdots \sigu_{i_\kpcp} = 
    \sigv_{i_1} \cdots \sigv_{i_\kpcp}$, with $(\kpcp > 0)$.
\end{definition}

From a high-level view, both proofs operate similarly.
We build, from a \pcp{} instance $\instpcp$, 
an implementation $\lib$ (which is here a regular language) 
-- and for the first proof, a specification $\spec$ --
such that $\instpcp$ has a positive answer if and only if 
$\lib$ contains an execution which is not causally consistent
(\resp \wcc,\scc,\ccv) with respect to $\spec$ (with respect to a 
bounded version of $\speckvs$ for the second proof).

The constructed implementations $\lib$ produce, for each possible pair of 
words $(u,v)$, an execution whose history $\encoding{u}{v}$ is \emph{not} 
causally consistent (\resp \wcc,\scc,\ccv) if and only if $(u,v)$
form a \emph{\validanswer} for $\instpcp$.

\begin{definition}
Two sequences $(\wordu, \wordv)$ in $\alppcp^*$ form a 
\emph{\validanswer} if $\wordu = \wordv$ and they can be decomposed into
$\wordu = \sigu_{i_1} \cc \sigu_{i_2} \cdots \sigu_{i_\kpcp}$ and
$\wordv = \sigv_{i_1} \cc \sigv_{i_2} \cdots \sigv_{i_\kpcp}$,
with each $(\sigu_{i_j},\sigv_{i_j})$ corresponding to a pair of problem
$\instpcp$.
\end{definition}

Therefore, $\lib$ is not causally consistent, if and only if 
$\lib$ contains an execution whose history $\encoding{u}{v}$ is 
not causally consistent, if and only if there exists 
$(u,v)$ which form a \validanswer{}
for $\instpcp$, if and only if $\instpcp$ has a positive answer.

\subsection{Undecidability For Regular Specifications}

For the first proof, 
we first prove that the \emph{shuffling problem}, 
a problem on formal languages that we introduce, 
is not decidable. This is done by reducing \pcp{} to the shuffling problem.

We then reduce the shuffling problem to checking whether an 
implementation is not causally consistent (\resp \wcc,\scc,\ccv), 
showing that verification of causal consistency is undecidable as well.

Given two words $\wu,\wv \in \alp^*$, the \emph{shuffling} operator returns 
the set of words which can be obtained from $\wu$ and $\wv$ by interleaving 
their letters. Formally, we define  
$\shuffle{\wu}{\wv} \subseteq \alp^*$ 
inductively: $\shuffle{\emptyseq}{\wv} = \set{\wv}$, 
$\shuffle{\wu}{\emptyseq} = \set{\wu}$ and 
$\shuffle{(a \cc \wu)}{(b \cc \wv)} = a \cc (\shuffle{\wu}{(b\cc\wv)}) \cup
b \cc (\shuffle{(a \cc \wu)}{\wv})$, with $a,b \in \alp$.

\begin{definition}
The \emph{shuffling problem} asks, given a
regular language $\autL$ over an alphabet $(\alpinu \uplus \alpinv)^*$,
if there exist
$\wu \in \alpinu^*$ and $\wv \in \alpinv^*$ such that 
    $\shuffle{\wu}{\wv} \cap \autL = \emptyset$.
\end{definition}

\begin{restatable}{lemma}{shuffling} \label{lem:shuffle}
The shuffling problem is undecidable. 
\end{restatable}

We now give the undecidability theorem for causal consistency
(\resp \wcc,\scc,\ccv), 
by reducing the shuffling problem to 
the problem of verifying (non-)causal consistency.
The idea is to let one \site{} simulate 
words from $\alpinu^*$, and the second \site{} 
from 
$\alpinv^*$. We then set the specification to be 
(roughly) the language $\autL$. We therefore obtain that there exists an 
execution which is \emph{not}
causally consistent with respect to $\autL$ if and only
if there exist $\wu \in \alpinu^*$, $\wv \in \alpinv^*$ such that 
no interleaving of $\wu$ and $\wv$ belongs to $\autL$, \ie 
    $\shuffle{\wu}{\wv} \cap \autL = \emptyset$.

\begin{restatable}{theorem}{undecidability}
\label{thm:causalundec}
Given an \library{} $\lib$ 
and a specification $\spec$ given 
as regular languages, checking whether all executions of 
$\lib$ are causally consistent (\resp \wcct, \scct, \ccvt) 
with respect to $\spec$ is undecidable.
\end{restatable}

\subsection{Undecidability for the read/write memory abstraction}

Our approach for the second undecidability proof is to reduce directly 
\pcp{} to the problem of checking whether a finite-state implementation is 
\emph{not} \wcct{} (\resp \scct, \ccvt) with respect to the \kvs{},
without going through the shuffling problem.
The reduction here is much more technical, and requires $13$ \sites{}.
This is due to the fact that we cannot encode the constraints we 
want in the specification (as the specification is set to be $\speckvs$),
and we must encode them using appropriately placed read and write operations.

\begin{restatable}{theorem}{causalundeckvs}
\label{thm:causalundeckvs}
Given an \library{} $\lib$  
as a regular language, checking whether all executions of 
$\lib$ are causally consistent (\resp \wcct, \scct, \ccvt) with respect to 
$\speckvs$ 
is undecidable.

\end{restatable}

\section{Causal Consistency under Data Independence}

\label{sec:dataind}

Implementations used in practice are typically
\emph{data independent}~\citep{conf/tacas/AbdullaHHJR13}, \ie their 
behaviors do not depend on the particular data values which are stored at 
a particular variable. 
Under this assumption, we prove in \sect{sec:complete}
that it is enough to verify causal consistency 
for histories which use distinct $\writemeth$ values, 
called \emph{\differentiated} histories.

We then show in \sect{sec:badpatterns}, 
for each definition of causal consistency,
how to characterize non-causally consistent (\differentiated{}) histories 
through the presence of certain sets of operations.

We call these sets of operations \emph{bad patterns}, 
because any history containing one bad pattern
is necessarily not consistent (for the considered consistency criterion).
The bad patterns are defined through various relations derived from a 
\differentiated{} history, and are all 
computable in polynomial time~(proven in \sect{sec:poly}).
For instance, for $\wcc$, we provide in \sect{sec:badpatterns} 
four bad patterns such that, a differentiated history $\hist$ is 
$\wcc$ if and only if $\hist$ contains none of these bad patterns.
We give similar lemmas for $\scc$ and $\ccv$.

\subsection{Differentiated Histories}
\label{sec:complete}

Formally, a \history{} $(\ops,\po,\getlabel)$ is said to 
be \emph{\differentiated} if for all $\op_1 \neq \op_2$, if 
$\getlabel(\op_1) = \mktrueaction{\writemeth}{\var}{\dat_1}$ and 
$\getlabel(\op_2) = \mktrueaction{\writemeth}{\var}{\dat_2}$, then 
$\dat_1 \neq \dat_2$, and there are no operation 
$\mkwrite{\var}{0}$ (which writes the initial value).
Let $\hists$ be a set of labeled posets.
We denote by $\getdiff{\hists}$ the subset of 
\differentiated{} \histories{} of $\hists$.

A \emph{renaming} $\renaming: \Nats \times \Nats$ is a function which 
modifies the data values of operations.
Given a \kvs{} history $\hist$, we define 
by $\applyrenaming{\hist}{\renaming}$ the history where 
any number $n \in \Nats$ appearing in a label of 
$\hist$ is changed to $\renaming(n)$.

A set of histories $\hists$ is \emph{\dataindependent} if, 
for every history
$\hist$,
\begin{itemize}
\item 
  there exists a \differentiated{} \history{} $\hist' \in \hists$, 
  and a renaming $\renaming$, such that 
  $\hist = \applyrenaming{\hist'}{\renaming}$.
\item
  for any renaming $\renaming$, $\applyrenaming{\hist}{\renaming} \in \hists$.
\end{itemize}

The following lemma shows that for the verification of a
\dataindependent{} set of histories, it is enough to consider
\differentiated{} histories.

\begin{restatable}{lemma}{diffhist}
\label{lem:diffhists}
Let $\hists$ be a \dataindependent{} set of histories.
Then, $\hists$ is causally consistent (\resp \wcc, \scc, \ccv) 
with respect to the \kvs{} 
if and only if 
$\getdiff{\hists}$ is causally consistent (\resp \wcc, \scc, \ccv)
with respect to the \kvs{}.
\end{restatable}

\begin{table}

\begin{tabular}{lp{0.31\textwidth}}
\bpbpcyclicco & \bpcocycle \\[0.5ex]
\bpbprfi & \bprfi  \\[0.5ex]
\bpbprf & \bprf \\[0.5ex]
\bpbpco & \bpco \\[0.5ex]
\bpbpchangerfi & \bpchangerfi{}  \\[0.5ex]
\bpbpcyclichb & \bpchange{}  \\[0.5ex]
\bpbpcycliccf & \bpdisagree{}  \\
\end{tabular}

\caption{All bad patterns defined in the paper.}
\label{tab:allbps}

\end{table}

\begin{table}

\centering
\begin{tabular}{lll}
\wcc  & \scc & \ccv 
\\
\hline
& & \\[-1ex]
    \bpbpcyclicco & \bpbpcyclicco & \bpbpcyclicco \\
    \bpbprfi & \bpbprfi & \bpbprfi \\
    \bpbprf & \bpbprf & \bpbprf \\
    \bpbpco & \bpbpco & \bpbpco \\
            & \bpbpchangerfi & \bpbpcycliccf \\
            & \bpbpcyclichb & \\
\end{tabular}

\caption{Bad patterns for each criteria.}
\label{table:bptable}

\end{table}
 
\subsection{Characterizing Causal Consistency (\wcc)}
\label{sec:badpatterns}

Let $\hist = (\ops,\po,\getlabel)$ be a \differentiated{} 
history.
We now define and explain the bad patterns of $\wcc$.
They are defined using the \emph{\readfrom} relation.
The \readfrom{} relation relates each write $w$ to each read that 
reads from $w$.
Since we are considering \differentiated{} histories, we can determine,
only by looking at the operations of a history, from which write each
read is reading from. There is no ambiguity, as each value can only 
be written once on each variable.

\begin{definition}
The 
\emph{\readfrom} relation $\rf$ is defined as:
\[
    \set{(\op_1,\op_2)\ |\ \exists \var \in \Var, \dat \in \Domain.\ 
        \getlabel(\op_1) = \mktrueaction{\writemeth}{\var,\dat}{\unit} \land 
        \getlabel(\op_2) = \mktrueaction{\readmeth}{\var}{\dat}
    }.
\]
 
The relation $\propco$ is defined as $\propco = (\po \cup \rf)^+$. 
\end{definition}

\begin{remark}
Note that we use lower-case $\co$ for the existentially quantified
causality order which appears in the definition of causal consistency,
while we use upper-case $\propco$ for the relation fixed as 
$(\po \cup \rf)^+$.
The relation $\propco$ represents the smallest causality order possible.
We in fact show in the lemmas~\ref{lem:causalreg},
\ref{lem:ccvbadpatterns}, and \ref{lem:strongcausalreg}, that when 
a history is $\wcct$ (\resp $\scct$, $\ccvt$), the causality order $\co$
can always be set to $\propco$.
\end{remark} 

There are four bad patterns for $\wcc$, defined in terms of the 
$\rf$ and $\propco$ relations: 
\bpbpcyclicco, \bpbprfi, \bpbprf, \bpbpco{} (see \tab{tab:allbps}).

\begin{example}

History~(\ref{fig:notcc}) contains bad pattern \bpbpco.
Indeed,  $\mkwrite{\var}{1}$ is causally related (through relation $\propco$)
to $\mkwrite{\var}{2}$, which is causally related to 
$\mkread{\var}{1}$.
Intuitively, this means that the \site{} executing 
$\mkread{\var}{1}$ is aware of both 
writes $\mkwrite{\var}{1}$ and $\mkwrite{\var}{2}$, but chose 
to order $\mkwrite{\var}{2}$ before $\mkwrite{\var}{1}$, while 
$\mkwrite{\var}{1}$ is causally related to $\mkwrite{\var}{2}$.
As a result, History~(\ref{fig:notcc}) is not \wcct{} (nor \scct{}, 
nor \ccvt{}).

History~(\ref{fig:ccnotsc}) contains none of the bad patterns defined 
in \tab{tab:allbps}, and satisfies all definitions of causal consistency.
In particular, History~(\ref{fig:ccnotsc}) is $\wcct$.

\end{example}

\begin{restatable}{lemma}{badpatterns}
\label{lem:causalreg}
A \differentiated{} \history{} $\hist$ is \wcct{} 
with respect to $\speckvs$ if and only if $\hist$ does not contain
one of the following bad patterns:
\bpbpcyclicco, \bpbprfi, \bpbprf, \bpbpco.
\end{restatable}

\begin{proof} 
Let $\hist = (\ops,\po,\getlabel)$ be a \differentiated{} history.

$(\Rightarrow)$ 
Assume that $\hist$ is \wcct{} with respect to $\speckvs$.
We prove by contradiction that $\hist$ cannot contain 
bad patterns \bpbpcyclicco, \bpbprfi, \bpbprf, \bpbpco{}.

First, we show that $\propco \subseteq \co$.
Given the specification of $\readmeth$'s, and given that $\hist$ 
is \differentiated{}, we must 
have $\rf \subseteq \co$. Moreover, by axiom $\axpoco$, 
$\po \subseteq \co$. Since $\co$ is a transitive order, 
we thus have $(\po \cup \rf)^+ \subseteq \co$ and $\propco \subseteq \co$.

(\bpbpcyclicco) Since $\co$ is acyclic, and $\propco \subseteq \co$,
$\propco$ is acyclic as well.

(\bpbprfi) If there is a
$\mkread{\var}{0}$ operation $\rop$, and an operation $\wop$
  such that $\wop \ltpropco \rop$ with $\getvar{\wop} = \var$:
  we obtain a contradiction, because \wcc{} ensures that there exists 
  a sequence $\locof{\rop} \in \speckvs$ that orders $\wop$ before 
  $\rop$.
  However, this is not allowed by $\speckvs$, as $\hist$ is 
  \differentiated{}, and does not contain operation that write the initial
  value $0$.
  A read 
  $\mkread{\var}{0}$ can thus happen only when there are no previous write 
  operation on $\var$.

(\bpbprf) Similarly, 
we cannot have a $\mkread{\var}{\val}$ operation $\rop$ such that 
$\val \neq 0$, and such that there is no $\wop$ operation with 
$\wop \ltrf \rop$. Indeed, \wcc{} ensures that there exists 
  a sequence $\locof{\rop} \in \speckvs$ that contains $\rop$.
  Moreover, 
  $\speckvs$ allows 
  $\mkread{\var}{\val}$ operations only when there is a previous 
  write $\mkwrite{\var}{\val}$.
  So there must exist a $\mkwrite{\var}{\val}$ operation $\wop$
  (such that $\wop \ltrf \rop$).

(\bpbpco)
If there exist 
$\wop_1,\wop_2,\rop_1 \in \ops$ such that  
\begin{itemize}
\item $\wop_1 \ltrf \rop_1$ and
\item $\wop_1 \ltpropco \wop_2$ with $\getvar{\wop_1} = \getvar{\wop_2}$ and
\item $\wop_2 \ltpropco \rop_1$.
\end{itemize}

Let $\var \in \Var$ and $\dat_1 \neq \dat_2 \in \Nats$ such that:
\begin{itemize}
\item $\getlabel(\wop_1) = \mktrueaction{\writemeth}{\var,\dat_1}{\unit}$,
\item $\getlabel(\wop_2) = \mktrueaction{\writemeth}{\var,\dat_2}{\unit}$,
\item $\getlabel(\rop_1) = \mktrueaction{\readmeth}{\var}{\dat_1}$.
\end{itemize}

By \wcc{}, and since $\propco \subseteq \co$ ,
we know there exists $\locof{\rop_1} \in \speckvs$ that 
contains $\wop_1$ before $\wop_2$, and ends with $\rop_1$.
The specification $\speckvs$ require the last 
write operation on $\var$ to be a $\mktrueaction{\writemeth}{\var,\dat_1}{\unit}$.
However, as $\hist$ is differentiated, the only 
$\mktrueaction{\writemeth}{\var,\dat_1}{\unit}$ operation is 
$\wop_1$. As a result, the last write operation on variable $\var$
in $\locof{\rop_1}$ cannot be $\wop_1$ (as $\wop_2$ is after $\wop_1$),
and we have a contradiction.

$(\Leftarrow)$ 
Assume that $\hist$ contains none
of the bad patterns described above.
We show that $\hist$ is \wcct.
We use for this the strict partial order 
$\propco = (\po \cup \rf)^+$ as the causal order $\co$.
The relation $\propco$ is a strict partial order, as $\hist$ does not 
contain bad pattern $\bpbpcyclicco$.
Axiom \axpoco{} holds by construction.
We define for each operation $\op \in \ops$ a sequence 
$\loc \in \speckvs$ such that 
$\projectrv{\causalpast{\op}}{\op} \weaker \loc$ (such that 
$\axwcc$ holds).

Let $\op \in \ops$. We have three cases to consider.

1) If $\op$ is a $\writemeth$ operation, then all the 
return values of the read 
operations in $\projectrv{\causalpast{\op}}{\op}$ are hidden.
We can thus define $\loc$ as any sequentialization of 
$\projectrv{\causalpast{\op}}{\op}$, and where we add the 
appropriate return values to the read operations (the value written by
the last preceding write on the same variable).

2)
If $\op$ is a $\readmeth$ operation $\rop$, labeled by
$\mkread{\var}{0}$ for some $\var \in \Var$.
We know by the fact that $\hist$ does not contain 
$\bpbprfi$ that there is no 
$\wop$ such that $\wop \ltpropco \rop$.
As a result, we can define $\loc$ as any sequentialization of
$\projectrv{\causalpast{\op}}{\op}$, where we 
add the appropriate return values to the read operations different from 
$\rop$.

3)
If $\op$ is a $\readmeth$ operation $\rop_1$, labeled by
$\mktrueaction{\readmeth}{\var}{\dat_1}$ for 
some $\var \in \Var$ and $\dat_1 \neq 0$, 
we know by assumption that there 
exists a $\writemeth$ operation $\wop_1$ such that $\wop_1 \ltrf \rop_1$
($\hist$ does not contain $\bpbprf$). 

We also know ($\hist$ does not contain $\bpbpco$) there is no
$\wop_2$ such that  
\begin{itemize}
\item $\wop_1 \ltrf \rop_1$ and
\item $\wop_1 \ltpropco \wop_2$ with $\getvar{\wop_1} = \getvar{\wop_2}$ and
\item $\wop_2 \ltpropco \rop_1$.
\end{itemize}
This ensures that 
$\wop_1$ must be a maximal $\writemeth$ operation on variable 
$\var$ in $\causalpast{\op}$. It is thus possible to sequentialize
$\projectrv{\causalpast{\op}}{\op}$ into 
$\locof{\op}$, so that $\wop_1$ is the last write on variable $\var$.
We can then add appropriate return values to the read operations
different than $\rop_1$,
whose
return values were hidden in $\projectrv{\causalpast{\op}}{\op}$
(the value written by the last preceding write in $\locof{\op}$ on the same 
variable).
\end{proof}

\subsection{Characterizing Causal Convergence (\ccv)}

\ccv{} is stronger than \wcc{}. Therefore, $\ccv$ excludes all the bad patterns 
of \wcc{}, given in \lem{lem:causalreg}.
$\ccv$ also excludes one additional bad pattern, defined in terms of a 
\emph{\conflict{} relation}.

The \emph{\conflict{} relation} is a relation on write operations
(which write to the same variable). 
It is used for the bad pattern $\bpbpcycliccf$ of $\ccv$, defined 
in \tab{tab:allbps}.
Intuitively, for two write operations $\wop_1$ and $\wop_2$,
we have $\wop_1 \ltcf \wop_2$ if some \site{} saw both writes, 
and decided to order $\wop_1$ before $\wop_2$ (so decided to return 
the value written by $\wop_2$).

\begin{example}
History~(\ref{fig:sccnotccv}) contains bad pattern \bpbpcycliccf.
The $\mkwrite{\var}{1}$ operation $\wop_1$ is causally related to 
the $\mkread{\var}{2}$ operation, so we have $\wop_1 \ltcf \wop_2$,
where $\wop_2$ is the $\mkwrite{\var}{2}$ operation.
Symmetrically, $\wop_2 \ltcf \wop_1$, and we obtain a cycle.
On the other hand, History~(\ref{fig:sccnotccv}) does not contain 
any of the bad patterns of $\scc$.
\end{example}

\begin{example}
History~(\ref{fig:wccnotccvnotscc})
contains bad pattern \bpbpcycliccf.
The cycle is on the two writes operations 
$\mkwrite{\var}{1}$ and $\mkwrite{\var}{2}$.
\end{example}

The formal definition of the \conflict{} relation is the following.

\begin{definition}
We define the \emph{\conflict{} relation} $\cf \subseteq \ops \times \ops$
to be the smallest relation such that: for all $\var \in \Var$, and 
$\dat_1 \neq \dat_2 \in \Nats$ and operations $\wop_1,\wop_2,\rop_2$, 
if 
    \begin{itemize}
    \item $\wop_1  \ltpropco  \rop_2$, 
    \item $\getlabel(\wop_1) = \mkwrite{\var}{\dat_1}$, 
    \item $\getlabel(\wop_2) = \mkwrite{\var}{\dat_2}$, and
    \item $\getlabel(\rop_2) = \mkread{\var}{\dat_2}$,
    \end{itemize} 
    then $\wop_1  \ltcf \wop_2$.
\end{definition}

We obtain the following lemma for the bad patterns of $\ccv$.

\begin{restatable}{lemma}{ccvbadpatterns}
\label{lem:ccvbadpatterns}
A \differentiated{} \history{} $\hist$ is \ccvt{} 
with respect to $\speckvs$ if and only if 
$\hist$ is \wcct{} and does not contain the following bad pattern:
\bpbpcycliccf.
\end{restatable}

\subsection{Characterizing Causal Memory (\scc)}

\scc{} is stronger than \wcc{}. Therefore, $\scc$ excludes all the bad patterns 
of \wcc{}, given in \lem{lem:causalreg}.
$\scc$ also excludes two additional bad patterns, defined in terms of a 
\emph{\happenedbefore{} relation}.

The \emph{\happenedbefore{} relation for an operation $\op \in \ops$}
intuitively represents the minimal constraints that must hold 
in a sequence containing all operations before $\op$, 
on the \site{} of $\op$.

\begin{example}
History~(\ref{fig:ccvnotscc}) contains bad pattern \bpbpchangerfi.
Indeed, we have 
$\mkwrite{z}{1} \ltpo \mkwrite{x}{1} \ltrel{\hb{\rop_2}}
\mkwrite{x}{2} \ltpo \mkread{z}{0}$, 
where $\rop_2$ is the $\mkread{x}{2}$ operation.
The edge $\mkwrite{x}{1} \ltrel{\hb{\rop_2}}
\mkwrite{x}{2}$ is induced by the fact that 
$\mkwrite{x}{1} \ltco \mkread{x}{2}$.
\end{example}

The formal definition of $\hb{\op}$ is the following.

\begin{definition}
Given $\op \in \ops$, 
we define the \emph{\happenedbefore{} relation for $\op$}, noted 
$\hb{\op} \subseteq \ops \times \ops$,
to be the smallest relation such that:\\
\begin{itemize}
\item $\projrel{\propco}{\causaldep{\op}} \subseteq \hb{\op}$, and
\item $\hb{\op}$ is transitive, and
\item 
  for $\var \in \Var$, and $\dat_1 \neq \dat_2 \in \Nats$, if 
  \begin{itemize}
  \item $\wop_1  \ltrel{\hb{\op}}  \rop_2$,
  \item $\rop_2 \leqpo \op$, 
  \item $\getlabel(\wop_1) = \mkwrite{\var}{\dat_1}$, 
  \item $\getlabel(\wop_2) = \mkwrite{\var}{\dat_2}$, and 
  \item $\getlabel(\rop_2) = \mkread{\var}{\dat_2}$,
  \end{itemize}
 then 
  $\wop_1   \ltrel{\hb{\op}} \wop_2$.
\end{itemize}
\end{definition}

There are two main differences with the \conflict{} relation $\cf$.
First, $\cf$ is not defined inductively in terms of itself, but only in terms
of the relation $\propco$.
Second, in the \happenedbefore{} relation for $\op$, in order to 
add an edge between write operations, there is the constraint that 
$\rop_2 \leqpo \op$, while in the definition of the \conflict{} relation,
$\rop_2$ is an arbitrary read operation.
These differences make the \conflict{} and \happenedbefore{} relations
not comparable (with respect to set inclusion).

We obtain the following lemma for the bad patterns of $\scc$
(see \tab{tab:allbps} for the bad patterns' definitions).

\begin{restatable}{lemma}{strongbadpatterns}
\label{lem:strongcausalreg}
A \differentiated{} \history{} $\hist$ is \scct{}
with respect to $\speckvs$ if and only if $\hist$ is 
\wcct{} and does \emph{not} contain the following bad patterns:
\bpbpchangerfi, \bpbpcyclichb.
\end{restatable}

\tab{table:bptable} gives, for each consistency criterion, 
the bad patterns which are excluded by the criterion.
  
\section{Single History Consistency under DI} 

\label{sec:poly}

The lemmas of the previous sections entail a polynomial-time algorithm for 
checking whether a given \differentiated{} history is causally consistent
(for any definition).
This contrasts with the fact that checking whether an arbitrary 
history is causally consistent is \npcomplete.

The algorithm first constructs the relations which are used in the definitions 
of the bad patterns, and then checks for the presence of the bad patterns in 
the given history.

\begin{lemma}
Let $\hist = (\ops,\po,\getlabel)$ be a \differentiated{} history.
Computing the relations $\rf$,$\propco$,$\cf$, and $\hb{\op}$ for
$\op \in \ops$ can be done in polynomial time
($O(n^5)$ where $n$ is the number of operations in $\hist$).
\end{lemma}

\begin{proof}
We show this for the relation $\hb{\op}$ (for some $\op \in \ops$).
The same holds for the other relations.
The relation $\hb{\op}$
can be computed inductively using its fixpoint definition.
At each iteration of the fixpoint computation, 
we add one edge between operations in $\ops$. 
Thus, there are at most $n^2$ iterations.

Each iteration takes $O(n^3)$ time. For instance,
an iteration of computation of $\hb{\op}$ can consist in 
adding an edge by transitivity, which takes $O(n^3)$ time.

Thus the whole computation of $\hb{\op}$ takes $O(n^5)$ time.
\end{proof}

Once the relations are computed, we can check for the presence of 
bad patterns in polynomial time.

\begin{theorem} 
Let $\hist = (\ops,\po,\getlabel)$ be a \differentiated{} history. 
Checking whether 
$\hist$ is \wcct{} (\resp \scct, \resp \ccvt) can be done in 
polynomial time 
($O(n^5)$ where $n$ is the number of operations in $\hist$).
\end{theorem}

\begin{proof}

First, we compute the relations $\rf,\propco,\cf,\hb{\op}$ (for all $\op \in \ops$),
in time $O(n^5)$.
The presence of bad patterns can be checked in polynomial time. 
For instance, for bad pattern $\bpbpcycliccf$, we need to find a cycle in the 
relation $\cf$.
Detecting the presence of a cycle in a relation takes $O(n^2)$.

The complexity of the algorithm thus comes from computing the relations,
which is $O(n^5)$.
\end{proof}
 
In the next two sections, we only consider criterion \wcc.

\section{Reduction to Control-State Reachability under Data Independence}
\label{sec:safety}

The undecidability proof of \thm{thm:causalundeckvs} uses an implementation
which is not \dataindependent.
Therefore, it does not apply when we consider only \dataindependent{}
implementations.
In fact, 
we show that for \kvs{} implementations which are \dataindependent, 
there is an effective reduction from checking \wcc{} to a 
non-reachability problem.

Using the characterization of \sect{sec:badpatterns}, we define an 
\emph{\observer} $\obscc$ that looks for the bad patterns leading 
to non-\wcc.
More precisely, our goal is to define $\obscc$ as a 
\emph{\registerautomaton} such that
(by an abuse of notation, the set of executions recognized by $\obscc$ is
also denoted $\obscc$):
\[
  \lib \text{ is \wcct{} with respect to } \speckvs \iff 
  \lib \cap \obscc = \emptyset
\]
where $\lib$ is any \dataindependent{} \library{}.

Ultimately, we exploit in \sect{sec:decidability} this reduction to prove that 
checking \wcc{} for (\finitestate) \dataindependent{} 
implementations, with respect to the \kvs{} specification, is decidable.

\begin{figure}[htb]
 
\scriptsize
\begin{tikzpicture}[xscale=4,yscale=3.3,->,>=latex]
\node[state,initial,initial text=,initial distance=0.2,initial above] 
  (X) at (0,0) { $q_0$ };
\node[state, accepting] (ERR) at ($(X) + (0.7,0)$) { $\errorstate$ };

\draw (X) -- 
  node[above]  { $\siteid,\mktrueaction{\readmeth}{\var}{\firstdat}$ } (ERR);

\node[state,initial,initial text=,initial distance=0.2,initial above] 
  (A) at (0,-0.6) { $q_1$ };
\node[state,ellipse,text width=1.5cm,align=center,inner sep=0pt] (B) at ($(A) + (0.8,0)$) { 
  $\clink$ $[\thedat\leftarrow \thirddat]$ };
\node[state,ellipse,text width=1.5cm,align=center,inner sep=0pt] (D) at ($(B) + (1,0)$) {
  $\clink$ $[\thedat\leftarrow \fourthdat]$ };
\node[state,accepting] (E) at ($(D) + (0,0.6)$) { $\errorstate'$ };

\draw (A) -- 
  node[above,text width=1.6cm]  { 
    $\siteid,\mktrueaction{\writemeth}{\var,\firstdat}{\unit}$ 
    $\siteid,\mkread{\var}{\firstdat}$
  } 
  node[below,text width=1.6cm,xshift=0.1cm]  { 
    \textcolor{color1}{$\assign{\regvar'}{\var}$}
    \textcolor{color2}{$\assign{\regvar}{\var}$}
  \textcolor{color3}{$\assign{\regsite}{\siteid}$}
  }  (B);

\draw (B) to  
    node[above, text width=1.6cm] { 
        $\siteid,\mktrueaction{\writemeth}{\var,\seconddat}{\unit}$
    }
    node[below, text width=1.6cm] {
        \textcolor{color1}{$\equality{\regvar'}{\var}$}
        \textcolor{color2}{$\assign{\regvar}{\var}$}
        \textcolor{color3}{$\equality{\regsite}{\siteid}$}
    } (D);
  
\draw (D) to  
  node[left, text width=1.6cm] {
    $\siteid,\mktrueaction{\readmeth}{\var}{\firstdat}$
    \textcolor{color1}{$\equality{\regvar'}{\var}$}
    \textcolor{color3}{$\equality{\regsite}{\siteid}$}
} (E);

\node[state,initial,initial text=,initial distance=0.2,initial above] 
  (AA) at (0,-1.2) { $q_2$ };
\node[state,ellipse,text width=1.5cm,align=center,inner sep=0pt] (BB) at ($(AA) + (0.8,0)$) { 
  $\clink$ $[\thedat\leftarrow \seconddat]$ };
\node[state,accepting] (EE) at ($(BB) + (0.8,0)$) { $\errorstate''$ };

\draw (AA) -- 
  node[above,text width=1.6cm]  { 
    $\siteid,\mktrueaction{\writemeth}{\var,\firstdat}{\unit}$ 
    $\siteid,\mkread{\var}{\firstdat}$
  } 
  node[below,text width=1.6cm,xshift=0.1cm]  { 
    \textcolor{color1}{$\assign{\regvar'}{\var}$}
    \textcolor{color2}{$\assign{\regvar}{\var}$}
  \textcolor{color3}{$\assign{\regsite}{\siteid}$}
  }  (BB);

\draw (BB) to  
    node[above] { 
        $\siteid,\mkread{\var}{0}$
    }
    node[below, text width=1.6cm] {
        \textcolor{color1}{$\equality{\regvar'}{\var}$}
        \textcolor{color2}{$\assign{\regvar}{\var}$}
        \textcolor{color3}{$\equality{\regsite}{\siteid}$}
    } (EE);

\end{tikzpicture}

\caption{
  The \observer{} $\obscc$, finding bad patterns for \wcc{}
  with respect to $\speckvs$. The first branch looks
  for bad pattern \bpbprf. 
  The second branch looks for bad pattern \bpbpco.
  The third branch looks for bad pattern \bpbprfi.\\
  Each state has a self-loop with any symbol
  containing value $\fifthdat$, which we do not represent.
  Two labels  $\tid,\mktrueaction{\meth}{\argv}{\rv}$
  above a transition denote two different transitions.\\
}
\label{fig:obscc}
\vspace{-9mm}
\end{figure}
 
\begin{figure}[htb]

\centering

\begin{tikzpicture}[xscale=4,yscale=4,->,>=latex]
\node[state,accepting] 
    (B) at ($(0,0)$) { $q_b$ };
\node[state,initial text=,initial distance=0.2,initial above, accepting] (C) at ($(B) + (-1,0)$) { $q_a$ };
  
\draw (B) to [bend left=35] 
    node[below, text width=1.6cm] { 
  $\siteid,\mktrueaction{\readmeth}{\var}{\thedat}$
  \textcolor{color3}{$\assign{\regsite}{\siteid}$}
  \textcolor{color2}{$\equality{\regvar}{\var}$}
} (C);

\draw (C) to [bend left=35] 
        node[above, text width=1.6cm] { 
  $\siteid,\mktrueaction{\writemeth}{\var,\thedat}{\unit}$
  \textcolor{color2}{$\assign{\regvar}{\var}$}
  \textcolor{color3}{$\equality{\regsite}{\siteid}$}
} (B);
  
\end{tikzpicture}
\caption{ 
  The \registerautomaton{} $\clink$, which recognizes causality chains by
  following links in the $\po \cup \rf$ relations. Both states are final.\\
  Each state has a self-loop with any symbol
  containing value $\fifthdat$, which we do not represent.
}
\label{fig:clink}
\vspace{-5mm}
\end{figure}

\subsection{Register Automata}

\emph{Register automata}~\citep{DBLP:conf/concur/BouyerPT01}
have a finite number of registers in which they can store 
values (such as the \siteidentifier{}, the name of a variable in the \kvs{}, 
or the data value stored at a particular variable), 
and test equality on stored registers. 

We describe the syntax of register automata that we use in the figures.
The label 
$\siteid,\mktrueaction{\writemeth}{\var,\firstdat}{\unit}$
above the transition going from $q_1$ in \fig{fig:obscc}
is a form of pattern matching. 
If the automaton reads a tuple 
$(\siteid_0, \mkwrite{\var_0}{\firstdat})$, 
for some $\siteid_0 \in \Tid$, $\var_0 \in \Var$,
then the variables $\siteid$, $\var$ are bound 
respectively to $\siteid_0$, $\var_0$.

If this transition, or another transition, gets executed afterwards,
the variables $\siteid$, $\var$ can be bound to other values.
These variables are only local to a specific execution of the transition.

The instruction \textcolor{color1}{$\assign{\regvar'}{\var}$}, on 
this same transition, is used to store the value $\var_0$ 
which was bound by $\var$
in register $\regvar'$.
This ensures that the operations 
$\mktrueaction{\writemeth}{\var,\seconddat}{\unit}$ and 
$\mktrueaction{\readmeth}{\var}{\firstdat}$ that come later
use the same variable $\var_0$, thanks to the equality check
\textcolor{color1}{$\equality{\regvar'}{\var}$}.

Note that, in \fig{fig:clink}, $\dat_0$ is not a binding variable as 
$\siteid$ and $\var$, but is instead a constant which is fixed to different values 
in \fig{fig:obscc}.

\subsection{Reduction}

$\obscc$ 
(see \fig{fig:obscc}) 
is composed of three parts.
The first part recognizes executions which contain 
bad pattern $\bpbprf$, \ie which have a $\readmeth$ operation 
with no corresponding $\writemeth$. The second part recognizes
executions containing bad pattern $\bpbpco$,
composed of operations 
$\wop_1$, $\wop_2$, and $\rop_1$, such that
$\wop_1 \ltpropco \wop_2 \ltpropco \rop_1$, $\wop_1 \ltrf \rop_1$, 
and $\getvar{\wop_1} = \getvar{\wop_2}$.
The third part recognizes executions containing bad pattern 
$\bpbprfi$, where a write on some variable $\var$, writing 
a non-initial value, is causally related to
a $\mkread{\var}{0}$.

To track the relation $\propco$, 
$\obscc$ uses another \registerautomaton{}, called $\clink$ 
(see \fig{fig:clink}), 
which recognizes unbounded chains in the $\propco$ relation.

The registers of $\obscc$ store \site{} ids and the 
variables' names of the \kvs{}
(we use registers because the number of \sites{} and variables  
in the causality links can be arbitrary).

By data independence, $\obscc$ only needs to use a bounded number of 
values.
For instance, for the second branch recognizing bad pattern \bpbpco, 
it uses value $\firstdat$ for operations $\wop_1$ and $\rop_1$,
and value $\seconddat$ for operation $\wop_2$.
It uses value $\thirddat$ for the causal link between $\wop_1$ and $\wop_2$,
and value $\fourthdat$ for the causal link between $\wop_2$ and $\rop_1$.
Finally, it uses the value $\fifthdat \in \Nats$ for all \action{s} of the 
\execution{} which are 
not part of the bad pattern. As a result, it can self-loop with any symbol
containing value $\fifthdat$.
We do not represent these self-loops 
to keep the figure simple.

We prove in \thm{th:obskvseq} that any execution recognized
by $\obscc$ is not \wcct{}.
Reciprocally, we prove that for any differentiated \execution{} of an 
implementation $\lib$ which is 
not \wcct{}, we can rename the values to obtain
an execution with $5$ values recognized by $\obscc$.
By data independence of $\lib$, the renamed execution is 
still an execution of $\lib$.

\vspace{-2mm}
\begin{remark}
Note here that the observer $\obscc$ does not look for bad pattern 
\bpbpcyclicco. We show in \thm{th:obskvseq} that, since the 
implementation is a prefix-closed set of executions, 
it suffices to look for bad pattern \bpbprf{} to recognize
bad pattern \bpbpcyclicco.
\vspace{-2mm}
\end{remark}

\begin{restatable}{theorem}{reduction}
\label{th:obskvseq}
Let $\lib$ be a \dataindependent{} implementation.
 $\lib$ is \wcct{} with respect to $\speckvs$ if and only 
if $\lib \cap \obscc = \emptyset$.
\end{restatable}

This result allows to reuse any tool or
technique that can solve reachability 
(in the system composed of the \implementation{} and the \observer{} 
$\obscc$)
for the verification of \wcc{} (with respect to the \kvs).

 \section{Decidability under Data Independence}
\label{sec:decidability}

In this section, we exploit this reduction to obtain decidability 
for the verification of \wcc{} with respect to the \kvs{}.

We consider a class of implementations $\theclass$ for which reachability
is decidable, making \wcc{} decidable ($\EXPSPACE$-complete).
We consider implementations $\lib$ which are distributed over a finite 
number of \sites{}. The \sites{} run asynchronously, and communicate 
by sending messages using peer-to-peer communication channels.

Moreover, we assume that the number of variables from the 
$\kvs$ that the implementation $\lib$ seeks to implement is fixed and 
finite as well. However, we do not bound the domain of values that 
the variables can store.

Each \site{} is a finite-state machine with registers
that can store values in $\Nats$ for the contents of the variables
in the \kvs{}.
Registers can be assigned using instructions of the form
$x \Coloneqq y$ and $x \Coloneqq \dat$ where $x,y$ are registers, 
and $\dat \in \Nats$ is a value (a constant, or a value provided
as an argument of a method).
Values can also be sent through the network to other 
\sites{}, and returned by a method. 
We make no assumption on the network: the
peer-to-peer channels  are unbounded and unordered.

Any implementation in $\theclass$ is thus necessarily \dataindependent{}
by construction, as the contents of the registers storing the values
which are written are never used in conditionals.

The \observer{} $\obscc$ we constructed only needs $5$ 
values to detect all \wcc{} violations. For this reason,
when modeling an \implementation{} in $\theclass$, 
there is no need to model the whole range of natural numbers $\Nats$, 
but only $5$ values.
With this in mind, any implementation in $\theclass$ can be modelled by a 
Vector Addition System with States (VASS)~\citep{karp1969parallel,hopcroft1979reachability}, or a 
Petri Net~\citep{petri1962kommunikation,petrinets}. 
The local state of each \site{} is encoded in the 
state of the VASS, and the content of the peer-to-peer channels is
encoded in the counters of the VASS. Each counter counts how many messages
there are of a particular kind in a particular peer-to-peer channel.
There exist similar encodings in the literature~\citep{DBLP:conf/popl/BouajjaniEH14}.

Then, since the number of \sites{} and the number of variables is 
bounded, we can get rid of the registers in the \registerautomaton{}
of the observer $\obscc$, and obtain a (normal) finite automaton.
We then need to solve control-state reachability in the system composed of the
VASS and the \observer{} $\obscc$ to solve \wcc{} (according 
to \thm{th:obskvseq}).
Since VASS are closed under composition with finite automata, 
and control-state reachability is $\EXPSPACE$-complete for VASS, we get the 
$\EXPSPACE$ upper bound for the verification of \wcc{}
(for the \kvs{}).

The $\EXPSPACE$ lower bound follows from: 
(1) State reachability in class $\theclass$ is 
$\EXPSPACE$-complete, 
equivalent to control-state reachability in 
VASS~\citep{DBLP:conf/fsttcs/AtigBT08}. Intuitively, 
a VASS can be modelled by an implementation in $\theclass$, 
by using the unbounded unordered channels to simulate the counters
of the VASS. (Similar to the reduction from $\theclass$ to VASS outlined
above, but in the opposite direction.)
(2) Checking reachability can be reduced to verifying \wcc{}.
Given an \library{} $\lib$ in $\theclass$, and a state $q$,
knowing whether $q$ is reachable can be reduced to 
checking whether a new \library{} $\lib'$ is not causally consistent.
$\lib'$ is an \library{} which simulates $\lib$, and 
produces only causally consistent executions; if it reaches state 
$q$, it artificially produces a non-causally consistent execution,
for instance by returning wrong values to read requests.

\begin{restatable}{theorem}{decidability}
\label{th:decidability}
Let $\lib$ be a \dataindependent{} \library{} in $\theclass$.
Verifying \wcc{} of $\lib$ with respect to the \kvs{} is 
$\EXPSPACE$-complete (in the size of the VASS of $\lib$).
\end{restatable}

\section{Related Work}
\label{sec:alt}

\citet{DBLP:journals/corr/abs-1302-5161} studied the complexity of verifying 
PRAM consistency (also called FIFO consistency) for one history. They proved
that the problem is $\NP$-complete.
For \differentiated{} \execution{s}, 
they provided a polynomial-time algorithm.

Independently, \citet{roland} showed that checking causal consistency 
($\scc$ definition) of one history is an $\NP$-complete problem.
They proved that checking consistency for one history for 
any criterion stronger than {\tt SLOW} 
consistency and weaker than sequential consistency is $\NP$-complete, 
where {\tt SLOW} consistency ensures that for each 
variable $\var$, and for each \site{} $\tid$, the reads of $\tid$ on 
variable $\var$ can be
explained by ordering all the writes to $\var$ while respecting the
program order.
This range covers $\scc$, 
but does not cover $\wcc$~(see \fig{fig:wccnotccvnotscc} for a history 
which is \wcct{} but not {\tt SLOW}). It is not clear whether 
this range covers $\ccv$.
To prove $\NP$-hardness, they used a reduction from the \npcomplete{} $\sat$ 
problem.
We show that their encoding can be reused to show $\NP$-hardness for 
checking whether a history is $\wcct$ (\resp $\ccvt$) with respect to the 
\kvs{} specification.

Concerning verification, 
we are not aware of any work studying the decidability or complexity 
of checking whether all executions of an \implementation{} are causally 
consistent.
There have been works on studying the problem for other
criteria such as linearizability~\citep{journals/toplas/HerlihyW90} 
or eventual consistency~\citep{DBLP:conf/sosp/TerryTPDSH95}.
In particular, it was shown that checking linearizability is an 
$\EXPSPACE$-complete problem when the number of \sites{} is 
bounded~\citep{journals/iandc/AlurMP00,netys-lin}.
Eventual consistency has been shown to be decidable~\citep{DBLP:conf/popl/BouajjaniEH14}.
Sequential consistency was shown to be 
undecidable~\citep{journals/iandc/AlurMP00}.

The approach we adopted to obtain decidability of causal consistency by
defining bad patterns for particular specifications has been used recently in the context of
linearizability~\citep{conf/concur/HenzingerSV13,conf/tacas/AbdullaHHJR13,linbadpatterns}. However, the bad patterns for linearizability do not 
transfer to causal consistency. Even from a technical point of view,  the results introduced for linearizability cannot be used in our case.
One reason for this is that, in causal consistency, there is the additional
difficulty that the causal order is existentially quantified, while the 
happens-before relation in linearizability is fixed (by a global clock).

\citet{lesani2016chapar} investigate mechanized proofs of causal consistency 
using the theorem prover Coq. This approach does not lead to full automation 
however, e.g., by reduction to assertion checking.

 \section{Conclusion}

We have shown that verifying causal consistency is hard, even undecidable, in general:
verifying whether one single execution satisfies causal consistency is NP-hard, and verifying if all the executions of an implementation are causally consistent is undecidable. These results are not due to the complexity of the implementations nor of the specifications: they hold even for finite-state implementations and specifications. They hold already when the specification corresponds to the simple read-write memory
abstraction. The undecidability result contrasts with known decidability results for other correctness criteria such as linearizability \cite{journals/iandc/AlurMP00} and eventual consistency \cite{DBLP:conf/popl/BouajjaniEH14}. 

Fortunately, for the read-write memory abstraction, an important and widely used abstraction in the setting of distributed systems, we show that, when  implementations are data-independent, which is the case in practice, the verification problems we consider become tractable. This is based on the very fact that data independence allows to restrict our attention to differentiated executions, where the written values are unique, which allows 
to deterministically establish the read-from relation along executions. 
This is crucial for characterizing by means of a finite number of bad patterns the set of all violations to causal consistency, which is the key to our complexity and decidability results for the read-write memory.

First, using this characterization we show that the problem of verifying the correctness of a single execution is polynomial-time in this case. This is important for building efficient and scalable testing and bug detection algorithms. Moreover, we provide an algorithmic approach for verifying causal consistency (w.r.t. the read-write memory abstraction) based on an effective reduction of this problem to a state reachability problem (or invariant checking problem) in the class of programs used for the implementation. Regardless from the decidability issue, this reduction holds for an unbounded number of sites (in the implementation), and an unbounded number of variables (in the read-write memory). In fact, it establishes a fundamental link between these two problems and allows to use all existing (and future) program verification methods and tools for the verification of causal consistency. In addition, when the number of sites is bounded, this reduction provides a decidability result for verifying causal consistency concerning a significant class of implementations: finite-control machines (one per site) with data registers (over an unrestricted data domain, with only assignment operations and equality testing), communicating through unbounded unordered channels. As far as we know, this is the first work that establishes complexity and (un)decidability results for the verification of causal consistency. 

All our results hold for the three existing variants of causal consistency \wcc, \scc, and \ccv, except for the reduction to state reachability and the derived decidability result that we give in this paper for \wcc\, only. For the other two criteria, building observers detecting their corresponding bad patterns is not trivial in general, when there is no assumption on the number of sites and the number of variables (in the read-write memory). We still do not know if this can be done using the same class of state-machines we use in this paper for the observers. However, this can be done if these two parameters are bounded. In this case, we obtain a decidability result that holds for the same class of implementations as for \wcc, but this time for a fixed number of variables in the read-write memory. This is still interesting since when data independence is not assumed, verifying causal consistency is undecidable for the read-write memory even when the number of sites is fixed, the number of variables is fixed, and the data domain is finite. We omit these results in this paper. 

Finally, let us mention that in this paper we have considered correctness criteria that correspond basically to safety requirements. Except for \ccv, convergence, meaning eventual agreement between the sites on their execution orders of non-causally dependent operations is not guaranteed. In fact, these criteria can be strengthened with a liveness part requiring the convergence property. Then, it is possible to extend our approach to handle the new criteria following the approach adopted in \citep{DBLP:conf/popl/BouajjaniEH14} for eventual consistency. Verifying correctness in this case can be reduced to a repeated reachability problem, and model-checking algorithms can be used to solve it. 

For future work, it would be very interesting to identify a class of specifications for which our approach is systematically applicable, i.e., for which there is a procedure producing the complete set of bad patterns and their corresponding observers in a decidable class of state machines. 

\section*{Acknowledgments}

This work is supported in part by the European Research Council (ERC) under the European
Union’s Horizon 2020 research and innovation programme (grant agreement No 678177), and by an EPFL-Inria postdoctoral grant. 
\newpage

\bibliographystyle{abbrvnat}
\bibliography{references}

\iftoggle{long}{
\newpage
\appendix

\newpage
\section{Differentiated Histories}

We detail here the results of \sect{sec:complete}.
We identify here conditions under which is it enough 
to check the causal consistency of only a subset $\hists' \subseteq \hists$
of histories, while ensuring that all histories $\hists$ are causally
consistent.

We then use this notion to prove that it is enough to check 
causal consistency with respect to the \kvs{} for histories which use 
distinct $\writemeth$ values.

\subsection{Reduction}

Let $\rel$ be a relation over labeled posets.
A subset $\getcomplete{\hists} \subseteq \hists$ is said to be 
\emph{complete} for $\hists$ if:
for all $\hist \in \hists$, 
there exists $\getcomplete{\hist} \in \getcomplete{\hists}$,
$\getcomplete{\hist} \ltr \hist$.

When checking whether a set of histories $\hists_1$ is a subset of
a set $\hists_2$ which is upward-closed with respect to $\rel$, 
it is sufficient to check the inclusion for a complete 
subset $\getcomplete{\hists_1} \subseteq \hists_1$.

\begin{lemma}[Complete Sets of Histories]
\label{lem:completeness}
Let $\rel$ be a relation, and 
let $\getcomplete{\hists_1} \subseteq \hists_1$ a complete set of 
labeled posets using relation $\rel$. Let $\hists_2$ be a set of 
histories which is upward-closed with respect to $\rel$,
we have $\hists_1 \subseteq \hists_2$ if and only if 
$\getcomplete{\hists_1} \subseteq \hists_2$.
\end{lemma}

\begin{proof}
$(\Rightarrow)$ Holds because $\getcomplete{\hists_1}$ is
a subset of $\hists_1$. 

$(\Leftarrow)$
Assume $\getcomplete{\hists_1} \subseteq \hists_2$ and let 
$\hist_1 \in \hists_1$. Since $\getcomplete{\hists_1}$ is complete for 
$\hists_1$, we know there exists 
$\getcomplete{\hist_1} \in \getcomplete{\hists_1}$ such that 
$\getcomplete{\hist_1} \ltr \hist_1$.
By assumption, $\getcomplete{\hist_1} \in \hists_2$. Finally, 
since $\hists_2$ is upward-closed, $\hist_1 \in \hists_2$.

\end{proof}

Given a function $\renaming: \Domain \rightarrow \Domain$ and a 
tuple $\act \in \Meth \times \Domain$ or $\act \in \Meth \times \Domain 
\times \Domain$, we denote by 
$\applyrenaming{\act}{\renaming}$ the tuple where each occurrence of 
$\dat \in \Domain$ has been replaced by $\renaming(\dat)$.
We lift the notation to $\Meth \times \Domain$ and 
$\Meth \times \Domain \times \Domain$ labeled posets by changing the labels of 
the elements.
We lift the notation to sets of labeled posets in a point-wise manner.

Let $\renamings \subseteq \Domain \rightarrow \Domain$ be a set of 
functions. 

\begin{definition}
$\spec$ is \emph{$\renamings$-invariant} if
for all $\renaming \in \renamings$, 
$\applyrenaming{\spec}{\renaming} \subseteq \spec$.
\end{definition}

We define a relation $\smallerF$ as follows:
$\hist_1 \smallerF \hist_2 \iff \exists \renaming \in \renamings.\ 
\hist_2 = \applyrenaming{\hist_1}{\renaming}$.
Let $\spec$ be a specification.
We denote by $\wccfor{\spec}$ (\resp $\sccfor{\spec}$, $\ccvfor{\spec}$)
the set of \histories{} which are \wcct{} (\resp \scct, \ccvt) 
with respect to $\spec$. We show that for any specification which is 
$\renamings$-invariant, the set $\wccfor{\spec}$ 
(\resp  $\sccfor{\spec}$, $\ccvfor{\spec}$) is upward-closed
with respect to the relation $\smallerF$.

\begin{lemma}
\label{lem:causalup}
Let $\renamings \subseteq \Domain \rightarrow \Domain$ be a set of 
functions.
Let $\spec$ be a specification which is $\renamings$-invariant.
The set $\wccfor{\spec}$ (\resp  $\sccfor{\spec}$, $\ccvfor{\spec}$)
is upward-closed with respect to $\smallerF$.
\end{lemma}

\begin{proof}
We show the proof for $\wccfor{\spec}$, but the proof can be directly
adapted to the sets $\sccfor{\spec}$ and $\ccvfor{\spec}$.
Let $\hist = (\ops,<,\getlabel) \in \wccfor{\spec}$ and 
$\hist'$ such that $\hist \smallerF \hist'$.
We know there exists $\renaming \in \renamings$ such that 
$\hist' = \applyrenaming{\hist}{\renaming}$. As a result, 
$\hist$ and $\hist'$ have the same underlying poset $(\ops,<)$. 
Let $\hist' = (\ops,<,\getlabel')$.

Let $\op \in \ops$.
By axiom \axwcc{}, we know there is $\loc \in \spec$ such that
$\projectrv{\causalpast{\op}}{\op} \weaker \loc$,
where $\causalpast{\op} = (\causaldep{\op},\co,\getlabel)$.

We have 
$(\causaldep{\op},\co,\getlabel') = \applyrenaming{\causalpast{\op}}{\renaming}$.
Therefore, by defining $\loc' = \applyrenaming{\loc}{\renaming}$, 
we have $(\causaldep{\op},\co,\getlabel') \weaker \loc'$.
Since $\spec$ is $\renamings$-invariant, 
$\loc' \in \spec$,
and axiom $\axwcc$ holds for history $\hist'$.

We conclude that $\hist'$ is in $\wccfor{\spec}$.

\end{proof}

Lemmas~\ref{lem:completeness} and \ref{lem:causalup} combined show
it is enough to check $\wcc$ (\resp $\scc$, $\ccv$) for a complete subset of 
histories, to obtain $\wcc$ (\resp $\scc$, $\ccv$) for all histories.

\begin{corollary}
\label{coro:complete}
Let $\renamings \subseteq \Domain \rightarrow \Domain$ be a set of 
functions. Let $\hists$ be a set of histories, and $\hists' \subseteq \hists$
a complete set of histories, using relation $\smallerF$.
Let $\spec$ be specification which is $\renamings$-invariant. Then,
$\hists$ is \wcct{} (\resp \scct{}, \ccvt{}) with respect to $\spec$ 
if and only if 
$\hists'$ is \wcct{} (\resp \scct{}, \ccvt{}) with respect to $\spec$.
\end{corollary}

\begin{proof} Using Lemmas~\ref{lem:completeness} and \ref{lem:causalup},
we know 
$\hists' \subseteq \wccfor{\spec}$ if and only if
$\hists \subseteq \wccfor{\spec}$. The same applies for the sets 
$\sccfor{\spec}$ and $\ccvfor{\spec}$.
\end{proof}

\subsection{Data Independence}

\label{apd:dataind}

We show how to apply the notion of completeness for the 
verification of the \kvs{}.
Most implementations used in practice are 
\emph{data independent}~\citep{conf/tacas/AbdullaHHJR13}, \ie their 
behaviors do not depend on the particular data values which are stored at 
a particular variable. Under this assumption, we show it is enough to verify 
causal consistency for histories which do not write twice the same
value on the same variable, and which never writes the initial 
value $0$, called 
\emph{\differentiated} histories.

Formally, a \history{} $(\ops,\po,\getlabel)$ is said to 
be \emph{\differentiated} if:
\begin{itemize}
\item for all $\op_1 \neq \op_2$, $\var \in \Var$,
$\getlabel(\op_1) = \mkwrite{\var}{\dat_1}$ and 
$\getlabel(\op_2) = \mkwrite{\var}{\dat_2}$ and
implies that $\dat_1 \neq \dat_2$, and
\item for all $\var \in \Var$,
$\hist$ does not contain $\mkwrite{\var}{0}$ operation.
\end{itemize}
Let $\hists$ be a set of labeled posets.
We denote by $\getdiff{\hists}$ the subset of 
\differentiated{} \histories{} of $\hists$.
 
A \emph{\datarenaming} is a function from $\Domain$ to $\Domain$ which 
modifies the data values of operations. More precisely, we can
build a \datarenaming{} $\renaming$ from any function 
$\renaming_0: \Nats \rightarrow \Nats$ in the following way.

Remember that for the \kvs{}, $\Domain$ is the set 
$(\Var \times \Nats) \uplus \Var \uplus \Nats \uplus \set{\unit}$, and 
we define:
\begin{itemize}
\item 
  $\renaming(\var,\argv) = (\var,\renaming_0(\argv))$ for 
  $\var \in \Var, \argv  \in \Nats$, 
\item
  $\renaming(\var) = \var$ for $\var \in \Var$,
\item $\renaming(\rv) = \renaming_0(\rv)$ for $\rv \in \Nats$,
\item
  $\renaming(\unit) = \unit$.
\end{itemize}

Let $\datarens$ be the set of all \datarenaming{s}. 

\begin{definition}
A set of histories is \emph{\dataindependent} if 
the subset $\getdiff{\hists} \subseteq \hists$ is complete using 
relation $\smallerFData$, and $\hists$ is $\datarens$-invariant.
\end{definition}

\begin{remark}
This definition corresponds to the other definition of data independence
we have in \sect{sec:complete}.
\end{remark}

Using the following lemma and then applying 
\coro{coro:complete}, we obtain that it is
enough to check causal consistency for histories which are \differentiated{}
(to ensure the causal consistency of all histories).

\begin{lemma}
\label{lem:datainvariant}
The $\speckvs$ specification is $\datarens$-invariant.
\end{lemma}

\begin{proof}
Let $\renaming \in \datarens$ and let $\seqposet \in \speckvs$.
We can see that $\applyrenaming{\seqposet}{\renaming} \in \speckvs$, 
as changing the written and read values in a valid sequence of
$\speckvs$ yields a valid sequence of $\speckvs$.

\end{proof}

\diffhist* 

\begin{proof}
Since $\hists$ is \dataindependent{}, 
we have that
$\getdiff{\hists} \subseteq \hists$ is complete using 
relation $\smallerFData$. The result then follows
directly from \coro{coro:complete} and \lem{lem:datainvariant}.

\end{proof}
 \newpage
\section{Undecidability of Causal Consistency For $2$ \Sites{}}
\label{sec:undec}

\shuffling*

\begin{proof}
Let $\alppcp = \set{\apcp,\bpcp}$
and $(\sigu_1,\sigv_1),\dots,(\sigu_\npcp,\sigv_\npcp) 
\in (\alppcp^* \times \alppcp^*)$ be $\npcp$ pairs forming the input of a 
\pcp{} problem $\instpcp$. Let $\alpu = \set{\au,\bu}$ and
$\alpv = \set{\av,\bv}$ be 
two disjoint copies of $\alppcp$, and let 
$\morph: (\alpu \uplus \alpv) \rightarrow \alppcp$, 
$\morphu: \alppcp \rightarrow \alpu$, 
$\morphv: \alppcp \rightarrow \alpv$, 
be the homomorphisms which map corresponding letters. 
Moreover, let $\xu$ and $\xv$ be two new 
letters. Let $\alpinu =  \set{\au,\bu,\xu}$, 
$\alpinv = \set{\av,\bv,\xv}$, and  $\alpwhole = \alpinu \uplus \alpinv$.

Our goal is to define a regular language $\autL \subseteq \alp^*$ such that
\begin{align}
&\text{The \pcp{} problem $\instpcp$ has a positive answer 
$\sigu_{i_1} \cdots \sigu_{i_\kpcp} = 
    \sigv_{i_1} \cdots \sigv_{i_\kpcp}$, $(\kpcp>0)$}
\nonumber\\&\Leftrightarrow
\label{eq:pcpreduction}
\exists \wu \in \alpinu^*, \wv \in \alpinv^*.\ 
    \shuffle{\wu}{\wv} \cap \autL = \emptyset
\end{align}

The idea is to encode in $\wu$ the sequence 
$\sigu_{i_1},\dots,\sigu_{i_\kpcp}$ as 
$\uify{\sigu_{i_1}} \cc \xu \cdots  \uify{\sigu_{i_\kpcp}} \cc \xu$ by using 
$\xu$ as a
separator (and end marker), and likewise for $\wv$ with the separator 
$\xv$. Then, we define the 
language $\autL$ by a disjunction of regular properties that no shuffling of an 
encoding of a valid \pcp{} answer to $\instpcp$ could satisfy.

Formally, $\ww \in \autL$ iff one of the following conditions holds:
  \begin{enumerate}
  \item \label{cond:mismatch}
    when ignoring the letters $\xu$ and $\xv$, $\ww$ starts with an 
    alternation of $\alpu$ and $\alpv$ such that two letters do not match\\
    $\proj{\ww}{\alpu\cup\alpv} \in (\alpu\alpv)^*(\au\bv+\bu\av)\alpwhole^*$
  \item \label{cond:mislettercount}
    when ignoring the letters $\xu$ and $\xv$, $\ww$ starts with an alternation 
    of $\alpu$ and $\alpv$ and ends with only $\alpu$ letters or only 
    $\alpv$ letters\\
    $\proj{\ww}{\alpu\cup\alpv} \in (\alpu\alpv)^*(\alpu^+ + \alpv^+)$
  \item \label{cond:mispaircount}
    when only keeping $\xu$ and $\xv$ letters, either $\ww$ 
    starts with an alternation of $\xu$ and $\xv$ and ends with only 
    $\xu$ or only $\xv$, or $\ww$ is the empty word 
    $\emptyseq$\\
    $\proj{\ww}{\set{\xu,\xv}} \in (\xu\xv)^*(\xu^+ + \xv^+) + \emptyseq$
  \item \label{cond:misformed}
    $\ww$ contains a letter from $\alpu$ not followed by $\xu$, or a letter 
    from $\alpv$ not followed by $\xv$\\
    $\ww \in  \alpwhole^*\alpu(\remlett{\alpwhole}{\xu})^* +
              \alpwhole^*\alpv(\remlett{\alpwhole}{\xv})^*$
  \item \label{cond:mispair}
    $\ww$ starts with an alternation of $\alpu^*\xu$ and $\alpv^*\xv$ such 
    that one pair of $\alpu^*,\alpv^*$ is not a pair of our \pcp{} instance\\
    $\ww \in (\alpu^*\xu\alpv^*\xv)^*
      (\remlang{\alpu^*\xu\alpv^*\xv}{
        \bigplus_i \morphu(\sigu_i) \cc \xu \cc \morphv(\sigv_i) \cc \xv
      })
      \alpwhole^*$
  \end{enumerate}

We can now show that equivalence~(\ref{eq:pcpreduction}) holds.

$(\Rightarrow)$ Let $\kpcp > 0$, $i_1,\dots,i_\kpcp \in \set{1,\dots,\npcp}$
such that $\sigu_{i_1}\cdots\sigu_{i_\kpcp} = \sigv_{i_1}\cdots\sigv_{i_\kpcp}$.
Let 
$\wu = \uify{\sigu_{i_1}}\cc \xu\cdots\uify{\sigu_{i_\kpcp}}\cc \xu$ and 
$\wv = \vify{\sigv_{i_1}}\cc \xv\cdots\vify{\sigv_{i_\kpcp}}\cc \xv$.
We want to show that no word $\ww$ in the shuffling of 
$\wu$ and $\wv$ satisfies one of the 
conditions of $S$. If $\ww$ starts 
with an alternation of $\alpu$ and $\alpv$, since 
$\sigu_{i_1}\cdots\sigu_{i_\kpcp} = \sigv_{i_1}\cdots\sigv_{i_\kpcp}$,
any pair of letters match, and thus, condition~\ref{cond:mismatch} 
cannot hold. Condition~\ref{cond:mislettercount} cannot hold since $\ww$ 
contains as many letters from $\alpu$ as from $\alpv$. Likewise for 
condition~\ref{cond:mispaircount} since $\ww$ contains as many $\xu$ as $\xv$
(and at least 1).

Since $\wu$ (\resp $\wv$) does not contain a letter from $\alpu$ not 
followed by $\xu$ (\resp a letter from $\alpv$ not followed by $\xv$), 
neither does $\ww$,  which shows that condition~\ref{cond:misformed} 
does not hold. Finally, if $\ww$ starts with an alternation of $\alpu^*\xu$ and 
$\alpv^*\xv$, then all the corresponding pairs of $\alpu^*,\alpv^*$ are pairs
from the \pcp{} input, and condition~\ref{cond:mispair} cannot hold either.

$(\Leftarrow)$ 
Let $\wu \in \alpinu^*, \wv \in \alpinv^*$ such that 
$\shuffle{\wu}{\wv} \cap \autL = \emptyset$. 
Since no word in $\shuffle{\wu}{\wv}$
satisfies condition~\ref{cond:mispaircount}, nor 
condition~\ref{cond:misformed},
$\wu$ ends with $\xu$, $\wv$ ends with $\xv$, and $\wu$ has as many $\xu$ as 
$\wv$ has $\xv$ (and at least 1). 
This shows that, $\wu = \subu{1}\xu\cdots\subu{k}\xu$ and 
$\wv = \subv{1}\xv\cdots\subv{k}\xv$ for some $\kpcp > 0$,
$\subu{1},\dots,\subu{k} \in \alpu^*$, $\subv{1},\dots,\subv{k} \in \alpv^*$.
Moreover, since no word in 
$\shuffle{\wu}{\wv}$ satisfies condition~\ref{cond:mispair}, for any $j$, 
$(\subu{j},\subv{j})$ corresponds to a pair $(\sigu_{i_j},\sigv_{i_j})$ of our 
input $\instpcp$ -- more precisely, $\remsubu{\subu{j}} = \sigu_{i_j}$ and 
$\remsubv{\subv{j}} = \sigv_{i_j}$ for some $i_j \in \set{1,\dots,\npcp}$.
Finally, the fact that no word in $\shuffle{\wu}{\wv}$ satisfies
condition~\ref{cond:mismatch}, nor condition~\ref{cond:mislettercount} ensures
that $\remsubu{\subu{1}\cdots\subu{k}} = \remsubv{\subv{1}\cdots\subv{k}}$ 
and that the \pcp{} problem $\instpcp$ has a positive answer 
$\sigu_{i_1}\cdots\sigu_{i_\kpcp} = \sigv_{i_1}\cdots\sigv_{i_\kpcp}$.
\end{proof}

\undecidability*

\begin{proof}
Let $\autL$ be a regular language over $\alp = \alpinu \uplus \alpinv$. 
We construct an \implementation{} $\lib$ and a specification 
$\spec$, in order to reduce the shuffling problem (undecidable by
\lem{lem:shuffle}) to the negation of causal consistency.

Said differently, if (and only if) the shuffling problem has a positive answer, \ie
there exist 
$\wu \in \alpinu^*$ and $\wv \in \alpinv^*$ such that 
    $\shuffle{\wu}{\wv} \cap \autL = \emptyset$, then
    there exists an execution in $\lib$ which is not
    causally consistent (\resp $\wcc$, $\scc$, $\ccv$)
    with respect to $\spec$.

For the methods, we use 
    $\Meth = 
      \set{\meth_\alet\ |\ \alet \in \alpinu} \cup 
      \set{\meth_\blet\ |\ \blet \in \alpinv} \cup 
        \set{\mEndA,\mEndB}$, and 
$\Domain = \set{\boolf,\boolt}$ for the domain of the return values.
The return values are only relevant for the method $\mEndB$, so we 
do not represent return values for the other methods.
We assume here that methods do not take arguments, as we do not need them
for the reduction. As a result, a specification is a set 
of sequences labeled by $\Meth \times \Domain$ (method, return value).

Let $\wu \in \alpinu^*$ and $\wv \in \alpinv^*$.
$\lib$ produces, for any such pair, an execution whose history is 
$\encoding{u}{v}$, described hereafter.
The specification $\spec$ is then built in such a way that $\encoding{u}{v}$
is \emph{not} causally consistent if and only if 
$\shuffle{\wu}{\wv} \cap \autL = \emptyset$. 
Therefore, the shuffling 
problem has a positive answer 
if and only if $\lib$ contains an execution which 
is not causally consistent.

Here is the description of an execution corresponding to a 
pair $\wu \in \alpinu^*$ and $\wv \in \alpinv^*$,
whose history is $\encoding{u}{v}$.
The implementation contains two \sites, $\tid_A$ and $\tid_B$.
\Site{} $\tid_A$ execute 
$\meth_\alet$ operations for the letters $\alet$ of $\wu$.
\Site{} $\tid_B$ execute
$\meth_\blet$ operations for the letters $\blet$ of $\wv$.
Method $\mEndA$ is then executed on \site{} $\tid_A$.
\Site{} $\tid_A$ then sends a message to $\tid_B$, informing $\tid_B$ that
a method $\mEndB$ returning $\boolt$ can now be executed on $\tid_B$.
When the message is received by $\tid_B$, 
a method $\mEndB$ returning $\boolt$ is then executed on $\tid_B$.
(If method $\mEndB$ is called on \site{} $\tid_B$ another time,
$\tid_B$ returns $\boolf$).

\begin{remark}
\label{rem:rem}
If a method $\meth_\blet$ with $\blet \in \alpinv$
gets executed on \site{} $\tid_A$, then 
$\tid_A$ does not send the message to \site{} $\tid_B$.
If a method $\meth_\alet$ with $\alet \in \alpinu$ gets executed 
on \site{} $\tid_B$, then $\tid_B$ return $\boolf$ when 
method $\mEndB$ gets called.

Such executions will be causally consistent by default,
by construction of $\spec$, defined below, because 
they will not contain any $\mkaction{\mEndB}{\boolt}$ operations.
(The specification  $\spec$ contains, among other sequences, any sequence 
which does not contain $\mkaction{\mEndB}{\boolt}$ operation.)
\end{remark}

Formally, the set of executions $\lib$ is a regular language, 
and can be represented by the 
finite automaton given in
\fig{fig:autoto} (we do not represent the executions 
described in Remark~\ref{rem:rem}, which can never lead to 
non-causally consistent executions).

The epsilon transition 
going from $q_2$ to $q_3$ represents the fact that 
\site{} $\tid_B$ receives the message sent by $\tid_A$.
After this, \site{} $\tid_B$ can execute a $\mEndB$ method 
returning $\boolt$.
This epsilon transition is only here for clarity and can be removed.

We here do not represent transitions with $\mEndB$ returning $\boolf$,
as $\mkaction{\mEndB}{\boolf}$ operations are ignored by the specification 
$\spec$.
They can be added as self-loops to the automaton. 

\begin{figure}
 
\begin{tikzpicture}[xscale=3,yscale=3]
 
\node[initial text=,initial,state] (A) at (0,0) { $q_1$ };
\node[state] (C) at ($(A) + (1,0)$) { $q_2$ };
\node[state] (D) at ($(C) + (0,-1)$) { $q_3$ };
\node[state] (E) at ($(D) + (-1,0)$) { $q_4$ };

\draw[->] (A) to[in=60,out=120,looseness=6] node[above,text width=1.4cm] { 
    $(\tid_A,\meth_\alet)$ } (A);
\draw[->] (A) to[in=240,out=300,looseness=6] node[below,text width=1.4cm] { 
    $(\tid_B,\meth_\blet)$ } (A);
\draw[->] (C) to[in=60,out=120,looseness=6] node[above,text width=1.4cm] { 
    $(\tid_A,\meth_\alet)$ } (C);
\draw[->] (C) to[in=-30,out=30,looseness=4] node[right,text width=1.4cm] { 
    $(\tid_B,\meth_\blet)$ } (C);
\draw[->] (D) to[in=240,out=300,looseness=6] node[below,text width=1.4cm] { 
    $(\tid_A,\meth_\alet)$ } (D);
\draw[->] (D) to[in=-30,out=30,looseness=4] node[right,text width=1.4cm] { 
    $(\tid_B,\meth_\blet)$ } (D);
\draw[->] (E) to[in=240,out=300,looseness=6] node[below,text width=1.4cm] { 
    $(\tid_B,\meth_\blet)$ } (E);
\draw[->] (E) to[in=60,out=120,looseness=6] node[above,text width=1.4cm] { 
    $(\tid_A,\meth_\alet)$ } (E); 
\draw[->] (A) -- node[above] { $(\tid_A,\mEndA)$ } (C);
\draw[->] (C) -- node[left] { $\epsilon$ } (D);
\draw[->] (D) -- node[above] { $(\tid_B,\mkaction{\mEndB}{\boolt})$ } (E);

\end{tikzpicture}

\caption{Finite automaton describing the executions of implementation $\lib$
of \thm{thm:causalundec}. All states are accepting.}
\label{fig:autoto}

\end{figure}

The specification $\spec$ is defined to contain any word $w$ such that:
\begin{itemize}
\item $w$ does not contain $\mkaction{\mEndB}{\boolt}$, or
\item when ignoring $\mkaction{\mEndB}{\boolf}$,  $w$ is of the form 
    $\autLL \cc 
        \mEndA \cc 
        \mkaction{\mEndB}{\boolt}$,
where $\autLL$ is $\autL$ with every 
letter $\alet \in \alpinu$ replaced by 
$\meth_\alet$ and every letter $\blet \in \alpinv$ 
replaced by $\meth_\blet$.
\end{itemize}

We now prove the following equivalence:
\begin{enumerate}
\item \label{st:causal} $\lib$ is not \wcc{} (\resp \scc,\ccv)
    with respect to $\spec$,
\item \label{st:esshuffle} 
$\exists \wu \in \alpinu^*, \wv \in \alpinv^*.\ 
    \shuffle{\wu}{\wv} \cap \autL = \emptyset$
\end{enumerate}

$(\ref{st:esshuffle}) \Rightarrow 
(\ref{st:causal})$ 
Let $\wu \in \alpinu^*$, and $\wv \in \alpinv^*$ such that
$\shuffle{\wu}{\wv} \cap \autL = \emptyset$.
We construct an execution $\exec$ in $\lib$ which is not 
causally consistent (\resp \wcc,\scc,\ccv).
The execution $\exec$ follows the description above, and the 
history of $\exec$ is $\encoding{u}{v}$.

\Site{} $\tid_B$ executes the sequence of
operations $\meth_\blet$ for each letter $\blet$ of $\wv$.
Independently, 
\site{} $\tid_A$ executes the sequence of 
operations $\meth_\alet$ for each letter $\alet$ of $\wu$.
The \site{} $\tid_A$ then executes a ${\mEndA}$ operation and sends
a message to \site{} $\tid_B$.
After $\tid_B$ receives the message, $\tid_B$ executes 
a $\mkaction{\mEndB}{\boolt}$ operation.

Assume by contradiction that $\exec$ is $\wcc$
(a contradiction here also proves that $\exec$ cannot be $\scc$ nor $\ccv$).
There must thus exists a causal order $\co$ 
(containing the program order $\po$).
Let $\op$ be the $\mkaction{\mEndB}{\boolt}$ operation of $\exec$
We know there exists $\loc \in \spec$, 
such that $\projectrv{\causalpast{\op}}{\op} \weaker \loc$.

Since $\loc$ contains $\op$, by definition of $\spec$,
$\loc$ must be of the form 
$\loc' \cc  
    \mEndA \cc 
        \mkaction{\mEndB}{\boolt}$, with $\loc'$.
In particular, this means that $\loc$ must contain the 
 $\mEndA$ operation of $\exec$.
 By transitivity of $\co$, and because $\po \subseteq \co$,
 $\loc$ must contain all operations of $\exec$.
 
Thus, the sequence $\loc'$ effectively defines 
a shuffling of $\wu$ and $\wv$ which is in $\autL$,
contradicting the assumption that 
$\shuffle{\wu}{\wv} \cap \autL = \emptyset$.

$(\ref{st:causal}) \Rightarrow 
(\ref{st:esshuffle})$
Let $\exec$ be an execution of $\lib$ which is not causally 
consistent (\resp $\wcc$, $\scc$, $\ccv$).  
It must contain a $\mkaction{\mEndB}{\boolt}$ operation, 
otherwise, by definition of $\spec$, $\exec$ is causally consistent
regardless of how we define the causality order $\co$ (as a strict partial 
order).

Note that there can only be one $\mkaction{\mCheck}{\boolt}$ operation.
Indeed, after executing $\mkaction{\mEndB}{\boolt}$, 
\site{} $\tid_B$ only returns $\boolf$ when method $\mEndB$ gets called.

Also, for $\tid_B$ to execute a $\mkaction{\mEndB}{\boolt}$ operation,
it must be the case that $\tid_A$ executed a 
$\mEndA$ operation $\op$, and sent a message to $\tid_B$.

This means that, prior to $\op$,
$\tid_A$ must have executed a sequence of $\meth_\alet$ 
operations, with $\alet \in \alpinu$, 
corresponding to a word in $\wu \in \alpinu$.
Similarly, prior to executing the  $\mkaction{\mEndB}{\boolt}$ operation
$\tid_B$ must have 
executed a sequence of $\meth_\blet$ 
operations corresponding to a word in $\wv \in \alpinv$.

Assume by contradiction that there exists a word $\ww$ in the shuffling of 
$\wu$ and $\wv$ which belongs to $\autL$, \ie assume by contradiction 
that $\shuffle{\wu}{\wv} \cap \autL \neq \emptyset$.
Using this, we can construct a sequence $\loc \in \spec$, containing the 
 $\mkaction{\mEndA}{\boolt}$ and  $\mkaction{\mEndB}{\boolt}$ operations 
as well as all operations corresponding to $\wu$ and 
$\wv$, to form a sequence which belongs $\spec$.
This means that $\exec$ must be causally consistent (for any definition)
and we have a contradiction.

This ends the proof of equivalence between
statements \ref{st:causal} and \ref{st:esshuffle}, and ends the reduction
from the shuffling problem to checking whether an implementation 
is not causally consistent.
This implies that checking whether a \library{} is causally consistent 
is not decidable.
\end{proof}

 \newpage 
\section{Undecidability for Non-Data-Independent Read/Write Memory Implementations}

\label{apd:kvsundec}

\causalundeckvs*

\begin{proof}

Let $\alppcp = \set{\apcp,\bpcp}$
and $(\sigu_1,\sigv_1),\dots,(\sigu_\npcp,\sigv_\npcp) 
\in (\alppcp^* \times \alppcp^*)$ be $\npcp$ pairs forming the input of a 
\pcp{} problem $\instpcp$. We call these pairs \emph{dominoes}.

Our goal is to build an implementation $\lib$ such that 
$\lib$ is \emph{not} causally consistent
(\resp \wcct, \scct, \ccvt)
with respect to the $\kvs$ if and 
only if the problem $\instpcp$ has a positive answer:
$\sigu_{i_1} \cdots \sigu_{i_\kpcp} = 
    \sigv_{i_1} \cdots \sigv_{i_\kpcp}$, $(\kpcp>0)$.
    
Two sequences $(\wordu, \wordv)$ in $\alppcp^*$ form a 
\emph{\domseq} if they can be decomposed into
$\wordu = \sigu_{i_1} \cc \sigu_{i_2} \cdots \sigu_{i_\kpcp}$ and
$\wordv = \sigv_{i_1} \cc \sigv_{i_2} \cdots \sigv_{i_\kpcp}$,
with each $(\sigu_{i_j},\sigv_{i_j})$ corresponding to a pair of 
$\instpcp$.
They form a \emph{\validanswer} if we additionally have $\wordu = \wordv$.
A \validanswer{} corresponds to a positive answer for the \pcp{} problem
$\instpcp$.

\paragraph{Reduction Overview}

The implementation $\lib$ will produce, for each 
\domseq{} $(\wordu,\wordv)$, an execution whose 
history is $\encoding{u}{v}$, defined thereafter.

We construct $\encoding{u}{v}$ so that $\encoding{u}{v}$ is not
causally consistent (\resp \wcc,\scc,\ccv)
if and only if $u = v$.

Therefore, if (and only if)
$\lib$ is not causally consistent (\resp \wcct, \scct, \ccvt), 
(and can produce a history which is not causally consistent),
the \pcp{} instance $P$ has a positive answer.

\paragraph{Construction of one History}

\begin{figure*}
\centering

\begin{minipage}[t]{0.2\textwidth}
$\placeholder^a$\\
(backup \site{} $a$):\\
$[a_1]:$ $\mkwrite{\letu}{a}$ \\
$[a_2]:$ $\mkwrite{\letu}{a}$ \\
$\dots$ \\
$[a_{\numa}]:$ $\mkwrite{\letu}{a}$
\end{minipage}
\begin{minipage}[t]{0.2\textwidth}
$\placeholder^b$\\
(backup \site{} $b$):\\
$[b_1]:$ $\mkwrite{\letu}{b}$ \\
$[b_2]:$ $\mkwrite{\letu}{b}$ \\
$\dots$ \\
$[b_{\numb}]:$ $\mkwrite{\letu}{b}$
\end{minipage}
\begin{minipage}[t]{0.2\textwidth}
$\placeholder^\finaling$\\
(backup \site{} $\finaling$):\\
$[c_1]:$ $\mkwrite{\letu}{\finaling}$ \\
$[c_2]:$ $\mkwrite{\letu}{\finaling}$ \\
$\dots$ \\
$[c_{\numc}]:$ $\mkwrite{\letu}{\finaling}$
\end{minipage}
\vspace{3em}

\begin{minipage}[t]{0.22\textwidth}
$\extra^a$\\
(extra \site{} $a$):\\
$[ex^a]:$ $\mkwrite{\letu}{a}$\\
$[ch^a]:$ $\mkwrite{\choicevar}{\vraival}$
\end{minipage}
\begin{minipage}[t]{0.22\textwidth}
$\extra^b$\\
(extra \site{} $b$):\\
$[ex^b]:$ $\mkwrite{\letu}{b}$\\
$[ch^b]:$ $\mkwrite{\choicevar}{\vraival}$
\end{minipage}
\begin{minipage}[t]{0.22\textwidth}
$\extra^\finaling$\\
(extra \site{} $\finaling$):\\
$[ex^\finaling]:$ $\mkwrite{\letu}{\finaling}$\\
$[ch^\finaling]:$ $\mkwrite{\choicevar}{\vraival}$
\end{minipage}
\vspace{3em}

\begin{minipage}[t]{0.18\textwidth}
$\ticking{u}$\\
(ticker \site{} $u$):\\
$[s_u]:$ $\mkwrite{\bbbu}{\vraival}$\\
$[g_1]:$ $\mkwrite{\tickuf}{\vraival}$\\
$\dots$\\
$[g_{\mres+2}]:$ $\mkwrite{\tickuf}{\vraival}$\\
\end{minipage}
\begin{minipage}[t]{0.18\textwidth}
$\ticking{v}$\\
(ticker \site{} $v$):\\
$[s_v]:$ $\mkwrite{\bbbv}{\vraival}$\\
$[h_1]:$ $\mkwrite{\tickvf}{\vraival}$\\
$\dots$\\
$[h_{\nres+2}]:$ $\mkwrite{\tickvf}{\vraival}$\\
\end{minipage}
\begin{minipage}[t]{0.18\textwidth}
$\savior{u}$\\
$[w^s_u]:$ $\mkwrite{\bbbu}{\vraival}$\\
$[ch_s^u]:$ $\mkwrite{\choicevar_s}{\vraival}$\\
\end{minipage}
\begin{minipage}[t]{0.18\textwidth}
$\savior{u}$\\
$[w^s_v]:$ $\mkwrite{\bbbv}{\vraival}$\\
$[ch_s^v]:$ $\mkwrite{\choicevar_s}{\vraival}$\\
\end{minipage}
\vspace{3em}

\begin{minipage}[t]{0.26\textwidth}
$\procu$:\\
\\
\\
$[x_1]:$ $\mkwrite{\letu}{U_1}$\\
\\
$[t_1^u]:$ $\mkwrite{\tickuf}{\vraival}$\\
$[r_1^u]:$ $\mkget{\tickvf}{\vraival}$\\[1.5ex]
$[x_2]:$ $\mkwrite{\letu}{U_2}$\\
\\
$[t_2^u]:$ $\mkwrite{\tickuf}{\vraival}$\\
$[r_2^u]:$ $\mkget{\tickvf}{\vraival}$\\
$\dots$\\
$[x_\nres]:$ $\mkwrite{\letu}{U_\nres}$\\
\\
$[t_\nres^u]:$ $\mkwrite{\tickuf}{\vraival}$\\
$[r_\nres^u]:$ $\mkget{\tickvf}{\vraival}$\\[1.5ex]
$[x_{\nres+1}]:$ $\mkwrite{\letu}{\finaling}$\\
\\
$[t_{\nres+1}^u]:$ $\mkwrite{\tickuf}{\vraival}$\\
$[r_{\nres+1}^u]:$ $\mkget{\tickvf}{\vraival}$
\\[0.7ex]
$[t_{\nres+2}^u]:$ $\mkwrite{\tickuf}{\vraival}$\\
$[r_{\nres+2}^u]:$ $\mkget{\tickvf}{\vraival}$
\\[0.7ex]
$[r^s_u]:$ $\mkread{\bbbu}{\fauxval}$\\
$[w_m]:$ $\mkwrite{M}{\vraival}$\\
\end{minipage}
\begin{minipage}[t]{0.26\textwidth} 
$\procv$:\\
$[r_{ch}]:$ $\mkread{\choicevar}{\vraival}$\\
$[d]:$ $\mkwrite{\letu}{\fauxval}$\\
$[y_1]:$ $\mkget{\letu}{\other{V_1}}$\\
$[z_1]:$ $\mkget{\letu}{\finaling}$\\
$[t_1^v]:$ $\mkwrite{\tickvf}{\vraival}$\\
$[r_1^v]:$ $\mkget{\tickuf}{\vraival}$\\[1.5ex]
$[y_2]:$ $\mkget{\letu}{\other{V_2}}$\\
$[z_2]:$ $\mkget{\letu}{\finaling}$\\
$[t_2^v]:$ $\mkwrite{\tickvf}{\vraival}$\\
$[r_2^v]:$ $\mkget{\tickuf}{\vraival}$\\
$\dots$\\
$[y_\mres]:$ $\mkget{\letu}{\other{V_\mres}}$\\
$[z_\mres]:$ $\mkget{\letu}{\finaling}$\\
$[t_\mres^v]:$ $\mkwrite{\tickvf}{\vraival}$\\
$[r_\mres^v]:$ $\mkget{\tickuf}{\vraival}$\\[1.5ex]
$[y_{\mres+1}]:$ $\mkget{\letu}{a}$\\
$[z_{\mres+1}]:$ $\mkget{\letu}{b}$\\
$[t_{\mres+1}^v]:$ $\mkwrite{\tickvf}{\vraival}$\\
$[r_{\mres+1}^v]:$ $\mkget{\tickuf}{\vraival}$\\[0.7ex]
$[t_{\mres+2}^v]:$ $\mkwrite{\tickvf}{\vraival}$\\
$[r_{\mres+2}^v]:$ $\mkget{\tickuf}{\vraival}$\\[0.7ex]
$[r^s_v]:$ $\mkread{\bbbv}{\fauxval}$\\[1.3ex]
$[r^a]:$ $\mkget{\letu}{a}$\\
$[r^b]:$ $\mkget{\letu}{b}$\\
$[r^\finaling]:$ $\mkget{\letu}{\finaling}$\\
$[w_n]:$ $\mkwrite{N}{\vraival}$\\
\end{minipage}
\begin{minipage}[t]{0.22\textwidth}
$\pag$:\\
$[r_m]:$ $\mkread{M}{\vraival}$ \\
$[r_n]:$ $\mkread{N}{\vraival}$ \\
$[r_{ch_s}]:$ $\mkread{\choicevar_s}{\vraival}$ \\
$[wf^s_u]:$ $\mkwrite{\bbbu}{\fauxval}$ \\
$[wf^s_v]:$ $\mkwrite{\bbbv}{\fauxval}$ \\
$[rf^s_u]:$ $\mkread{\bbbu}{\vraival}$ \\
$[rf^s_v]:$ $\mkread{\bbbv}{\vraival}$ \\
\end{minipage}

\caption{
This history $\encoding{u}{v}$ corresponds to a \domseq{} 
$(U_1\cdots U_\nres, V_1\cdots V_\mres)$, with
$U_i,V_i \in \set{a,b}$ for all $i$.
$\encoding{u}{v}$ is \emph{not} causally consistent (\resp \wcc,\scc,\ccv)
if and only ($\nres = \mres$ and) $U_1\cdots U_\nres = V_1 \cdots V_\mres$, 
\ie $(U_1\cdots U_\nres, V_1\cdots V_\mres)$ form a \validanswer{} for $\instpcp$.
This history uses
$9$ variables and a domain size of $4$.
}
\label{fig:encoding}
\end{figure*}

Given a letter $L \in \alppcp$, 
we define $\other{L} = b$ if $L = a$, $\other{L} = a$ if $L = b$.

Let $u = U_1 \cdots U_\nres$, and $v = V_1 \cdots V_\mres$,
with $\nres,\mres > 0$, 
and $U_i,V_i \in \alppcp$ for all $i$.
We depict in \fig{fig:encoding} the history $\encoding{u}{v}$ in $\lib$
corresponding to $(u,v)$. 
We show in \lem{lem:oneexec} that $\encoding{u}{v}$ is 
not causally consistent (for any definition) if and only if $u = v$.

To define $\encoding{u}{v}$,  we make use of the construct 
$\mkget{x}{\dat}$, which denotes the sequence of operations 
$
\mkwrite{x}{\fauxval} \cc
\mkread{x}{\dat} \cc
\mkwrite{x}{\fauxval}$, 
for $x\in\Var$, and $\dat \neq \fauxval$.
This is only a notation, and does not imply that the three operations must be 
executed atomically.
It is introduced only to simplify the presentation of the proof.

This construct ensures the useful property that 
$\mkget{x}{\vraival}$ operations
made by the same \site{} $\tid$ need distinct
$\mkwrite{x}{\vraival}$ to read from, as \site{} $\tid$ overwrite 
$x$ with $\mkwrite{x}{\fauxval}$ after reading $\mkread{x}{\vraival}$.
More generally, for any $m \in \Nats$, 
if a \site{} $\tid$ does $m$ operations  
$\mkget{x}{\dat}$ for some $\dat \neq \fauxval$,
then there need to be at least $m$ distinct $\mkwrite{x}{\dat}$ in the  
execution for the execution to be causally consistent.

We now have all the ingredients needed to prove that the execution $\encoding{u}{v}$
of \fig{fig:encoding} satisfies the property we want:
$\encoding{u}{v}$ is causally consistent (\resp \wcct, \scct, \ccvt)
if and only if $u \neq v$.
In the figure, we put the operations' names between brackets, so that we 
can refer to them in the proof. When there is an operation name next to a 
{\tt uniq\_rd} operation, it is actually the name corresponding to the 
underlying {\tt rd} operation.

The idea is the following. If $u$ and $v$ are different, 
then there exists $i$ such that $U_i \neq V_i$ and $U_i = \other{V_i}$
(or $u$ and $v$ have different sizes).
This means that the write $x_i$ can be used for the read $y_i$,
which means that one write from the backup \sites{} need not be used 
for $y_i$, and can be used instead for one of the three reads $r^a$, $r^b$, 
$r^\finaling$. Moreover, two writes (out of three) from the extra \sites{}
can be used for the reads $r^a$, $r^b$, 
$r^\finaling$, which makes the history $\encoding{u}{v}$ causally consistent.

(The case where $u$ and $v$ have different sizes is handled 
thanks to the ticker \sites{}, the symbol $\finaling$, and 
the variables $\bbbu$,$\bbbv$,$M$,$N$, and $\choicevar_s$.)

If $u$ and $v$ are equal, then $U_i$ is different than 
$\other{V_i}$, for all $i$. Then, the writes of all backup processes
must be used for the reads $y_i$ and $z_i$, and cannot be 
used for the three reads $r^a$, $r^b$, $r^\finaling$.
Since only two (out of three) writes from the extra \sites{} can be used 
(one write is lost because of variable $\choicevar$, and read operation 
$r_{ch}$), the three reads $r^a$, $r^b$, $r^\finaling$ 
cannot be consistent, and the history $\encoding{u}{v}$
is not causally consistent (for any definition).

The technicalities are given in the following lemma.

\begin{lemma}
\label{lem:oneexec}
The following statements are equivalent:
\begin{enumerate}
\item \label{enum:execwcc} $\encoding{u}{v}$ is \wcc
\item \label{enum:execscc} $\encoding{u}{v}$ is \scc
\item \label{enum:execccv} $\encoding{u}{v}$ is \ccv 
\item \label{enum:noteq} $u \neq v$
\end{enumerate}
\end{lemma}

\begin{proof} 
Let $\numa$ (\resp $\numb$, $\numc$) be the number of 
$\mkread{\letu}{a}$ (\resp $\mkread{\letu}{b}$, $\mkread{\letu}{\finaling}$)
operations among the 
$y_1,z_1,\dots,y_{\mres+1},z_{\mres+1}$ operations.
Note that $\numa + \numb + \numc = 2*\mres + 2$.

$(\ref{enum:execwcc} \Rightarrow \ref{enum:noteq})$ 
Assume $\encoding{u}{v}$ is \wcct, and let $\co$ be a causality order 
which proves it.
Assume by contradiction that $u = v$.
Thus, $\nres = \mres$ and for all $i \in \set{1,\dots,\nres}$, 
we have $U_i = V_i$,
$U_i \neq \other{V_i}$ and $U_i \neq \finaling$.

Our goal (*) is to show that the reads $y_i$ and $z_i$ can only read 
from $x_i$ (or from the backup/extra \sites{}), but not from 
$x_j$ for $i \neq j$.

Then, since $U_i \neq \other{V_i}$ and $U_i \neq \finaling$,
the reads $y_i$ and $z_i$ must use the writes from 
the backup or extra \sites{}, and the writes $x_i$ have no use.

Overall, there are $2*\mres + 5$ {\tt uniq\_rd}'s on variable $\letu$ in $\procv$,
and $2*\mres + 5$ writes to $\letu$ in the backup and extra \sites{}.
However, because of the \emph{choice} variable $\choicevar$,
one of the write from the extra \sites{} must be causally related to 
$r_{ch}$ (and to operation $d$), and thus cannot be used for the 
{\tt uniq\_rd}'s on variable $\letu$ in $\procv$.
We are left with only $2*\mres + 4$ writes to $\letu$, which are usable 
for the $2*\mres + 5$ {\tt uniq\_rd} on variable $\letu$ in $\procv$, which
means that $\encoding{u}{v}$ cannot be \wcc.
We conclude that $u \neq v$.

(*) We now show that the reads $y_i$ and $z_i$ can only read 
from $x_i$ (or from the backup/extra \sites{}), but not from 
$x_j$ for $i \neq j$.

First, notice that because of the variables $M$ and $N$, 
all the operations of $p_u$ and $p_v$ must be causally 
related to $r_n$ in $\pag$.

Assume by contradiction that an operation $g_i$ is causally related to 
an operation $r^v_j$, and that an operation $h_{i'}$ is causally related 
to an operation $r^u_{j'}$. Then both write operations $s_u$ and $s_v$ would 
be causally related to $r_n$ is $\pag$, and are not usable for
the reads $rf_u^s$ and $rf_v^s$.
Moreover, because of variable $\choicevar_s$, 
only one of the writes $w_u^s$ and $w_v^s$ is usable for the 
reads  $rf_u^s$ and $rf_v^s$. Since there are no other writes to 
$\bbbu$ or $\bbbv$, this is contradiction.

Now, assume by contradiction an operation $g_i$ is causally related to
an operation $r^v_j$, with $j \leq \mres + 1$.
Then, we know that no 
operation $h_{i'}$ can be causally related 
to an operation $r^u_{j'}$.
The reads $r^u_{j'}$ must therefore use the writes $t^v_{j'}$,
and $t^v_{j'}$ must be causally related to $r^u_{j'}$ for all 
$j' \in \set{1,\dots,\mres+2}$.
But then, by transitivity, we would have that 
$g_i$ is causally related to $r^s_u$, and also that 
$s_u$ is causally related to $r^s_u$, which is not possible
(as there are no $\mkwrite{s_u}{\fauxval}$ operation in the history).

Similarly, we can prove that no operation $h_i$ can be causally related to 
an operation $r^u_j$ with $j \leq \mres + 1$.

This entails that, each $r^u_j$  with $j \leq \mres + 1$,
must use the write $t^v_j$. And read $r^v_j$ with $j \leq \mres + 1$
must use the write $t^u_j$.
This implies in particular than each $x_j$ is causally related to 
$r^v_j$ for $j \leq \mres + 1$, and cannot be 
used for a read $y_i$ with $i > j$.

Moreover, $x_j$ cannot be causally related to 
$y_i$ with $i < j$.
This would create a cycle in the causality relation, as 
we know that $t^v_i$ is causally related to $r^u_i$.

This concludes the proof that 
the reads $y_i$ and $z_i$ can only read 
from $x_i$ (or from the backup/extra \sites{}), but not from 
$x_j$ for $i \neq j$.

$(\ref{enum:noteq} \Rightarrow \ref{enum:execscc})$ Assume $u \neq v$.
We have three cases to consider: $\nres = \mres$, $\nres > \mres$
and $\mres > \nres$.

Case $\nres = \mres$.
Let $j \in \set{1,\dots,\nres}$ such that $U_j \neq V_j$.
Assume without loss of generality that we have
$U_j = a$, and $V_j = b$ (the other case is symmetric).
By definition, $U_j = \other{V_j}$.
In that case, we prove that $\encoding{u}{v}$ is \scct{}.

Let $q \in \set{1,\dots,\numa}$ such that $y_j$ is the 
$q$th $\mkread{\letu}{a}$ in $\procv$.
We define the causality relation $\co$ as the transitive closure
of the program order and the following constraints:
\begin{itemize}
\item $t_i^u$ to $r_i^v$, for $i \in \set{1,\dots,\nres+2}$,
\item $t_i^v$ to $r_i^u$, for $i \in \set{1,\dots,\nres+2}$,
\item $b_i$ to $y_{i'}$, for $i \in \set{1,\dots,\numb}$, and where 
$y_{i'}$ is the $i$th $\mkread{\letu}{b}$ in $\proc_v$,
\item $c_i$ to $z_{i'}$, for $i \in \set{1,\dots,\numc}$, and where 
$z_{i'}$ is the $i$th $\mkread{\letu}{\finaling}$ in $\proc_v$,
\item $a_i$ to $y_{i'}$, for $i \in \set{1,\dots,q-1}$, and where 
$y_{i'}$ is the $i$th $\mkread{\letu}{a}$ in $\proc_v$,
\item $x_j$ to $y_j$,
\item $a_{i-1}$ to $y_{i'}$, for $i \in \set{q+1,\dots,\numa}$, and where 
$y_{i'}$ is the $i$th $\mkread{\letu}{a}$ in $\proc_v$,
\item $a_\numa$ to $r^a$,
\item $ch^a$ to $r_{ch}$,
\item $ex^b$ to $r^b$,
\item $ex^\finaling$ to $r^\finaling$,
\item $w_m$ to $r_m$,
\item $w_n$ to $r_n$,
\item $s_u$ to $rf^s_u$,
\item $s_v$ to $rf^s_v$,
\item $ch_s^u$ to $r_{ch_s}$.
\end{itemize}

By construction, $\co$ is a strict partial order, and we have
$\po \subseteq \co$.
Also, for each read operation $\rop$ (in particular $r^u_\nres$ and 
$r^b$), we can construct a sequence, which respects the causality 
order, and containing all return values of the read operations
before $\rop$ in the program order.
for every operation $\op$ of $\encoding{u}{v}$, 
there exists $\loc \in \speckvs$ such that
$\projectrv{\causalpast{\op}}{\poback{\op}} \weaker \loc$
(axiom \axscc).

The key idea here is that, since 
$U_j = \other{V_j}$, $y_j$ can read from $x_j$.
As a result, there is a $\mkwrite{\letu}{a}$ from the backup 
\site{} $a$ which is not needed for the $\mkget{\letu}{a}$ operations of 
$\procv$.
We can thus use the last write from the backup \site{} $a$ (\ie $a_\numa$), 
for the read $r^a$.
Then, we use $ch^a$ to explain the return value of $r_{ch}$, and we can 
therefore use $ex^b$ to explain the return value of $r^b$, 
and $ex^\finaling$ to explain the return value of $r^\finaling$.

\vspace{2eX}

Case $\nres > \mres$.
We prove that $\encoding{u}{v}$ is \scct{}.
We define the causality relation $\co$ as the transitive closure
of the program order and the following constraints:
\begin{itemize}
\item $t_i^u$ to $r_i^v$, for $i \in \set{1,\dots,\mres+2}$,
\item $t_i^v$ to $r_i^u$, for $i \in \set{1,\dots,\mres+2}$,
\item $h_i$ to $r_i^u$, for $i \in \set{\mres+2,\nres+2}$,
\item $a_i$ to $y_{i'}$, for $i \in \set{1,\dots,\numa}$, and where 
$y_{i'}$ is the $i$th $\mkread{\letu}{a}$ in $\proc_v$,
\item $b_i$ to $y_{i'}$, for $i \in \set{1,\dots,\numb}$, and where 
$y_{i'}$ is the $i$th $\mkread{\letu}{b}$ in $\proc_v$,
\item $c_i$ to $z_{i'}$, for $i \in \set{1,\dots,\numc}$, and where 
$z_{i'}$ is the $i$th $\mkread{\letu}{\finaling}$ in $\proc_v$,
\item $ch^\finaling$ to $r_{ch}$,
\item $ex^a$ to $r^a$,
\item $ex^b$ to $r^b$,
\item $x_{\nres+1}$ to $r^\finaling$,
\item $w_m$ to $r_m$,
\item $w_n$ to $r_n$,
\item $s_u$ to $rf^s_u$,
\item $w_v^s$ to $rf^s_v$,
\item $ch_s^u$ to $r_{ch_s}$.
\end{itemize}

\vspace{2eX}

Case $\mres > \nres$. 
Similar to the previous two cases.
Here, $x_{\nres+1}$ will be used for $z_{\nres+1}$,
thus allowing the write $c_\numc$ to be used for $r^\finaling$.

$(\ref{enum:noteq} \Rightarrow \ref{enum:execccv})$ 
We can prove this by using the same causality order used for $\scc$.
We can then define an arbitration order, 
as all the \sites{} agree on the order of write 
operations.

$(\ref{enum:execscc} \Rightarrow \ref{enum:execwcc})$ 
By \lem{lem:cmtocc}. 

$(\ref{enum:execccv} \Rightarrow \ref{enum:execwcc})$ 
By \lem{lem:ccvtocc}.

\end{proof}

\paragraph{Construction of the Implementation}

We now describe how to build the implementation $\lib$,
such that 
$\lib$ is \emph{not} causally consistent
with respect to the $\kvs$ if and  
only if the problem $\instpcp$ has a positive answer.

More precisely, we describe how to define $\lib$ as a regular language, 
so that $\lib$ produces, for each \domseq{} $(u,v)$, an execution 
whose history is $\encoding{u}{v}$.

In an execution of $\lib$, the following happens:
First, the extra \sites{}, as well as the \sites{}
$\tid_{\bbbu}$ and $\tid_{\bbbv}$ execute their operations.
Each ticker \site{} executes its first operation.
Then, \site{} $\procv$ executes operations $r_{ch}$ and $d$.

After that,
\Site{} $\procu$ chooses non-deterministically a domino 
$(u_i,v_i)$ from the \pcp{} instance $\instpcp$.
It sends messages to the backup \sites{}, the ticker \sites{}, 
and $\procv$ so that they execute the operations corresponding 
to this domino $(u_i,v_i)$ (following \fig{fig:encoding}).

This step, of choosing non-deterministically a domino and what follows,
can happen an arbitrary number of times.

All \sites{} thus synchronize after each choice 
of a domino. The history of an execution $\exec$ of $\lib$ 
thus always corresponds to a prefix of the history given in 
\fig{fig:encoding}.

Finally, the ticker \sites{}, as well as $\procu$, $\procv$ and $p_f$ execute 
their last operations, as depicted in \fig{fig:encoding}.

Since the \sites{} synchronize after each choice of a domino,
$\lib$ can be described by a regular language (or equivalently, 
by a distributed implementation where each \site{} has a bounded local 
memory, and where the \sites{}
communicate through a network whose capacity is bounded).

\begin{remark}
In an implementation, each method can be called at any time, on any 
\site{}. We handle this like in \thm{thm:causalundec}: if a
\site{} detects a method call that it is not expecting (\ie that does not
follow \fig{fig:encoding}), the implementation
falls back to a default implementation which is causally 
consistent (\resp $\wcc$, $\scc$, $\ccv$).
Therefore, if $\lib$ can produce an execution which is not causally 
consistent, 
it must be an execution whose history is of the form 
$\encoding{u}{v}$ where $(u,v)$ form a \domseq.
\end{remark}

\begin{lemma}
$\lib$ is not causally consistent (\resp \wcc,\scc,\ccv) if and only if 
the \pcp{} problem $\instpcp$ has a positive answer.
\end{lemma}

\begin{proof}
$(\Rightarrow)$ If $\lib$ is not causally consistent (\resp \wcc,\scc,\ccv), 
it produces a history $\hist$ which is not causally consistent.
By construction of $\lib$, 
$\hist$ must be of the form $\encoding{u}{v}$, for some \domseq{} $(u,v)$.
By \lem{lem:oneexec}, we know that $u = v$, and 
$(u,v)$ for a valid answer for $\instpcp$.

$(\Leftarrow)$ If $(u,v)$ form a valid answer to $\instpcp$,
then $\lib$ is not causally consistent, as it can produce an execution
whose history is $\encoding{u}{v}$, which is not causally 
consistent (\resp \wcc,\scc,\ccv) by \lem{lem:oneexec}.
\end{proof}

\end{proof}
 \newpage
\section{Reduction to Control-State Reachability}

\reduction*

\begin{proof}
$(\Rightarrow)$ Assume by contradiction that there is an execution 
$\exec \in \lib$ which is accepted by $\obscc$. We make a case analysis 
based on which branch of $\obscc$ accepts $\exec$.

(First branch)
If $\exec$ is accepted on state $\errorstate$, 
then it has a $\mktrueaction{\readmeth}{\var}{\firstdat}$
  operation with no corresponding $\writemeth$ operation.
  It therefore contains bad pattern $\bpbprf$, and $\exec$ is not 
  $\wcc$.

(Second branch)
Otherwise, 
  $\exec$ is accepted on state $\errorstate'$.
  Let $\wop_1$ be the $\mktrueaction{\writemeth}{\var}{\firstdat}$ operation 
  read after $q_1$.
  Let $\wop_2$ be the $\mktrueaction{\writemeth}{\var}{\seconddat}$ operation
  read after the first causal link.
  Let $\rop_1$ be the $\mktrueaction{\readmeth}{\var}{\firstdat}$
  operation read just before $\errorstate'$.
  Let $\getdiff{\exec} \in \lib$ be a \differentiated{} \execution{} and 
  $\renaming$ a renaming such that 
  $\exec = \applyrenaming{\getdiff{\exec}}{\renaming}$.
  Let $\getdiff{\firstdat} \in \Nats$ be the data value of $\rop_1$ 
  in $\getdiff{\exec}$. The renaming $\renaming$ maps $\getdiff{\firstdat}$
  to $\firstdat$.
  
  We show that, in the execution $\getdiff{\exec}$, the operations 
  $\wop_1$, $\wop_2$, and $\rop_1$ form bad pattern $\bpbpco$ because 
  $\wop_1 \ltco \wop_2 \ltco \rop_1$,
  $\wop_1 \ltrf \rop_1$, and $\getvar{\wop_1} = \getvar{\wop_2}$.
  The conditions $\getvar{\wop_1} = \getvar{\wop_2}$, and 
  $\wop_1 \ltrf \rop_1$ hold, because these three operations all operate
  on the same variable (ensured by the register $\regvar'$), 
  and $\wop_1$ and $\rop_1$ use the same data value $\getdiff{\firstdat}$.
  
  The condition $\wop_2 \ltco \rop_1$ holds because either $\rop_1$
  is between $q_5$ and $\errorstate'$, in which case $\regsite$ ensures that
  it is on the same \site{} as $\wop_2$; or $\rop_1$ is between $q_6$ 
  and $\errorstate'$, and in that case the preceding 
  $\mktrueaction{\readmeth}{\var}{\seconddat}$ operation makes a causality 
  link with $\wop_2$.
  
  The causality links $\wop_1 \ltco \wop_2 \ltco \rop_1$ are ensured by 
  the presence of the $\clink$ subautomata. $\clink$ recognizes 
  unbounded chains in the $\po \cup \rf$ relation.
  In $\clink$, a $\readmeth$ operation read from $q_b$ to $q_a$ reads-from 
  a preceding $\writemeth$ operation from $q_a$ to $q_b$ (thanks to register
  $\regvar$).
  Moreover, the register $\regsite$ ensures that a $\readmeth$ operation is 
  before, in the program order, the $\writemeth$ operation which is following 
  it.
  
  We conclude by \lem{lem:causalreg} that $\getdiff{\exec}$ is not causally
  consistent.
  
(Third branch) Similar to the previous branch, but for bad pattern 
$\bpbprfi$.
  
$(\Leftarrow)$
Assume by contradiction that there is an execution $\getdiff{\exec} \in \lib$
which is not causally consistent. By \lem{lem:diffhists}, we can assume 
that $\getdiff{\exec}$ is \differentiated{}. 
Using \lem{lem:causalreg}, we have four bad patterns to consider. 

(\bpbpcyclicco)
The first case is when there is a cycle in $\po \cup \rf$.
Without loss of generality, we can assume that the cycle is an alternation 
of $\po$ and $\rf$ edges, of the form ($\ncycle > 1$): 
\[
  \rop_1 \ltpo \wop_2 \ltrf \rop_2 \ltpo \wop_3 \dots \ltrf 
    \rop_{\ncycle-1} \ltpo \wop_\ncycle \ltrf \rop_\ncycle = \rop_1.
\]
This is true for two reasons. First, $\po$ is transitive, so two $\po$ edges
can always be contracted to one. Second, $\rf$ connects $\writemeth$ to 
$\readmeth$ operations, so there cannot be two $\rf$ edges one after the other.

Consider the minimal prefix $\getdiff{\exec}'$ of $\getdiff{\exec}$ which 
contains only one out of these $2*(\ncycle-1)$ operations. This operation must 
be a $\readmeth$ operation, as every $\writemeth$ operation $\wop_i$ is 
preceded by $\rop_{i-1} \ltpo \wop_i$ in the program order.
Note that $\getdiff{\exec}'$ belongs to $\lib$, as $\lib$ is prefix-closed.

The \execution{} $\getdiff{\exec}'$ thus contains a $\readmeth$ operation $\rop$
which has no corresponding $\writemeth$ \operation{} anywhere else in the execution, 
as its corresponding $\writemeth$ \operation{} was among the 
$2*(\ncycle-1)$ operations above, and was not kept in $\getdiff{\exec}'$.

Consider the renaming $\renaming$ which maps the data value of $\rop$ to 
$\firstdat$, and every other value to $\fifthdat$.
By data independence, $\applyrenaming{\getdiff{\exec}'}{\renaming}$ belongs to $\lib$. Moreover, 
$\applyrenaming{\getdiff{\exec}'}{\renaming}$ can be recognized by
(the first branch of) $\obscc$.
We thus obtain a contradiction, as $\lib \cap \obscc$ is not empty.

(\bpbprf)
  The second case is when there is a $\readmeth$ operation with no 
  corresponding $\writemeth$ operation. Again, such an execution can be 
  recognized by the first branch of $\obscc$ (after renaming).
  
(\bpbpco)
  The third case is when there are operations $\wop_1$, $\wop_2$, and $\rop_1$ 
  in $\getdiff{\exec}$ such that
  $\wop_1 \ltrf \rop_1$,
  $\getvar{\wop_1} = \getvar{\wop_2}$, and
  $\wop_1 \ltco \wop_2 \ltco \rop_1$.
  
  Consider the renaming $\renaming$ which maps:
  \begin{itemize}
  \item   
    the data value of $\wop_1$ and
    $\rop_1$ to $\firstdat$, 
  \item 
    the data value of $\wop_2$ to $\seconddat$,
  \item 
    maps any value which appears in a $\writemeth$ \operation{} in a
    causality chain between $\wop_1 \ltco \wop_2$ to $\thirddat$, 
  \item 
    maps any value which appears in a $\writemeth$ \operation{} in a
    causality chain between $\wop_1 \ltco \wop_2$ to $\fourthdat$, 
  \item
    maps any other value to $\fifthdat$.
  \end{itemize}
  
  Then, $\applyrenaming{\getdiff{\exec}}{\renaming}$ can be recognized 
  by (the second branch of) $\obscc$, and  
  $\lib \cap \obscc$ is not empty.
  
(\bpbprfi)
    This bad pattern can be treated similarly to the previous one, but
    using the third branch of $\obscc$ instead of the second.

\end{proof}
 \newpage

\section{\ccv{} Bad Patterns}

\ccvbadpatterns*

\begin{proof} 
Let $\hist = (\ops,\po,\getlabel)$ be a \differentiated{} \history{}.

$(\Rightarrow)$
Assume that $\hist$ is \ccvt{}, and let $\co \subseteq \theirarb$
be relations satisfying the properties of \ccv{}.
By \lem{lem:ccvtocc}, we know that $\hist$ is \wcct{}.
Assume by contradiction that $\hist$ contains bad
pattern \bpbpcycliccf.

Consider any edge $\wop_1 \ltcf \wop_2$ in the $\cf$ relation,
where 
$\getlabel(\wop_1) = \mkwrite{\var}{\dat_1}$ and 
$\getlabel(\wop_2) = \mkwrite{\var}{\dat_2}$ for 
some $\var \in \Var$ and $\dat_1 \neq \dat_2 \in \Nats$.

By definition of $\cf$, we have 
$\wop_1 \ltco \rop_2$, where 
$\getlabel(\rop_2) = \mkread{\var}{\dat_2}$.

Moreover, \ccv{} ensures that 
there exists $\locof{\rop_2} \in \speckvs$ such that
$\projectrv{(\causaldep{\rop_2},\theirarb,\getlabel)}{\rop_2} 
\weaker \locof{\rop_2}$.

Since both $\wop_1$ and $\wop_2$ are in $\locof{\rop_2}$,
$\wop_1$ must be before $\wop_2$ in $\locof{\rop_2}$, as 
$\rop_2$ is the last operation  of $\locof{\rop_2}$.
(and $\hist$ is \differentiated).
As a result, $\wop_1 \ltrel{\theirarb} \wop_2$, and 
$\cf \subseteq \theirarb$. 
The cycle in $\cf \cup \propco$ thus induces a cycle in 
$\theirarb \cup \propco$, which contradicts the fact that 
$\propco \subseteq \co \subseteq \theirarb$ 
and that $\theirarb$ is strict total order.

$(\Leftarrow)$
Assume that $\hist$ is \wcct{} and does not contain 
bad pattern \bpbpcycliccf.
We use the causal order $\co = \propco = (\po \cup \rf)^+$ to show that 
$\hist$ is \ccvt{} with respect to $\speckvs$.
We must also construct the arbitration order $\theirarb$, which is 
a strict total order over $\ops$.

We define $\theirarb$ as any strict total order which contains 
$\cf \cup \propco$. This is possible since $\cf \cup \propco$ is acyclic 
($\hist$ does not contain bad pattern \bpbpcycliccf).

Let $\rop \in \ops$ be a read operation.
We prove that there exists 
$\locof{\rop} \in \speckvs$ such that
$\projectrv{\causalarb{\rop}}{\rop} \weaker \locof{\rop}$.

In the case that $\rop$ returns the initial value $0$,
and because $\hist$ does not contain bad pattern 
\bpbprfi, there is no write on the same variable as $\rop$
in $\causaldep{\rop}$.
The sequence $\locof{\rop}$ can thus be defined as 
adding appropriate values to the reads different from $\rop$ in 
$\projectrv{\causalarb{\rop}}{\rop}$
(that is, the value of the preceding write on the same variable,
or the initial value $0$ if there is no such write).

If $\rop$ returns a value different that $0$, we know that there is 
a corresponding write $\wop$ in $\causaldep{\rop}$.
Consider any write operation $\wop' \neq \wop$ in $\causaldep{\op}$ which is 
on the same variable as $\op$.
By definition of the \conflict{} relation $\cf$ and 
by definition of $\theirarb$, we know that 
$\wop' \ltrel{\theirarb} \wop$.
Thus, the last write operation on variable in 
$\projectrv{\causalarb{\rop}}{\rop}$ must be $\wop$.
As previously,
we can thus define $\locof{\rop}$ as 
$\projectrv{\causalarb{\rop}}{\rop}$
where we add appropriate return values to the reads different that
$\rop$. 
\end{proof}

 \newpage

\section{\scc{} Bad Patterns}

\strongbadpatterns*

\begin{proof}
Let $\hist = (\ops,\po,\getlabel)$ be a \differentiated{} \history{}.

$(\Rightarrow)$
Assume that $\hist$ is \scct.
By \lem{lem:cmtocc}, we know that $\hist$ is \wcct{}.

Assume by contradiction that 
$\hist$ contains bad pattern \bpbpchangerfi{} or \bpbpcyclichb{}
for some operation $\op$.

By \scc{}, there is $\locof{\op} \in \speckvs$ 
with
$\projectrv{\causalpast{\op}}{\poback{\op}}  \weaker \locof{\op}$.
This implies in particular that the return values of all read 
operations which are before $\op$ (in $\poback{\op}$)
are still present in $\locof{\rop_1}$.
They are not abstracted away by the
projection $\projectrv{\causalpast{\op}}{\poback{\op}}$.

We show below (*), by induction on the definition of $\hb{\op}$, that any 
edge $\op_1 \ltrel{\hb{\op}} \op_2$ for operations in
$\causaldep{\op}$ implies that 
$\op_1$ must be before $\op_2$ in $\locof{\op}$.
Since $\locof{\op}$ is a strict total order over 
$\causaldep{\op}$, there can be no cycle in $\hb{\op}$.
So $\hist$ cannot contain bad pattern \bpbpcyclichb.

Moreover, by definition of $\speckvs$, 
for any $\mkread{\var}{0}$ operation $\rop$, there can be no 
write operation $\wop$ such that $\wop$ is before $\rop$ in 
$\locof{\op}$ ($\hist$ is \differentiated, so it cannot contain
$\mkwrite{\var}{0}$ operations).
So $\hist$ cannot contain bad pattern \bpbpchangerfi.

(*) We now prove by induction, that any 
edge $\op_1 \ltrel{\hb{\op}} \op_2$ for operations in
$\causaldep{\op}$ implies  that
$\op_1$ is before $\op_2$ in $\locof{\op}$.

Let $\op_1,\op_2 \in \causaldep{\op}$ such that
$\op_1 \ltrel{\hb{\op}} \op_2$. Based on the definition of 
$\hb{\op}$, we have three cases to consider.
\begin{itemize}
\item If $\op_1 \ltpropco \op_2$, then $\op_1$ is before $\op_2$ in 
$\locof{\op}$ because \\
$\projectrv{\causalpast{\op}}{\poback{\op}} \weaker \locof{\op}$.
\item (transitivity) If there exists $\op_3$ such that
  $\op_1 \ltrel{\hb{\op}} \op_3$ and 
  $\op_3 \ltrel{\hb{\op}} \op_2$, we can assume by induction 
  that $\op_1$ is before $\op_3$, and $\op_3$ is before $\op_2$
  is $\locof{\op}$. Since $\locof{\op}$ is a sequence,
  $\op_1$ is before $\op_2$ in  $\locof{\op}$.
\item 
  If there is $\var \in \Var$, and $\dat_1 \neq \dat_2 \in \Nats$, 
  and a read operation $\rop_2$ such that:
  \begin{itemize}
  \item $\op_1   \ltrel{\hb{\op}} \rop_2$,
  \item $\rop_2 \leqpo \op$, 
  \item $\getlabel(\op_1) = \mkwrite{\var}{\dat_1}$, 
  \item $\getlabel(\op_2) = \mkwrite{\var}{\dat_2}$, and 
  \item $\getlabel(\rop_2) = \mkread{\var}{\dat_2}$:
  \end{itemize}
  We know by induction that $\op_1$ is before $\rop_2$ in 
  $\locof{\op}$. Since $\hist$ is differentiated, 
  the only $\mkwrite{\var}{\dat_2}$ operation in $\hist$ is
  $\op_2$, and $\op_2$ must thus be after $\op_1$ in  $\locof{\op}$.
\end{itemize}

This concludes the first part ($\Rightarrow$) of the proof.

$(\Leftarrow)$
Assume that $\hist$ is \wcct{} and does not 
contain the bad patterns \bpbpchangerfi{} and \bpbpcyclichb.
We use the causal order $\propco = (\po \cup \rf)^+$ to show that 
$\hist$ is \scct{} with respect to $\speckvs$.
Since $\hist$ is \wcct, we know that 
$\propco$ is a strict partial order.

Let $\op$ be an operation, and let $\tid$ be the \site{} of $\op$.
Our goal is to show that 
there exists $\loc \in \speckvs$ such that
$\projectrv{\causalpast{\op}}{\poback{\op}} \weaker \loc$.
Said differently, we must sequentialize the operations 
$\causalpast{\op}$, 
while keeping the return values of all read operations done on the 
same \site{} as $\op$, and before $\op$ (in $\poback{\op}$).

We prove this by induction on the size of $\poback{\op}$.
We prove actually the stronger property 
that $\loc$ must also respect the order $\hb{\op}$.

Let $\op'$ be the operation immediately preceding $\op$ in the 
program order (if it exists). We apply the induction hypothesis on $\op'$, 
and obtain a sequence $\seqposet' = \locof{\op'}$ such that 
$\projectrv{\causalpast{\op'}}{\poback{\op'}} \weaker \seqposet'$.
We can apply the induction hypothesis on $\op'$, because acyclicity in
$\hb{\op'} \subseteq \hb{\op}$, so acyclicity in 
$\hb{\op}$ implies acyclicity in $\hb{\op'}$.
By induction hypothesis, we also know that 
$\seqposet'$ respects the order $\hb{\op'}$.

If $\op$ is the first operation on the \site{} (base case of the induction), 
then $\op'$ does not exist, but we define $\seqposet' = \emptyseq$.
In both cases, we have $\seqposet' \in \speckvs$.

We consider three cases:

1) $\op$ is a write operation.
Here, the causal past of $\op$ is the causal past of $\op'$ where 
$\op$ has been added as a maximal operation. 
The reason is that $\propco$ is defined as $(\po \cup \rf)^+$ and 
the read-from relation $\rf$ only relates writes to reads.
Thus, there cannot exist an operation $\op''$ such that $\op'' \ltpropco \op$
and $\op'' \not \ltpropco \op'$. And we can define 
$\loc$ as $\seqposet'$ with $\op$ added at the end.
We obtain that $\loc \in \speckvs$.

2) $\op$ is a read operation $\mkread{\var}{0}$ for some variable 
$\var \in \Var$.
The fact that $\hist$ does not contain bad patterns \bpbprfi{} ensures 
that the causal past $\op$ does not contain write operations on variable 
$\var$. As in the previous case, 
the causal past of $\op$ is the causal past of $\op'$ where 
$\op$ has been added as a maximal operation. We can thus define 
$\loc$ as $\seqposet'$ with $\op$ added at the end.
We obtain that $\loc \in \speckvs$.

3) $\op$ is a read operation $\mkread{\var}{\dat}$ for some $\var \in \Var$
and $\dat \neq 0$.
The fact that $\hist$ does not contain bad patterns 
\bpbprf{} and \bpbpco{} ensures 
that there exists a corresponding $\wop$ operation such that 
$\wop \ltrf \op$ (in the causal past of $\op$), and such there is no 
$\wop_2$ operation with $\wop \ltpropco \wop_2 \ltpropco \op$
and $\getvar{\wop} = \getvar{\wop_2}$.

We consider two subcases:

a) $\wop$ is in the causal past of $\op'$.
By definition of $\hb{\op}$, 
for any $\wop' \in \causaldep{\op}$ with $\wop \neq \wop'$ and
$\getvar{\wop} = \getvar{\wop'}$, we have 
$\wop' \ltrel{\hb{\op}} \wop$.
The last write operation on variable $\var$ in $\seqposet'$ must thus 
be $\wop$.

Moreover, the causal past of $\op$ is the causal past of $\op'$ where 
$\op$ has been added as a maximal operation. We can thus define 
$\loc$ as $\seqposet'$ with $\op$ added at the end.
We obtain that $\loc \in \speckvs$, as $\op$ can read the value written 
by $\wop$.

b) $\wop$ is not in the causal past of $\op'$.
This implies that $\op$ is the only 
$\mkread{\var}{\dat}$ operation in $\poback{\op}$.
(If there was another such read operation $\rop'$, we would have 
$\wop \ltrf \rop' \leqpo \op'$, and $\wop$ would be in the causal past of 
$\op'$.)
Let $\ops''$ be the set of operations contained in the causal past 
of $\op$ (including $\op$), but not contained in the causal past 
of $\op'$. That is, 
$\ops'' = \causaldep{\op} \setminus \causaldep{\op'}$.
With the absence of bad pattern $\bpbpcyclichb$, 
we know that $\wop$ is a maximal (in the $\hb{\op}$ order) 
write operation on variable $\var$ in $\op''$
(there is no write operation $\wop_2$ on variable var such that 
$\wop \ltrel{\hb{\op}} \wop_2 \ltrel{\hb{\op}} \op$,
otherwise, by definition of $\hb{\op}$, 
we would have $\wop_2 \ltrel{\hb{\op}} \wop$ and 
$\hb{\op}$ would be cyclic.

We can thus define a sequence $\seqposet'' \in \speckvs$ 
such that the last write operation 
on variable $\var$ is $\wop$, and such that 
$\seqposet''$ respect the order $\hb{\op}$.
We then define $\loc$ as $\seqposet' \cc \seqposet''$, while setting the 
return values of all reads which are not in \site{} $\tid$ to the last
corresponding write in $\loc$
(these can be freely modified, as they are hidden by
the projection $\projectrv{\causalpast{\op}}{\poback{\op}}$).
We thus obtain that 
$\projectrv{\causalpast{\op}}{\poback{\op}} \weaker \loc$, 
and $\loc$ respects the order $\hb{\op}$.

\end{proof}
 \newpage

 }

\end{document}